\documentclass[11pt]{article}

\usepackage{amsmath,amssymb,amsthm,amsfonts,latexsym,bbm,xspace,graphicx,float,mathtools,braket,dirtytalk}
\usepackage[inkscapeformat=png]{svg}
\usepackage[letterpaper,margin=1in]{geometry}

\usepackage{enumitem}

\usepackage{yhmath}
\usepackage{thmtools}
\usepackage{thm-restate}
\usepackage{subcaption}
\usepackage{yfonts}
\usepackage[linesnumbered,ruled,vlined]{algorithm2e}
\SetKwInput{KwInput}{Input}                
\SetKwInput{KwOutput}{Output}
\SetKwFor{RepTimes}{repeat}{times}{end}

\usepackage{svg}

\usepackage[colorinlistoftodos]{todonotes}

\usepackage[pagebackref,colorlinks,citecolor=blue,,linkcolor=magenta,bookmarks=true,breaklinks]{hyperref}
\usepackage[nameinlink,capitalize]{cleveref}
\newcommand{\trc}{\mathrm{tr}\!}

\renewcommand{\backref}[1]{}

\renewcommand{\backrefalt}[4]{%
\ifcase #1 %
\or
[p.\ #2]%
\else
[pp.\ #2]%
\fi}

\newtheorem{theorem}{Theorem}[section]
\newtheorem{lemma}[theorem]{Lemma}

\newtheorem{proposition}[theorem]{Proposition}
\newtheorem{fact}[theorem]{Fact}
\newtheorem{corollary}[theorem]{Corollary}

\theoremstyle{definition}
\newtheorem{definition}[theorem]{Definition}

\newtheorem{remark}[theorem]{Remark}

\crefformat{footnote}{#2\footnotemark[#1]#3}

\Crefname{enumi}{Fact}{Facts}
\crefname{enumi}{fact}{facts}

\renewcommand{\Pr}{\mathop{\bf Pr\/}}
\newcommand{\E}{\mathop{\bf E\/}}
\newcommand{\Ex}{\mathop{\bf E\/}}

\newcommand{\tr}{\mathrm{tr}}  
\newcommand{\poly}{\mathrm{poly}}

\newcommand{\polylog}{\mathrm{polylog}}

\newcommand{\purestate}{\calS}
\newcommand{\mixedstate}{\calD}

\newcommand{\R}{\mathbb R}

\newcommand{\N}{\mathbb N}

\newcommand{\QAC}{\mathsf{QAC}}

\newcommand{\BQE}{\mathsf{BQE}}
\newcommand{\BQTIME}{\mathsf{BQTIME}}

\newcommand{\PSPACE}{\mathsf{PSPACE}}

\newcommand{\sPSPACE}{\mathsf{statePSPACESIZE}}
\newcommand{\sBQP}{\mathsf{stateBQP}}
\newcommand{\sBQTIME}{\mathsf{stateBQTIME}}
\newcommand{\sBQSUBEXP}{\mathsf{stateBQSUBEXP}}
\newcommand{\puresBQSUBEXP}{\mathsf{pureStateBQSUBEXP}}
\newcommand{\sBQE}{\mathsf{stateBQE}}
\newcommand{\puresBQE}{\mathsf{pureStateBQE}}

\newcommand{\sCircuit}[1]{\mathsf{stateBQSIZE}\left[#1\right]}
\newcommand{\puresCircuit}[1]{\mathsf{pureStateBQSIZE}\!\left[#1\right]}
\newcommand{\puresPSPACE}{\mathsf{pureStatePSPACESIZE}}
\newcommand{\puresBQTIME}{\mathsf{pureStateBQTIME}}
\newcommand{\uPSPACE}{\mathsf{unitaryPSPACESIZE}}
\newcommand{\uBQTIME}{\mathsf{unitaryBQTIME}}
\newcommand{\uBQP}{\mathsf{unitaryBQP}}
\newcommand{\uBQSPACE}{\mathsf{unitaryBQSPACESIZE}}

\newcommand{\eps}{\varepsilon}

\newcommand{\calA}{\mathcal{A}}
\newcommand{\calB}{\mathcal{B}}
\newcommand{\calC}{\mathcal{C}}
\newcommand{\calD}{\mathcal{D}}

\newcommand{\calF}{\mathcal{F}}
\newcommand{\calG}{\mathcal{G}}

\newcommand{\calO}{\mathcal{O}}
\newcommand{\calP}{\mathcal{P}}

\newcommand{\calS}{\mathcal{S}}

\newcommand{\calU}{\mathcal{U}}

\newcommand{\calY}{\mathcal{Y}}

\newcommand{\tracedistance}[1]{d_{\mathrm{tr}}(#1)}
\newcommand{\diamonddistance}[1]{d_{\Diamond}(#1)}

\newcommand{\npack}{\mathcal{N}_\text{pack}}

\newcommand{\abs}[1]{\lvert #1 \rvert}

\newcommand{\ketbra}[2]{\ket{#1}\!\!\bra{#2}}

\renewcommand{\hat}{\widehat}

\newcommand{\ignore}[1]{}

\renewcommand{\Re}{\mathrm{Re}}

\newcommand{\anote}[1]{}
\newcommand{\jnote}[1]{}
\newcommand{\hnote}[1]{}

\newcommand{\ainnote}[1]{}
\newcommand{\jinnote}[1]{}
\newcommand{\hinnote}[1]{}
\newcommand{\tnote}[1]{}
\newcommand{\znote}[1]{}

\newcounter{termcounter}[equation]
\renewcommand{\thetermcounter}{\the\numexpr\value{equation}+1\relax.\roman{termcounter}}
\crefname{term}{term}{terms}
\creflabelformat{term}{#2\textup{(#1)}#3}

\makeatletter
\def\term{\@ifnextchar[\term@optarg\term@noarg}
\def\term@optarg[#1]#2{%
  \textup{#1}%
  \def\@currentlabel{#1}%
  \def\cref@currentlabel{[][2147483647][]#1}%
  \cref@label[term]{#2}}
\def\term@noarg#1{%
  \refstepcounter{termcounter}%
  \textup{\thetermcounter}%
  \cref@label[term]{#1}}
\makeatother

\title{
Quantum State Learning Implies Circuit Lower Bounds
}

\author{Nai-Hui Chia\thanks{\texttt{nc67@rice.edu}. Rice University.}\and Daniel Liang\thanks{\texttt{dl88@rice.edu}. Rice University.} \and Fang Song\thanks{\texttt{fsong@pdx.edu}. Portland State University.}}

\date{\today}

\begin{document}

\maketitle

\begin{abstract}
    We establish connections between state tomography, pseudorandomness, quantum state synthesis, and circuit lower bounds.
    In particular, let $\mathfrak{C}$ be a family of non-uniform quantum circuits of polynomial size and suppose that there exists an algorithm that, given copies of $\ket \psi$, distinguishes whether $\ket \psi$ is produced by $\mathfrak{C}$ or is Haar random, promised one of these is the case.
    For arbitrary fixed constant $c$, we show that if the algorithm uses at most $O\!\left(2^{n^c}\right)$ time and $2^{n^{0.99}}$ samples then $\sBQE \not\subset \mathsf{state}\mathfrak{C}$.
    Here $\mathsf{stateBQE} \coloneqq \mathsf{stateBQTIME}\left[2^{O(n)}\right]$ and $\mathsf{state}\mathfrak{C}$ are state synthesis complexity classes as introduced by Rosenthal and Yuen \cite{rosenthal_et_al:LIPIcs.ITCS.2022.112}, which capture problems with classical inputs but quantum output.
    Note that efficient tomography implies a similarly efficient distinguishing algorithm against Haar random states, even for nearly exponential-time algorithms.
    Because every state produced by a polynomial-size circuit can be learned with $2^{O(n)}$ samples and time, or $O\!\left(n^{\omega(1)}\right)$ samples and $2^{O(n^{\omega(1)})}$ time, we show that even slightly non-trivial quantum state tomography algorithms would lead to new statements about quantum state synthesis.
    Finally, a slight modification of our proof shows that distinguishing algorithms for quantum states can imply circuit lower bounds for decision problems as well.
    This help sheds light on why time-efficient tomography algorithms for non-uniform quantum circuit classes has only had limited and partial progress.

    Our work parallels results by Arunachalam, Grilo, Gur, Oliveira, and Sundaram \cite{arunachalam2022quantum} that revealed a similar connection between quantum learning of Boolean functions and circuit lower bounds for classical circuit classes, but modified for the purposes of state tomography and state synthesis.
    As a result, we establish a conditional pseudorandom state generator, a circuit size hierarchy theorems for non-uniform state synthesis, and connections between state synthesis class separations and decision class separations, which may be of independent interest.
\end{abstract}

\newpage 
\tableofcontents
\vfill
\thispagestyle{empty}
\newpage

\section{Introduction}
\emph{Quantum state tomography} is the task of constructing an accurate classical description of an unknown quantum state given copies of said unknown state, and is the quantum generalization of learning a probability distribution given access to samples from said distribution.
Dating back to the 1950s \cite{fano1957description}, it has become a fundamental problem in quantum information that has numerous applications in verification of quantum experiments and the like \cite{d2003quantum,Banaszek_2013}.

However, for general quantum states this becomes a famously expensive task \cite{o2016efficient,haah2017sample,chen2022tight} and requires $\Omega\!\left(2^n\right)$ samples even for pure states \cite{BRU1999249}.
As such, major attention has been placed on performing efficient tomography for specific classes of quantum states, such as stabilizer states \cite{aaronson43identifying,montanaro-bell-sampling} (and some of their generalizations \cite{lai2022learning,grewal2023efficient,leone2023learning,hangleiter2023bell,Chia2023,grewal2023efficient2}), non-interacting fermion states \cite{aaronson2023efficient}, matrix product states \cite{landoncardinal2010efficient}, and low-degree phase states \cite{arunachalam2022phase}.

However, the class of states produced by low-complexity circuits has remained particularly challenging.
Informally, we define low-complexity circuits to have depth or number of gates that cannot be too large.
For instance, only recently do we have an efficient algorithm for learning the output of states produced by polynomial-size constant-depth unitary circuits of $2$-local gates (i.e., $\mathsf{QNC}^0$), but with the strong restriction that their connectivity must lie on a 2D lattice \cite{huang2024learning}.
Likewise, only recently do we have a \emph{quasi-poly sample} algorithm for learning the $\emph{Choi}$ states produced by constant-depth unitary circuits with both $2$-local gates and $n$-ary Toffoli gates (i.e. Choi states of $\mathsf{QAC}^0$ circuits), with the runtime still exponential in the number of qubits \cite{nadimpalli2024pauli} and again with a strong restriction on the number of ancilla qubits allowed and that there is only one qubit of output.
In contrast to these quantum results, $\mathsf{NC}^0$ is trivially easy to Probably Approximately Correct (PAC) learn and $\mathsf{AC}^0$ has a quasi-poly \emph{time} PAC learning algorithm \cite{linial1993constant}.
See \cite{kearns1994survey,hanneke2016pac} for details on the PAC learning model.

Lower bounds for non-uniform circuit classes have been similarly challenging in the world of computational complexity theory.
The best known circuit lower bounds for explicit functions follow from \cite{kumar:LIPIcs.CCC.2023.18}, which holds for a class of circuits in-between $\mathsf{AC}^0$ and $\mathsf{TC}^0$.
Meanwhile, the breakthrough results of Williams \cite{williams2014acc,williams2018new,murray2020circuit} showed that non-deterministic quasi-poly time (i.e., $\mathsf{NTIME}\left[n^{\log^{O(1)} n}\right]$) cannot be expressed as quasi-poly-size $\mathsf{ACC}^0$ circuits (or even $\mathsf{ACC}^0$ with a bottom layer of threshold gates), where $\mathsf{ACC}^0$ is another class that sits between $\mathsf{AC}^0$ and $\mathsf{TC}^0$.
In both of these cases, $\mathsf{TC}^0$ remains a major roadblock for proving circuit lower bounds.

In this work, we relate the hardness of learning low-depth quantum circuit classes to lower bounds for non-uniform state synthesis.
Specifically, let $\mathfrak{C}$ refer to a class of non-uniform polynomial-size quantum circuits.
We relate the difficulty of giving learning algorithms for states produced by $\mathfrak{C}$ to lower bounds for the set of quantum states that $\mathfrak{C}$ can produce.
This is done through the language of state synthesis complexity classes, such as $\sBQP, \mathsf{statePSPACE}, \mathsf{state}\mathfrak{C}$, etc., which were introduced in a series of recent work \cite{rosenthal_et_al:LIPIcs.ITCS.2022.112,metger2023pspace,bostanci2023unitary,rosenthal2023efficient}.
While the standard decision problem complexity classes capture problems with classical input and classical output (even for quantum models of computations), these state synthesis complexity classes attempt to capture the complexity of problems with classical inputs and quantum outputs.
Despite their differences, these state synthesis complexity classes seem to (and, in many cases, are \emph{designed to}) mirror the usual decision classes, such as in the case of $\mathsf{statePSPACE} = \mathsf{stateQIP}$ \cite{rosenthal_et_al:LIPIcs.ITCS.2022.112, metger2023pspace, rosenthal2023efficient}.
We also generalize these results by relating the complexity of state tomography to circuit lower bounds for decision problems (\cref{thm:main-informal2}), as well as relating the complexity of process tomography to circuit lower bounds for unitary synthesis (\cref{thm:main-unitary-informal}).
We hope that future work will both further explore the relationship between learning and lower bounds against non-uniform circuits, as well instantiate this relationship to give useful and novel lower bounds.

We now informally state our main result about state synthesis.
We show that the existence of a sufficiently efficient (in both time and samples) learner for states produced by a circuit class $\mathfrak{C}$ implies that, for every $k \geq 1$, there exists a sequence of pure states $(\ket{\psi_x})_{x \in \{0, 1\}^*}$ that can be synthesized to arbitrary inverse-exponential accuracy by a uniform quantum algorithm in time $2^{O(n)}$ but not by non-uniform $\mathfrak{C}$ circuits of size at most $O(n^k)$.

\begin{theorem}[Informal statement of \cref{cor:main}]\label{thm:main-informal}
    Let $\mathfrak{C}$ be a class of non-uniform quantum circuits.
    Suppose that for some fixed constant $c$, states produced by $\mathfrak{C}$ could be learned to constant precision (in trace distance) and constant success probability using no more than $O\left(\frac{2^{n^c}}{n^c}\right)$ time and $2^{n^{0.99}}$ many samples.
    Then, for every $k \geq 1$, there exists a state sequence in $\mathsf{stateBQE}$ that cannot be synthesized by non-uniform $\mathfrak{C}$ circuits of size at most $n^k$.\footnote{We define $\mathsf{BQE} = \mathsf{BQTIME}\left[2^{n}\right]$ to be algorithms that run in strictly $2^{O(n)}$ time, rather than $2^{\poly(n)}$, which is how $\mathsf{BQEXP}$ is defined.}
\end{theorem}

Note that any class of pure states can be learned in $2^{\Theta(n)}$ time and samples by running general pure state tomography \cite{franca_et_al:LIPIcs.TQC.2021.7}.
Furthermore, classical shadows \cite{HKP20-classical-shadows} allows one to use $\omega(\poly(n))$ samples instead, at the cost of $2^{\omega(\poly(n))}$ time.
Thus, even slightly non-trivial learning algorithms would imply state synthesis lower bounds for $\mathfrak{C}$.

More carefully, we actually show that non-trivial $\emph{distinguishing}$ of pure states produced by $\mathfrak{C}$ from Haar random states either separates $\sBQE$ from $\mathsf{state}\mathfrak{C}$ or separates $\sBQSUBEXP \coloneqq \bigcap_{\gamma \in (0, 1)} \mathsf{stateBQTIME}\!\left[2^{n^\gamma}\right]$ from all polynomial-size circuits from $1$ and $2$-qubit quantum gates (i.e., $\mathsf{BQSIZE}[n^k]$).
This distinguishing task is generally a much easier task to do than tomography (see \cref{ssec:distinguish-without-learning} for a discussion of this), making the result all the more striking.
In particular, we show that even a $\frac{1}{2^{n^{0.99}}}$ advantage in distinguishing a state from $\mathfrak{C}$ from Haar random, while using at most $2^{n^{0.99}}$ samples and $O\left(2^{n^c}\right)$ time gives interesting and novel lower bounds for state synthesis.

\begin{theorem}[Informal statement of \cref{cor:main2}]\label{thm:main-informal3}
    Let $\mathfrak{C}$ be a class of non-uniform quantum circuits.
    Suppose for some fixed constant $c$ that there exists an algorithm that can distinguish states produced by $\mathfrak{C}$ from Haar random states, with $\frac{1}{2} + \frac{1}{2^{n^{0.99}}}$ success probability, using no more than $O\left(2^{n^c}\right)$ time and $2^{n^{0.99}}$ many samples.
    Then at least one of the following is true:
    \begin{itemize}
        \item For every $k \geq 1$, there exists a state sequence in $\mathsf{stateBQSUBEXP}$ that cannot be synthesized by non-uniform quantum circuits of arbitrary $1$ and $2$ qubit gates of size at most $O(n^k)$,
        \item There exists a state sequence in $\mathsf{stateBQE}$ that cannot be synthesized by polynomial size $\mathfrak{C}$ circuits (i.e., $\mathsf{stateBQE} \not\subset \mathsf{state}\mathfrak{C})$.
    \end{itemize}
\end{theorem}

The only restriction on the circuit class is that (1) all states produced by $\mathfrak{C}$ circuits of size at most $s$ be approximated to $0.49$ in trace distance to states produced by non-uniform circuits of size $\poly(s)$ in the commonly used $\{H, \mathrm{CNOT}, T\}$ gate set and (2) the circuits do not somehow get `weaker' when the input size increases such that if $\ket{\psi}$ is a quantum state that can be efficiently synthesized in terms of input size $n$, then $\ket \psi$ can also be synthesized on inputs of size greater than $n$.
This is a very weak set of restrictions and includes a wide variety of circuit classes, such as circuits of bounded depth (such as $\mathsf{QNC}$), circuits with bounded locality, circuits with non-standard gate sets (such as $\mathsf{QAC}^0_f$), circuits with bounds on the how many times a particular gate can be used (such as the $T$ gate count in the Clifford + $T$ model), etc.
Here $\mathsf{QAC}^0_f$ is defined as constant-depth unitary circuits with both $2$-local gates, $n$-ary Toffoli gates, and the additional \emph{fanout} gate, a unitary that allows parallel \emph{classical} copying of a single qubit to many output qubits.\footnote{Think of the $\mathrm{CNOT}$ gate as \emph{classical} copying of a single qubit to a single output qubit. The fanout gate is multiple $\mathrm{CNOT}$ with the same target qubit being applied as a single action.}
This mimics the power of classical circuits to have \emph{unbounded} fanout, whereas the laws of quantum mechanics do not permit the cloning of quantum data.
In this way, the fanout gate ensures that $\mathsf{AC}^0 \subset \mathsf{QAC}^0_f$, whereas it is unknown if $\mathsf{AC}^0 \subset \mathsf{QAC}^0$.
The addition of the fanout gate even implies $\mathsf{TC}^0 \subset \mathsf{QAC}^0_f$ \cite{hoyer2005fan,takahashi2016collapse}.

We remark that, while both conclusions of \ref{thm:main-informal3} are certainly plausible, showing this is another matter.
And if the intuition that state synthesis complexity classes mirror their decision problem counterparts, formally proving these separations would be highly non-trivial.
Interestingly, when the circuit class in question contains $\mathsf{QAC}^0_f$, we can somewhat formalize this connection and show that distinguishing would imply breakthrough circuit lower bounds for the traditional setting of decision problems.
As such, it illustrates how problems purely about state synthesis and state distinguishing can have breakthrough consequences for the traditional model of complexity theory.

\begin{theorem}[Informal statement of \cref{thm:decision-lower}]\label{thm:main-informal2}

    Let $\mathfrak{C} \supseteq \mathsf{QAC}^0_f$ be a class of non-uniform quantum circuits.
    Suppose that there exists a fixed constant $c$ such that states produced by $\mathfrak{C}$-circuits of depth at most $d+3$ could be distinguished from Haar random states, with $\frac{1}{2} + \frac{1}{2^{n^{0.99}}}$ success probability, using no more than $O\left(2^{n^c}\right)$ time and $2^{n^{0.99}}$ many samples.
    Then at least one of the following is true:
    \begin{itemize}
        \item For every $k \geq 1$, there exists a language in $\mathsf{BQSUBEXP}$ that cannot be decided by non-uniform quantum circuits of arbitrary $1$ and $2$ qubit gates of size at most $O(n^k)$,
        \item There exists a language in $\mathsf{E}$ that cannot be decided to bounded error by $\mathfrak{C}$ circuits of depth at most $d$ (i.e., $\mathsf{E} \not\subset \text{(depth $d$)-}\mathfrak{C}$).\footnote{As is usually the case in complexity theory, we abuse notation and also refer to the set of languages that can be decided by circuits in $\mathfrak{C}$ as $\mathfrak{C}$ as well. We also take $\mathsf{E} \coloneqq \mathsf{DTIME}\!\left[2^{O(n)}\right].$}
    \end{itemize}
\end{theorem}

This follows as a combination of the proof of \cref{thm:main-informal3} along with ideas from \cite{arunachalam2022quantum} and \cite{chia_et_al:LIPIcs.ITCS.2022.47}.
We remark that \cref{thm:main-informal,thm:main-informal3} (as well as \cref{thm:main-unitary-informal}) preserve the fine-grained depth of $\mathfrak{C}$ exactly, unlike \cref{thm:main-informal2}, which only preserves it up to a constant.
As an example, non-trivial distinguishing of states produced by depth $5$ $\mathsf{QAC}^0_f$ circuits would imply $\mathsf{stateBQE} \not\subset \text{(depth 5)-}\mathsf{stateQAC^0_f}$ and $\mathsf{E} \not\subset \text{(depth 2)-}\mathsf{QAC^0_f}$ respectively as the second possible scenario in \cref{thm:main-informal3,thm:main-informal2}.

As a result of \cref{thm:main-informal,thm:main-informal3,thm:main-informal2}, it would seem unlikely to prove formal learning results for states produced by non-uniform polynomial-size quantum circuits given the difficulty that surrounds non-uniform circuit lower bounds.
A more optimistic view would indicate that this illuminates a possible plan of attack for showing state synthesis separation against non-uniform models of computation.

Finally, we show a related result that the ability to distinguish \emph{unitaries} produced by $\mathfrak{C}$ circuits of size $n^k$ using non-adaptive queries implies circuit lower bounds for \emph{unitary synthesis}.
Unitary complexity classes consider an even broader class of problems with quantum outputs \cite{metger2023pspace,bostanci2023unitary}.
\begin{theorem}[Informal statement of \cref{thm:main-unitary}]\label{thm:main-unitary-informal}
    Let $\mathfrak{C}$ be a class of non-uniform quantum circuits.
    Suppose that there exists a fixed constant $c$ such that unitaries produced by $\mathfrak{C}$ could be distinguished from Haar random unitaries with $\frac{1}{2^{n^{0.99}}}$ success probability, using no more than $O\left(2^{n^c}\right)$ time and $2^{n^{0.99}}$ many \emph{non-adaptive} queries.
    Then, at least one of the following is true:
    \begin{itemize}
        \item For every $k \geq 1$, there exists a unitary sequence in $\mathsf{unitaryBQSUBEXP}$ that cannot be synthesized by non-uniform quantum circuits of arbitrary $1$ and $2$ qubit gates of size at most $O(n^k)$,
        \item There exists a unitary sequence in $\mathsf{unitaryBQE}$ that cannot be synthesized by polynomial size $\mathfrak{C}$ circuits (i.e., $\mathsf{unitaryBQE} \not\subset \mathsf{unitary}\mathfrak{C}$).
    \end{itemize}
\end{theorem}
This is a consequence of recent results on non-adaptive pseudorandom unitaries by \cite{metger2024pseudorandom}.
We leave the problem of recovering a result similar to \cref{thm:main-informal2} for unitary learning to future work (see \cref{para:intro-unitary}).

\subsection{Proof Techniques}

\subsubsection{Learning Boolean Functions Implies Circuit Lower Bounds}
We model our proof of \cref{thm:main-informal} heavily after the work of Arunachalam, Grilo, Gur, Oliveira, and Sundaram \cite{arunachalam2022quantum}, which examined the relationship of quantum learning of Boolean functions with separations between $\mathsf{BQE}$ and circuit lower bounds.
This was itself a generalization of a line of works for \emph{classical} learning of Boolean functions\,\cite{FORTNOW2009efficient,harkins2013exact,klivans2013constructing,volkovich2014learning,volkovich16,oliveira2017conspiracies,carbonioliveira_et_al:LIPIcs.APPROX-RANDOM.2018.55}.
Let $\mathfrak{C}[s]$ denote circuits from $\mathfrak{C}$ of size at most $O(s)$.
Informally, \cite[Theorem 3.7]{arunachalam2022quantum} states that sufficiently efficient learning algorithms for boolean functions implemented by $O(n^k)$-size \emph{classical} circuit class $\mathfrak{C}$ implies that $\mathsf{BQE} \not\subset \mathfrak{C}[n^k]$.
The high-level of their proof works as follows:
\begin{enumerate}
    \item Assume there exists a Pseudorandom Generator (PRG) that is secure against sub-exponential-time uniform \emph{quantum} adversaries that can also be computed in $2^{O(n)}$ time.
    \item Create a language $L \in \BQE$ that requires being able to compute the PRG.
    \item Show that a sufficiently efficient learner for $\mathfrak{C}$ implies that no sub-exponential-secure PRG can be computed by $\mathfrak{C}$.
    \item Conclude that $L \not\in \mathfrak{C}$.
\end{enumerate}
Unfortunately, this proof also has the added condition of this PRG existing, and no unconditional PRGs with even close to that security guarantee are known exist.
As is the case in many predecessor works, \cite{arunachalam2022quantum} uses a win-win argument to get around this.
The two cases depend on the relationship between $\PSPACE$ being `a subset of' or `not a subset of' $\mathsf{BQSUBEXP} = \bigcap_{\gamma \in (0, 1)} \mathsf{BQTIME}\left[2^{n^\gamma}\right]$.
\begin{enumerate}
    \item Unconditionally, for every $k \geq 1$, $\PSPACE \not\subset \mathfrak{C}[n^k]$ via a diagonalization argument \cite[Lemma 3.3]{arunachalam2022quantum}.
    \item Appeal to a win-win argument based on the relationship between $\PSPACE$ and $\mathsf{BQSUBEXP}$:
    \begin{itemize}
        \item If $\PSPACE \subset \mathsf{BQSUBEXP}$ then, for every $k \geq 1$, $\mathsf{BQSUBEXP}$ and $\mathsf{BQE}$ can also diagonalize against $\mathfrak{C}[n^k]$.
        \item If $\PSPACE \not \subset \mathsf{BQSUBEXP}$ then a PRG exists such that the above argument holds.
    \end{itemize}
    \item Observe that either way, for every $k \geq 1$, $\mathsf{BQE} \not\subset \mathfrak{C}[n^k]$.
\end{enumerate}

\subsubsection{High-Level Proof for \texorpdfstring{\cref{thm:main-informal}}{Theorem 1.1}}

In our work, we follow a similar high-level path except where (1) $\mathfrak{C}$ now refers to a quantum circuit and (2) the PRG is replaced by a new, inherently quantum, pseudorandom object called a Pseudorandom State (PRS).
This object was introduced by \cite{ji-pseudorandom-states2018} as a quantum analogue for PRGs and involves an ensemble of quantum \emph{pure} states that are indistinguishable from a Haar random state by computationally bounded adversaries.
Our win-win argument is then replaced by looking at the relationship of a variant of $\mathsf{statePSPACE}$, which we call $\sPSPACE$ (see \cref{def:pspacesize}), versus $\mathsf{stateBQSUBEXP}$.

In particular, when $\sPSPACE \not \subset \mathsf{stateBQSUBEXP}$ then there exists a PRS that is secure against sub-exponential time uniform quantum adversaries that can also be synthesized in $2^{O(n)}$ time (i.e., is in $\mathsf{stateBQE}$).
Conversely, a sufficiently efficient learning algorithm for states produced by $\mathfrak{C}$ implies that $\mathfrak{C}$ cannot synthesize \emph{any} PRS that is sub-exponential-time-secure against uniform quantum adversaries.

On the other hand, when $\sPSPACE \subset \mathsf{stateBQSUBEXP}$, we can similarly prove that, for all $k\geq 1$, $\sPSPACE \not\subset \mathsf{state}\mathfrak{C}[n^k]$ to complete both sides of the win-win argument.

Despite the noticeable similarities between our proof and that of \cite{arunachalam2022quantum}, there are a number of differences that arise due to the difference in learning models as well as the change to state synthesis complexity classes.
We highlight the effect of some of the most salient differences, as well as how we alter the proofs with respect to them.

\paragraph{Learning Implies Non-Pseudorandomness}\label{para:learning-implies-distinguishing}
Learning and cryptography are natural antagonists of one another: learning implies lower bounds for constructing cryptographic objects and cryptographic objects imply hardness for learning.
One of the key ingredients in both our proof as well as \cite{arunachalam2022quantum} is utilizing this relationship to show that learning implies that certain complexity classes cannot contain pseudorandom objects.
This was done previously via a quantum natural property, where classical natural properties also appeared in \cite{volkovich2014learning,oliveira2017conspiracies}.
Without going into too much detail, a natural property against classical circuit class $\mathfrak{C}$ is an polynomial time (in the input size of $N = 2^n$) algorithm that acts on the truth table of a Boolean function and accepts with high probability for most random functions and instead rejects with high probability when given a truth table for a function in $\mathfrak{C}$.

In \cref{sec:learning} we completely eschew the notion of natural properties, or any generalization of them.
Instead, we utilize the fact that pseudorandom quantum states (PRS) always involve pure states.
This means that a SWAP test can be performed to determine how close a mixed state $\rho$ is to the pseudorandom state.
By letting $\rho$ be the output of a sub-exponential-time quantum learning algorithm for $\mathsf{state}\mathfrak{C}$ that takes $m$ copies of the state, we can run the SWAP test on an extra copy to verify if our learning algorithm did a good job.
By the definition of a learning algorithm, when given an (initially) unknown state that is in $\mathsf{state}\mathfrak{C}$ the SWAP test will accept with probability bounded away from $\frac{1}{2}$ by a non-negligible amount.
On the other hand, when given a Haar random state, any sub-exponential-sample quantum algorithm will succeed with probability at most a $\exp(-\exp(n))$.
Thus the SWAP test will accept with probability at most $\frac{1}{2} + \exp(-\poly(n))$.
We conclude that $\mathfrak{C}$ cannot synthesize any sub-exponential-secure PRS.

\begin{remark}\label{remark:learning-to-distinguishing}
    We emphasize that using a SWAP test is certainly not a novel idea, but that our analysis is the novel contribution as we could not find an existing analysis in the literature that fit out purposes.
    For instance, a similar analysis was done in \cite[Theorem 14, Theorem 15]{zhao2023learning}.
    However, it is not sufficient for our purposes because the sample complexity of the algorithms they consider (as stated) are limited to $\poly(n)$, whereas we need our reduction to hold even when the learning algorithms use up to $O\!\left(2^{n^{0.99}}\right)$ samples.
    Other benefits of \cref{lem:learning-to-distinguishing} include (1) a much simpler proof, (2) tighter bounds, and (3) a fine-grained reduction from the parameters of the learning algorithm to the distinguishing algorithm.
    Likewise, the non-cloneability of pseudorandom states as shown by \cite[Theorem 2]{ji-pseudorandom-states2018} is lossy in several ways.
    For instance, if the \emph{cloning} algorithm uses $m$ copies, then the distinguisher from \cite[Theorem 2]{ji-pseudorandom-states2018} will use $2m+1$ copies, rather than the $m+1$ in \cref{lem:learning-to-distinguishing}.
    More importantly, to get a cloning algorithm from a learning algorithm, the learning algorithm will potentially (and, in fact, likely will) destroy all $m$ copies via measurement.
    Therefore, to create the $m+1$ approximate copies of $\ket \psi$ necessary to call this an approximate cloning algorithm, we must produce $\ket{\hat{\psi}}^{\otimes m + 1}$ from scratch, where $\ket{\hat{\psi}}$ is the output of the learning algorithm.
    It follows from the Fuchs-van de Graaf inequalities (see \cref{cor:td-multi-copy}) that in order to get $\eps$-approximate cloning, the trace distance between $\ket{\psi}$ and $\ket{\hat{\psi}}$ needs to be at most $\frac{\eps}{\sqrt{m}}$.
    Since $m = O\!\left(2^{n^{0.99}}\right)$ in our circumstances, this is incredibly lossy for all but the smallest values of $\eps$.
\end{remark}

We note that, because the SWAP test only takes $O(n)$ time rather than $2^{O(n)}$ time like a natural property, we can break the security of a PRS with large amounts of samples and time, and also very small accuracy.
\cite[Theorem 3.1]{arunachalam2022quantum} instead has a trade-off between time/samples and accuracy, where taking more time/samples requires the learning algorithm to get more accurate learning and vice-versa.
Our approach also allows for achieving tighter bounds if the PRS constructions can be improved.
See \cref{ssec:improvements-to-prs} for a discussion on this.

\paragraph{PRS from State Synthesis Separations}
Perhaps the most important technical contribution of \cite{arunachalam2022quantum} is proving the existence of sub-exponential-secure PRG against uniform \emph{quantum} adversaries given $\PSPACE \not\subset \mathsf{BQSUBEXP}$.
For the win-win argument to go through, we now need a sub-exponential-secure PRS against uniform quantum computation given $\sPSPACE \not\subset \sBQSUBEXP$.
To do this, we note that there are constructions of PRS given a quantum-secure \emph{pseudorandom function} (PRF) \cite{ji-pseudorandom-states2018,brakerski10.1007/978-3-030-36030-6_10,aaronson2022quantumpseudo,giurgicatiron2023pseudorandomness,jeronimo2024pseudorandom}.
Likewise, there exist quantum-secure PRF constructions when given a quantum-secure PRG \cite{goldreich1984randolli,Zhandry21construct}.
As such, we show in \cref{sec:pseudo} the existence of a sub-exponential-secure PRS against uniform quantum computation given $\PSPACE \not\subset \mathsf{BQSUBEXP}$

\begin{remark}
We need these constructions to retain the sub-exponential security of the underlying PRG to get a sub-exponential-secure PRS, whereas most analysis of these reductions only consider polynomial-time adversaries.
This is more restrictive than it sounds.
For instance, constructions such as random subset states \cite{giurgicatiron2023pseudorandomness,jeronimo2024pseudorandom} are not known to be secure against sub-exponential-time adversaries.
Additionally, because the PRG from \cite{arunachalam2022quantum} requires $2^{\poly(n)}$ time to compute, the usual method of obtaining a quantum-secure PRF from a quantum-secure PRG  \cite{goldreich1984randolli,Zhandry21construct} cannot be used.
Finally, other constructions such as non-binary phase states \cite{ji-pseudorandom-states2018} and random subset states with random phases \cite{aaronson2022quantumpseudo} would not allow us to achieve \cref{thm:main-informal2}, as they require super-constant depth in the construction.
\end{remark}

While this is close to what we want, we still need a PRS conditional on $\sPSPACE \not\subset \sBQSUBEXP$, rather than their decision version counterparts.
To bridge the gap, in \cref{lem:state-implies-decision} we show that $\sPSPACE \not\subset \sBQSUBEXP$ implies $\PSPACE \not\subset \mathsf{BQSUBEXP}$.
This is done by defining a decision problem $L$ based on a particular bit in the description of the separating circuit in $\sPSPACE$.
Therefore, if $L \in \mathsf{BQSUBEXP}$ then, by iterating over all bits in the description, the entire description can be learned.
This is why we need to use $\sPSPACE$, where the circuits always have polynomial size descriptions, as opposed to $\mathsf{statePSPACE}$, where the description could potentially have unbounded size. 
In this way, the entire circuit description can be learned bit-by-bit by a $\mathsf{BQSUBEXP}$ algorithm.
Having the circuit description allows us to then synthesize the state, ultimately implying $\sPSPACE \subset \sBQSUBEXP$, which contradicts the initial assumption.

\paragraph{State Synthesis ``Diagonalization''}\label{para:intro-diagonalization}
On the other side of the win-win argument, the proof of: for every $k \geq 1$, $\PSPACE \not\subset \mathfrak{C}[n^k]$ uses a diagonalization argument.
At a high level, the strategy involves computing what the majority of some set of circuits in $\mathfrak{C}[n^k]$ do for a particular input and then outputting the opposite bit.
By iterating over each element of the truth table, each new input cuts the number of $\mathfrak{C}[n^k]$ circuits that agree with the $\PSPACE$ algorithm by at least a half.
Since there are only $2^{\poly(n)}$ circuits in $\mathfrak{C}[n^k]$ (since they are polynomial size), after about $\poly(n)$ entries in the truth table no circuits will agree.

Unfortunately, such a bit-flipping method does not apply to state synthesis, nor does it work against the notion of non-uniformity used in state synthesis (see \cref{remark:non-uniform-difference}).
For instance, there's no clear definition of what the ``majority'' state is for a given set of states.\footnote{\cite{buhrman_et_al:LIPIcs.ITCS.2023.29} is \emph{not} applicable here, as we are not deciding between two pairs of orthogonal single-qubit quantum states, but rather a collection of $n$-qubit quantum states that has no restrictions to their pair-wise inner products.}
Thus it's not clear how to perturb a state to make it disagree with many other states.
Additionally, while bit strings (even when viewed as computational basis states) are ``well-spaced'', the set of quantum states produced by quantum circuits can be arbitrarily close to each other in trace distance.
Thus even if the state is successfully perturbed, it might not even be appreciably far from the set of old states.

Nevertheless, we still need to show that a $\poly(n)$-space algorithm can find a quantum state that is $\eps$-far from any state produced by a $n^k$-size quantum circuit for some fixed $k \in \N$.
To do so, we utilize a result of Oszmaniec, Kotowski, Horodecki and Hunter-Jones \cite{oszmaniec2022saturation}, which gives a lower bound for a quantity known as the packing number (see \cref{def:packing-number}) for circuits of bounded depth.
Informally, the packing number is the maximal number of states that are $\eps$-far from each other.
Therefore, for sets of states $A \subset B$, if the packing number of $A$ is strictly less than the packing number of $B$ then some state in $B$ is $\eps$-far from all states in $A$.

Utilizing the above lower bound, as well as another counting argument, we establish a circuit-size hierarchy theorem for state synthesis in \cref{ssec:circuit-heirarchy}.
In particular, we prove that circuits of size $s^\prime \coloneqq \poly(s)$ can always create a state far away from anything in produced by size at most $s$.
Finally, in \cref{ssec:diagonalization} we use the fact that trace distance can be approximated in $\poly(n)$-space \cite{watrous2002quantum} to iterate through circuits of $s^\prime$-size and find said state.

\subsubsection{State Tomography Implies Circuit Lower Bounds for Decision Problems}\label{sssec:intro-decision-techniques}
The proof of \cref{thm:main-informal2} much more closely resembles \cite{arunachalam2022quantum} as we are now dealing with decision problems.
As such, the win-win argument now centers around $\mathsf{PSPACE}$ and $\mathsf{BQSUBEXP}$.
In a world where $\mathsf{PSPACE} \subset \mathsf{BQSUBEXP}$, we note that $\mathsf{BQSUBEXP}$ can now diagonalize against general non-uniform quantum circuits \cite{chia_et_al:LIPIcs.ITCS.2022.47}.\footnote{We re-emphasize that the notion of non-uniformity differs between state synthesis and decision problems. Thus this result does \emph{not} imply the work in \cref{para:intro-diagonalization}.}
Conversely, when $\mathsf{PSPACE} \not\subset \mathsf{BQSUBEXP}$ then we use the argument of \cite{arunachalam2022quantum} to construct a language $L$ that is in $\mathsf{BQE}$.
Furthermore, if $L$ was in (depth $d$)-$\mathfrak{C}$ then (depth $d+3$)-$\mathfrak{C}$ could create a PRS, which is contradicted by the assumed existence of a learning algorithm.
This implies $L \not\in \text{(depth $d$)}-\mathfrak{C}$.

\subsubsection{Unitary Synthesis Separations From Unitary Distinguishing}\label{para:intro-unitary}
The proof of \cref{thm:main-unitary-informal} follows the same strategy as \cref{thm:main-informal3}, but with the PRS replaced by a (non-adaptive) Pseudorandom Unitary (PRU) \cite{ji-pseudorandom-states2018,metger2024pseudorandom,chen2024efficient}.
Unfortunately, there are not currently known to be PRU constructions that are secure against adaptive queries.
Indeed, even in very limited settings there has only been recent progress towards constructing these objects \cite{lu2023quantum,metger2024pseudorandom,chen2024efficient}.
Due to the modular nature of our proof, any improvements to PRU constructions will likely immediately strengthen \cref{thm:main-unitary-informal}.
Additionally, we detail a pathway to an analogue of \cref{thm:main-informal2} for unitary learning.
See \cref{remark:unitary-decision} for details.

\subsection{Related Work}

Despite our work largely paralleling the proof techniques of \cite{arunachalam2022quantum}, we note that our result and theirs are not directly comparable.
In their model, the learning is given access to an oracle that implements a Boolean function $f : \{0, 1\}^n \rightarrow \{0, 1\}$.
The learner then needs to output a single-qubit-output quantum process $C_f$ such that \[\Ex_{x \sim \{0, 1\}^n}\left[\tr\left(\ketbra{f(x)}{f(x)} \cdot C_f\left(\ketbra{x}{x}\right)\right)\right]\] is large.
This makes their model a restricted form of \emph{quantum process tomography}, rather than state tomography.
The circuit classes they deal with are also strictly classical.

Additionally, \cite{zhao2023learning} recently used the fact that $\mathsf{TC}^0  \subset \mathsf{QNC}$ to show that learning $\mathsf{QNC}$ circuits is as hard as a variant of Learning with Errors, which can be encoded into a $\mathsf{TC}^0$ circuit \cite{banerjee2021pseudorandom,arunachalam2021hardness}.\footnote{Using $\mathsf{TC}^0 \subset \mathsf{QAC}_f^0 \subset \mathsf{QNC}^1$ \cite{hoyer2005fan,takahashi2016collapse,rosenthal2023query} one can extend this result to hardness for logarithmic depth rather than just poly-logarithmic.}
While the hardness of Learning with Errors is certainly not only plausible but currently widely-believed, it could only be super-polynomial hardness or just not hard at all.
Our result complements their result by giving hardness without the need of any cryptographic assumptions.
Put another way, in a world where an sub-exponential-time attack against Learning with Errors was possible, sufficiently non-trivial learning of $\mathsf{QAC}^0_f$ states could still imply interesting results.

Finally, \cite{chia_et_al:LIPIcs.ITCS.2022.47} make similar kinds of statements to \cref{thm:main-informal2} in the context of the Minimum Circuit Size Problem (\textsf{MCSP}) for quantum circuits.
This is the problem of deciding if the truth table of a function requires a certain circuit size to compute.
Our work is actually most similar to their notion of State Minimum Circuit Size (\textsf{SMCSP}), for which there are not similar results.

\subsection{Discussion and Open Questions}\label{sec:discussion}

\paragraph{Distinguishing Without Learning}\label{ssec:distinguish-without-learning}
\sloppy
One may note that, as per \cref{thm:main-informal3,thm:main-informal2}, that the only assumption really necessary for a lower bound is a \emph{distinguisher} from Haar random and not a learner.
While we show (in \cref{lem:learning-to-distinguishing}) that learning implies distinguishing to eventually achieve \cref{thm:main-informal}, there is the possibility of a distinguisher existing whereas a learning algorithm does not.
For instance, at the time of writing this manuscript, Clifford + $T$ circuits only up to $O(\log n)$ $T$ gates can be efficiently learned \cite{grewal2023efficient,leone2023learning,hangleiter2023bell,Chia2023,grewal2023efficient2}, whereas up to $n$ $T$ gates can be distinguished from Haar random in polynomial time \cite{grewal2023improved}.

We remark that this is the opposite intuition of what's formally stated in \cref{lem:learning-to-distinguishing}, where the distinguisher is slower than the learner by a log factor.
This is entirely a consequence of the log factors in \cref{fact:universal-circuit} needed to test the algorithm via the SWAP test with only black-box access to the learning algorithm.
However, in an informal sense, learning should actually be \emph{harder} than distinguishing.
Thus, it is usually the case that a learner can be turned into a distinguisher that runs as fast as (or oftentimes, even faster than) the learning algorithm by not using the SWAP test approach, but rather the underlying structure of the state it is trying to learn.

\paragraph{Faster PRS Gives Better Circuit Lower Bounds}\label{ssec:improvements-to-prs}
In this work, we show that a sufficiently efficient learning/distinguishing algorithm for states prepared by $\mathsf{state}\mathfrak{C}$ implies new circuit lower bounds for state synthesis (or decision problems resp.).
This bound holds regardless of how efficient the learning algorithm is, as long as it is more efficient than a certain threshold.
However, it would be great if a poly-time (or quasi-poly) learning algorithm implied a stronger lower bound, such as $\mathsf{pureStateBQP}_{\exp} \not \subset \mathsf{pureState}\mathfrak{C}_{\exp}$.

The major technical roadblock is the construction of a PRS against uniform quantum computations that can be computed more efficiently than $O\!\left(2^\kappa\right)$, where $\kappa$ is the key length (see \cref{def:prs}).
In turn, we don't need the PRS to be as secure, since we have assumed a much faster algorithm for distinguishing.
For instance, suppose the underlying PRG (see \cref{lem:conditional-prg}) could be computed in time $O(f(\kappa))$ for $f = o\!\left(2^\kappa\right)$, but was only secure against polynomial time adversaries.
This would imply that the PRS lies in $\mathsf{pureStateBQTIME}\!\left[f\right]_{\exp}$ instead, improving the lower bound.

A second approach would be to directly get a PRS from a state synthesis complexity theoretic assumption such as $\puresPSPACE_\delta \not\subset \sBQTIME\!\left[f\right]_{\delta + \eps}$ without using a PRF or PRG as an intermediary.
By the work of \cite{Kre21-pseudorandom,KQST23-prs}, there is evidence to believe that a PRS may be possible in scenarios where a PRG/PRF (or quantum-secure OWF, more generally) is not.
Both approaches are independently interesting and the proof techniques necessary would likely have many significant consequences.

Note that a unique benefit of using the SWAP test over something akin to a natural property \cite{RAZBOROV199724}, is that the learning-to-distinguishing (see \cref{para:learning-implies-distinguishing}) argument holds even when the parameters of the PRS are significantly altered.
In contrast, the analogous results in \cite{arunachalam2022quantum,chia_et_al:LIPIcs.ITCS.2022.47} require the security of the PRG to have nearly exponential stretch and the security to then be super-polynomial in the stretch (i.e., also nearly exponential).

\paragraph{Better Security Allows For Stronger Adversaries}\label{para:discussion-prs-parameters}

The parameters of the underlying PRG also affect the proof in other ways, such that improvements would affect different parameters in the statements of \cref{thm:main-informal,thm:main-informal3,thm:main-informal2}.
This is documented precisely in \cref{lem:prs-and-learning-to-lowerbound}.
For instance, if the PRG was secure against stronger adversaries, relative to the number of output bits, then the running time of the learning/distinguishing algorithms would increase accordingly.
Quantitatively, if the number of output bits of the PRG is $2^\ell$ (such that the PRS is on $\ell$-qubits) and the security is $f(\ell)$, then the allowed running time of the distinguishing algorithm simply becomes $f(n)$.
As $f(n)$ grows, the requirements to invoke our learning-to-circuit-lowerbound results becomes weaker, making this an appealing open direction.

Finally, we note that in the proof of \cref{lem:prg-to-prf}, one might na\"ively hope to truncate the output of the PRG to artificially decrease the stretch relative to the security.
This would make the security very big relative to the output bits, which, as pointed out above, allows for a learner with larger and larger amounts of time to still imply lower bounds.
However, when the truncation is too great, \cref{lem:prs-and-learning-to-lowerbound} tells us that we will only be proving results about very small circuits. For instance, when the output of the PRG is truncated to a polynomial number of bits, we we will only be able to say things about poly-logarithmic size circuits.

Another way of seeing this is that, because the key length does not change, this will make the state synthesis about producing (or rather, not producing) pseudo-random states on a smaller and smaller number of qubits.
From the point of view of a learning algorithm, the size of the quantum circuit that produces the state will be growing inversely relative to the truncation.
Therefore, a learning algorithm allowed to run in time $2^{f(n)}$ in the number of qubits $n$ will have to learn states produced by $\mathfrak{C}[\poly(f(n))]$, making the problem seemingly no more tractable if $f$ grows super-polynomially.
In fact, at a large enough growth of $f(n)$, $\mathfrak{C}[\poly\left(f(n)\right)]$ may simply be able to produce state sequences that are \emph{statistically} indistinguishable from Haar random, making the problem completely intractable.

\section{Preliminaries}
For functions $f: X \rightarrow Y$ and $f: Y \rightarrow Z$ we define $f \circ g : X \rightarrow Z$ to be the composition of $f$ and $g$.

We define the $p$-norm of vector $v$ to be $\lVert v \rVert_p \coloneqq \sqrt[p]{\sum_i v_i^p}$.
The Schatten $p$-norm of a matrix $A$ is defined to be $\lVert A \rVert_p \coloneqq \sqrt[p]{\trc\left[\left(A^\dagger A\right)^p\right]}$ or the $p$-norm of the singular values of $A$, with the convention that $\lVert \cdot \rVert_\infty$ is the operator norm.

\subsection{Quantum States, Maps, and Measurements}

\begin{definition}
    A quantum state on $n$ qubits is a $2^n \times 2^n$ positive semi-definite Hermitian matrix $\rho$ such that $\trc\left[\rho\right] = 1$.
\end{definition}

We will refer to quantum states $\rho$ such that the rank of $\rho$ is $1$ as \emph{pure states}.
Oftentimes such pure states will simply be written as a $2^n$-dimensional complex unit column vector $\ket{\psi}$ or row vector $\bra{\psi}$ such that their outer-product $\ketbra{\psi}{\psi} = \rho$.
Quantum states that cannot be written in such a manner will often be referred to as \emph{mixed states}.

We will denote the set of quantum states on $n$ qubits as $\mixedstate_n$, while the set of of quantum pure states on $n$ qubits is denoted by $\purestate_n$.
For convenience we will often drop the subscript as the $n$ will be obvious from context.

A trace-preserving completely positive map from $n$-qubits to $n$-qubits is any linear map $\Phi : \mixedstate_n \rightarrow \mixedstate_n$ such that $\text{Id}_m \otimes \Phi : \mixedstate_{m+n} \rightarrow \mixedstate_{m+n}$ for any $m \in \N$, where $\text{Id}_m$ is the identity map on $m$ qubits.

A positive operator-valued measurement on $n$-qubits is a set of positive semi-definite matrices $\{\Pi_i\}$ such that $\sum_i \Pi_i$ is the $2^n \times 2^n$ identity matrix. We define the probability of event $i$ happening when measuring a state $\rho$ with $\{\Pi_i\}$ to be $\trc\left[\Pi_i \rho\right]$.

\subsubsection{Distances Between Quantum States}

The following is the standard notion of distance between quantum states used in this paper.

\begin{definition}[Trace Distance]\label{def:td}
    The trace distance between two states $\rho$ and $\sigma$ is 
    \[
        \tracedistance{\rho, \sigma} \coloneqq \frac{1}{2}\lVert \rho - \sigma \rVert_1.
    \]
\end{definition}

The trace distance is useful because it tells us that distinguishing one copy of $\rho$ from one copy of $\sigma$ can be done with bias at most $\tracedistance{\rho, \sigma}$.
\begin{fact}[\cite{nielsen2002quantum}]\label{fact:td-to-tv}
    For any positive operater-valued measurement $\{\Pi_i\}$ and quantum states $\rho$ and $\sigma$,
    \[
        \tracedistance{\rho, \sigma} \geq \frac{1}{2} \sum_i \abs{\trc\left[\Pi_i \rho\right] - \trc\left[\Pi_i \sigma\right]} = \frac{1}{2} \sum_i \abs{\trc\left[\Pi_i \cdot \left(\rho -  \sigma\right)\right]}.
    \]
\end{fact}

Another useful notion of distance, known as fidelity, will be relevant in \cref{sec:learning}.

\begin{definition}[Fidelity]
    The fidelity between two quantum states $\rho$ and $\sigma$ is
    \[
        \calF(\rho, \sigma) \coloneqq \left(\tr\sqrt{\sqrt{\rho} \sigma \sqrt{\rho}}\right)^2.
    \]
\end{definition}

It is well known that if $\rho$ is a pure state $\ketbra{\psi}{\psi}$ then the fidelity becomes $\calF(\ketbra{\psi}{\psi}, \sigma) = \braket{\psi | \sigma | \psi}$.
Likewise, it is well known that both fidelity and trace distance lie in the interval $[0, 1]$.

To relate the two quantities, we give bounds between trace distance and fidelity when one of the states is a pure state.
They can be seen as a form of Fuchs-van de Graaf inequality \cite{fuchs1999cryptographic,Watrous_2018}.

\begin{fact}[Folklore]\label{fact:td-vs-fidelity}
    Given pure state $\ket{\psi}$ and mixed state $\sigma$ then
    \[
        1 - \calF(\ketbra{\psi}{\psi}, \sigma) \leq \tracedistance{\ketbra{\psi}{\psi}, \sigma} \leq \sqrt{1 - \calF(\ketbra{\psi}{\psi}, \sigma)}.
    \]
    Furthermore, the upper bound is tight when $\sigma$ is also a pure state.
\end{fact}
\begin{proof}
    The second inequality and its condition when $\sigma$ is a pure state follows from the standard Fuchs-van de Graaf inequality.

    For the first inequality, consider the POVM $\{\ketbra{\psi}{\psi}, I - \ketbra{\psi}{\psi}\}$.
    This implies the following lower bound on the trace distance:
    \begin{align*}
        \tracedistance{\ketbra{\psi}{\psi}, \sigma} &
        \geq \frac{1}{2}\left(\bigg|\trc\left[\ketbra{\psi}{\psi} \cdot \left(\ketbra{\psi}{\psi} - \sigma\right) \right]\bigg| + \bigg|\trc\left[\left(I - \ketbra{\psi}{\psi}\right) \cdot \left(\ketbra{\psi}{\psi} - \sigma\right) \right]\bigg| \right) && (\text{\cref{fact:td-to-tv}})\\
        &= \bigg|  1 - \calF(\ketbra{\psi}{\psi}, \sigma) \bigg|\\
        &= 1 - \calF(\ketbra{\psi}{\psi}, \sigma)
    \end{align*}
    where the second line follows from $\trc\left[\rho - \sigma\right] = 0$ for any two quantum states and the third line follows because fidelity is never bigger than $1$.
\end{proof}

We can use the upper bound of \cref{fact:td-vs-fidelity} being tight for pure states to derive the following bounds for the trace distance between multiple copies of two pure states.

\begin{corollary}\label{cor:td-multi-copy}
    For pure states $\ket{\psi}$ and $\ket{\phi}$ and $m \in \N$,
    \[
        \tracedistance{\ket{\psi}^{\otimes m}, \ket{\phi}^{\otimes m}} \leq \sqrt{m} \cdot \tracedistance{\ket{\psi}, \ket{\phi}}.
    \]
\end{corollary}
\begin{proof}
    \begin{align*}
        \tracedistance{\ket{\psi}^{\otimes m}, \ket{\phi}^{\otimes m}}
        &= \sqrt{1-\calF(\ket{\psi}^{\otimes m}, \ket{\phi}^{\otimes m})} && (\text{\cref{fact:td-vs-fidelity}})\\
        &= \sqrt{1-\calF(\ket{\psi}, \ket{\phi})^m}\\
        &= \sqrt{1-\left(1- \tracedistance{\ket{\psi}, \ket{\phi}}^2\right)^m} && (\text{\cref{fact:td-vs-fidelity}})\\
        & \leq \sqrt{m} \cdot \tracedistance{\ket{\psi}, \ket{\phi}},
    \end{align*}
    where the second line holds from the multiplicativity of fidelity with respect to the tensor product.
\end{proof}

We can also upper bound the trace distance of pure states to their distance as vectors.
This largely follows because their distance as vectors cares about global phase, whereas trace distance knows that global phase does not matter in quantum mechanics.

\begin{fact}\label{fact:operator-to-td}
    For pure states $\ket \psi$ and $\ket \phi$,
    \[
        \tracedistance{\ket \psi, \ket \phi} \leq \lVert \ket \psi - \ket \phi \rVert_2.
    \]
\end{fact}
\begin{proof}
    \begin{align*}
        \tracedistance{\ket \psi, \ket \phi} &= \sqrt{1 - \abs{\braket{\psi | \phi}}^2} && (\text{\cref{fact:td-vs-fidelity}})\\
        &= \sqrt{(1 + \abs{\braket{\psi | \phi}})(1 - \abs{\braket{\psi | \phi}})}\\
        &\leq \sqrt{2 - 2 \abs{\braket{\psi|\phi}}}\\
        &\leq \sqrt{\braket{\psi|\psi} + \braket{\phi|\phi} - 2 \Re\left[\braket{\psi|\phi}\right]}\\
        &= \lVert \ket \psi - \ket \phi \rVert_2
    \end{align*}
\end{proof}

Because $\lVert U \ket \psi - V \ket \psi \rVert_2 \leq \lVert U - V \rVert_\infty$, acting on a state with unitary $V$ instead of unitary $U$ only affects the resulting quantum state by at most $\lVert U - V \rVert_\infty$.

\subsection{Quantum Circuits}

We largely follow the notation of \cite{bostanci2023unitary} and \cite{rosenthal2023efficient} both here and in \cref{ssec:state-synthesis}.
Note that we will regularly abuse notation, such that for functions acting on the natural numbers we will actually be implicitly referring to the output of that function on $n$, the input size. For instance, given $f: \mathbb{N} \rightarrow [0, 1]$, we will generally shorthand $f(n)$ to just $f$.

A \emph{unitary quantum circuit} is any circuit that obeys the following pattern: (1) initializing ancilla qubits, (2) applying gates from $\{H, \mathrm{CNOT}, T\}$, and (3) optional tracing out of ancilla qubits at the end.\footnote{$\{H, \mathrm{CNOT}, T\}$ is well-known to be a universal quantum gate set \cite[Chapter 4.5]{nielsen2002quantum},}

A \emph{general quantum circuit} adds the ability to perform \emph{intermediate} single-qubit non-unitary gates: (1) prepare an auxiliary qubit in the $\ket{0}$ state (2) trace out a qubit and (3) measure a qubit in the computational basis.
These actions can be taken at any point in the computational process.
We will often denote a general quantum circuit $C$ acting on input state $\rho$ to be $C(\rho)$ for brevity.

We note that the actual gate set is not important, as long as it has only $O(1)$ distinct gates, each gate has algebraic entries, and is universal for quantum computation.
The $O(1)$ distinct gates is necessary in \cref{lem:num-circuits,cor:npack-upper} to bound the number of quantum circuits of polynomial size.
The algebraic entries allows the distance between states to be computed in polynomial space (see \cref{lem:pspace-compute-td}).
Finally, the universality part allows us to use the Solovay-Kitaev algorithm \cite{dawson2005solovaykitaev}.

\begin{lemma}[Solovay-Kitaev algorithm]\label{lem:solovay-kitaev}
    Given a universal gateset $\calG$ containing its own inverses, then any $1$ or $2$-qubit unitary can be approximated by a unitary $\hat{U}$ via a finite sequence of gates from $\calG$ with length $O\left(\log^{3.97}\frac{1}{\eps}\right)$ that can be found in time $O\left(\log^{2.71}\frac{1}{\eps}\right)$ such that $\lVert U - \hat{U} \rVert_\infty \leq \eps$.
\end{lemma}

Observe that, since trace distance is contractive under trace-preserving completely positive maps (which adding an ancilla qubit, measuring a qubit, and tracing out a qubit fall under), these non-unitary actions do not contribute to any error in synthesizing, regardless of the unitary gate set.

\begin{definition}[Quantum Circuit Size and Space]
    We say that the \emph{size} of a general (resp. unitary) quantum circuit $C$ is the number of gates (both unitary and non-unitary) used in $C$.
    We say that the \emph{space} of a general (resp. unitary) quantum circuit $C$ is the maximum number of qubits used at any point in the circuit.
    Finally we say that the \emph{depth} of a general (resp. unitary) quantum circuit $C$ is the maximum number of layers of gates that are needed.
\end{definition}

Note that a size $k$ circuit on $n$ qubits of input uses space at most $n+k$ with any $2$-local gate set, such as the $\{H, \mathrm{CNOT}, T\}$ gate set that we will default to.
Additionally, a depth $d$ space $s$ circuit has size at most $s\cdot d$

\begin{definition}[Uniform Circuit Families]\label{def:uniform-circuit}
For $t: \mathbb{N} \rightarrow \mathbb{R}^+$, a family of quantum circuits $\left(C_x\right)_{x \in \{0, 1\}^*}$ is called \emph{$t$-time-uniform} if $\left(C_x\right)_{x \in \{0, 1\}^*}$ is size at most $\poly\left(\abs{x}, t(\abs{x})\right)$ and there exists a classical Turing machine that on input $x \in \{0, 1\}^*$ outputs the description of $C_x$ in time at most $O\left(t(\abs{x}\right)$.
Similarly, for $s: \mathbb{N} \rightarrow \mathbb{R}^+$ a family of quantum circuits $\left(C_x\right)_{x \in \{0, 1\}^*}$ is called \emph{$s$-space-uniform} if $\left(C_x\right)_{x \in \{0, 1\}^*}$ uses space at most $\poly(\abs{x}, s(\abs{x}))$ and there exists a classical Turing machine that on inputs $x \in \{0, 1\}^*$ and $i \in \N$ outputs the $i$-th bit of the description of $C_x$ in space at most $O(s(\abs{x}))$.
\end{definition}

We will oftentimes refer to uniform \emph{general} quantum circuits simply as quantum algorithms.
We will also use the circuit sequence $(C_n)_{n \in \N}$ to be composed of $C_{1^n}$ from uniform quantum circuits $(C_x)_{x \in \{0, 1\}^*}$ as a shorthand for uniform circuit families that only depend on the length of the input.
These too will often be referred to as simply quantum algorithms.

As with \cite{arunachalam2022quantum}, it will be necessary throughout this work that given a valid \emph{description}, $\mathrm{desc}(U)$, of a \emph{unitary} circuit $U$, there is an efficient procedure to \emph{apply} said unitary circuit in time polynomial in the size of the input and the size of the description within the gateset $\{H, \mathrm{CNOT}, T\}$.

\begin{fact}[{\cite[Theorem 3]{green2009efficient}}]\label{fact:universal-circuit}
    Fix $n \in \N$ and $s: \N \rightarrow \R^+$ and let $\textsc{desc}(C) \in \{0, 1\}^m$ refer to the description of a \emph{unitary} quantum circuit $U$ on $n$ qubits of size $s(n)$. There exists an $(n + m)$-qubit $O\left((n + s) \log \left(n + s \right) \right)$-time-uniform \emph{unitary} quantum circuit $\calU$ such that
    \[
    \calU\left(\ket{x} \otimes \ket{\textsc{desc}(U)}\right) = \left(U \ket{x}\right) \otimes  \ket{\textsc{desc}(U)}
    \]
    for all $x \in \{0, 1\}^n$.
\end{fact}

\subsection{State Synthesis Complexity Classes}\label{ssec:state-synthesis}

    We now define the state synthesis version of some complexity classes.
    Defined previously in \cite{rosenthal_et_al:LIPIcs.ITCS.2022.112,metger2023pspace,bostanci2023unitary,rosenthal2023efficient}, they  capture the complexity of problems with quantum output.
    Informally, let $\mathsf{A}$ be a decision class associated with some computational models and let $\mathfrak{C}_{\mathsf{A}}$ be the set of circuit sequences $(C_x)_{x \in \{0, 1\}^*}$ (uniform or otherwise) associated with $\mathsf{A}$ with $\poly(\abs{x})$-qubits of output.
    Then $\mathsf{stateA}_\delta$ is simply the set of state sequences $(\rho_x)_{x \in \{0, 1\}^*}$ such that there exists a corresponding $(C_x)_{x \in \{0, 1\}^*} \in \mathfrak{C}_{\mathsf{A}}$ that outputs a state that is $\delta$-close to each $\rho_x$ in trace distance.
    As an example, $\mathsf{stateBQE}_\delta$ (see \cref{def:stateBQE}), is, informally, the set of state sequences that can be synthesized to trace distance at most $\delta$ by $2^{O(n)}$-time-uniform general quantum circuits.
    To make sure that these classes are properly comparable, we will restrict the number of qubits of each $\rho_x$ to be $\poly(\abs{x})$.\footnote{Some works, such as \cite{rosenthal2023efficient} even restrict the final number of qubits to simply be $\abs{x}$.}

    We will generally want to work with arbitrary inverse-polynomial or arbitrary inverse-exponential trace distance.
    As such, for a complexity class $\mathsf{stateA}_\delta$, define $\mathsf{stateA}$ and $\mathsf{stateA}_{\exp}$ to be \[
        \mathsf{stateA} \coloneqq \bigcap_{q} \mathsf{stateA}_{1/q}
    \]
    and
    \[
    \mathsf{stateA}_{\exp} \coloneqq \bigcap_{q} \mathsf{stateA}_{\exp(-q)}
    \]
    where the union is over all polynomials $q: \N \rightarrow \R$.
    Furthermore, for an arbitrary state synthesis complexity class over mixed states $\mathsf{stateA}_\delta$, we define $\mathsf{pureStatA}_\delta \subset  \mathsf{stateA}_\delta$ to be the subset of $\mathsf{stateA}_\delta$ with state sequences consisting only of pure states $(\ketbra{\psi_x}{\psi_x})_{x \in \{0, 1\}^*}$.
    We will generally be dealing with pure state synthesis classes throughout this work.
    We will also abuse notation and often write down such pure state sequences simply as $(\ket{\psi_x})_{x \in \{0, 1\}^*}$.
    
    We now state some (largely trivial) facts about state synthesis complexity classes and their variations, to aid in intuition.
    
    \begin{fact}\label{fact:generic-synthesis-facts}
        For arbitrary state synthesis complexity classes $\mathsf{stateA}_\delta$ and $\mathsf{stateB}_{\eps}$:
        \begin{enumerate}[label=\rm{(\roman*)}]
            \item\label{fact:1} $\mathsf{stateA}_\delta \subseteq \mathsf{stateA}_{\delta^\prime}$ for $\delta \leq \delta^\prime$.
            \item\label{fact:2} $\mathsf{stateA}_\delta \subset \mathsf{stateB}_\eps \Rightarrow \mathsf{pureStateA}_\delta \subset \mathsf{pureStateB}_\eps$.
            \item\label{fact:3} $\mathsf{stateA}_\delta \subset \mathsf{stateB}_\eps \Rightarrow \mathsf{stateA}_{\delta + \gamma} \subset \mathsf{stateB}_{\eps + \gamma}$.
        \end{enumerate}
    \end{fact}

\subsubsection{Uniform Computation}

\begin{definition}[{$\sBQTIME\left[f\right]_\delta$}, {$\mathsf{stateBQSPACE}\left[f\right]_\delta$}]\label{def:uniform-state-synthesis}
    Let $\delta : \mathbb{N} \rightarrow [0, 1]$ and $f : \mathbb{N} \rightarrow \mathbb{R}^+$ be functions. Then $\sBQTIME\left[f\right]_\delta$ (resp. $\mathsf{stateBQSPACE}\left[f\right]_\delta$) is the class of all sequences of density matrices $(\rho_x)_{x \in \{0, 1\}^*}$ such that each $\rho_x$ is a state on $\poly(\abs{x})$ qubits, and there exists an $f$-time-uniform (resp. $f$-space-uniform) family of general quantum circuits $\left(C_x\right)_{x \in \{0, 1\}^*}$ such that for all sufficiently large input size $\abs{x}$, the circuit $C_{x}$ takes no inputs and outputs a density matrix $\sigma_x$ such that $\tracedistance{\rho_x, \sigma_x} \leq \delta.$
\end{definition}

We define the following state synthesis complexity classes, which are the analogues of $\mathsf{BQE}$ and $\PSPACE$.

\begin{definition}[$\sBQE_\delta$, $\mathsf{statePSPACE}_\delta$]\label{def:stateBQE}
    
    \begin{align*}
        \sBQE_\delta \coloneqq \bigcup_{c \geq 0} \sBQTIME\left[2^{c\cdot n}\right]_\delta
    \,\,\,\text{  and  }\,\,\,
    \mathsf{statePSPACE}_\delta \coloneqq \bigcup_p \mathsf{stateBQSPACE}\left[p\right]_{\delta}
    \end{align*}
    where the union for $\mathsf{PSPACE}_\delta$ is over all polynomials $p : \mathbb{N} \rightarrow \mathbb{R}^+$.
\end{definition}

We can likewise define $\mathsf{stateBQSUBEXP}_\delta \coloneqq \bigcap_{\gamma \in (0, 1)}\sBQTIME\left[2^{n^{\gamma}}\right]_\delta$ and $\mathsf{stateBQP}_\delta \coloneqq \bigcup_{p} \sBQTIME\left[p\right]_\delta$ for polynomials $p : \N \rightarrow \R^+$.\footnote{Note that functions that grow as $2^{o(n)}$, such as $2^{\sqrt{n}}$ are also generally referred to as \emph{sub-exponential}. To avoid confusion, we will use the term `sub-exponential' only to refer to such $2^{o(n)}$ growths, and instead refer to algorithms in $\mathsf{BQSUBEXP}$ as simply `$\mathsf{BQSUBEXP}$ algorithms'.}

It is worth commenting on the choice of gate set and how it affects (or does not affect) $\sBQTIME\left[f\right]_\delta$.
While we use the $\{H, \mathrm{CNOT}, T\}$ gate set, if we had some other universal gate set then we could apply the \nameref{lem:solovay-kitaev} to approximate each of these gates.
Because the size is at most $\poly(n, f(n))$, we need to approximate each gate to accuracy $\frac{\delta}{\poly(n, f(n))}$, which increases the runtime by a multiplicative factor of $O\left(\log^{2.71}\left(\frac{n \cdot f(n)}{\delta}\right)\right)$.
Thus $\mathsf{stateBQP}$, $\mathsf{stateBQP}_{\exp}$, $\sBQE$, and $\sBQE_{\exp}$ are not affected by which universal gate set is used.
For comparison, we note the same is not true for a version of $\mathsf{stateBQP}$ with arbitrary doubly-exponentially-small error.

\subsubsection{Non-Uniform Computation}

We now introduce \emph{non-uniform} models of state synthesis.
Unlike uniform classes, there does not necessarily need to be an algorithm that finds the correct circuit.
It just needs to exist.

The following will be the non-uniform class used the most often in the proofs.

\begin{definition}[$\sCircuit{s}_\delta$]
    Let $\delta : \mathbb{N} \rightarrow [0, 1]$ and $s : \N \rightarrow \mathbb{R}^+$.
    Then $\sCircuit{s}_\delta$ is the class of all sequences of quantum states $(\rho_x)_{x \in \{0, 1\}^*}$ such that each $\rho_x$ is a state on $\poly(\abs{x})$ qubits, and there exists a family of \emph{unitary} quantum circuits $\left(C_x\right)_{x \in \{0, 1\}^*}$, taking in $\poly(\abs{x})$ qubits as input and of size at most $O(s)$, such that for all $x \in \{0, 1\}^*$ with sufficiently large length $\abs{x}$, the circuit $C_{x}$ acting on the all-zeros state outputs a state $\hat{\rho}_x$ satisfying
    \[
        \tracedistance{\rho_x, \hat{\rho}_x} \leq \delta.
    \]
\end{definition}

\begin{remark}\label{remark:non-uniform-difference}
    As pointed out in \cite[Section 3.2]{bostanci2023unitary}, the notion of non-uniformity here is different than that of decision problems.
    For state synthesis, there is a different circuit allowed for each input $x \in \{0, 1\}^*$.
    For decision problems, there is just one circuit for all $x \in \{0, 1\}^n$ and the input becomes $\ket{x}{x}$, rather than the all-zeros state.
    That is, there is a non-uniform sequence $(C_n)_{n \in \N}$ and for $x \in \{0, 1\}^n$, the output $\rho_x$ is $C_n(\ketbra{x}{x})$.
    This restriction is necessary, otherwise even the simplest of non-uniform circuits could decide all languages.
\end{remark}

A very important feature for us will be that circuits of bounded size will bounded description lengths.

\begin{lemma}[Folklore]\label{lem:num-circuits}
    For $s \geq n$, the number of unitary quantum circuits of size $s$ with $n$-qubits of output can be described using at most $s \cdot \left(3\log_2 s + 4\right)$ bits.
\end{lemma}
\begin{proof}
    Recall that circuit-size bounds circuit-space.
    Thus we can assume there are never more than $2s$ qubits (when including ancilla) at any point in the computation.
    As such, for every gate in a quantum circuit, we can describe its location by which gate it is among $\{H, T, \mathrm{CNOT}\}$ (and tracing out when at the end of the circuit), the qubit(s) it acts on, and what layer). Since the depth is also at most $s$, this requires at most $2$, $2\log_2 (2s)$, and $\log_2 s$ bits respectively. The whole circuit of size $s$ can therefore be written using $s \cdot (2 + 2\log_2 (2s) + \log_2 s) = s \cdot \left(3\log_2 s + 4\right)$ bits.\footnote{In an alternative local gate set, the $+4$ would be replaced by $2 + \lceil \log_2 G \rceil$, where $G$ is the number of gates in the gate set. Our results should hold for any $G = O(1)$.}
\end{proof}

Now let us generally define a set of unitary quantum circuits $\mathfrak{C}[s]$ to be the circuits in some family of size at most $O(s)$ and let
\[
    \mathfrak{C} \coloneqq \bigcup_p \mathfrak{C}[p]
\]
for all polynomials $p : \N \rightarrow \R^+$, such that $\mathfrak{C}$ be the set of $\mathfrak{C}$-circuits of arbitrary polynomial size.
For clarity, we will sometimes write $\mathfrak{C}\left[\poly(n)\right]$ instead.

We will give a special name to the class of all circuits produced by non-uniform polynomial-size quantum circuits.

\begin{definition}[$\mathsf{stateBQP/poly}_\delta$]
\[
\mathsf{stateBQP/poly}_\delta \coloneqq \bigcup_p \sCircuit{p}_\delta
\]
where the union is over all polynomials $p : \N \rightarrow \mathbb{R}^+$.
\end{definition}

As with $\mathsf{stateBQP}$ and $\mathsf{stateBQP}_{\exp}$, $\mathsf{stateBQP/poly}$ and $\mathsf{stateBQP/poly}_{\exp}$ do not depend on the choice of universal gate set.

The following will be one of the most important state synthesis classes in our proof.
We emphasize that, despite being in the non-uniform section, it is a \emph{uniform} state synthesis class due to $\mathsf{statePSPACE}$ being uniform.

\begin{definition}[$\mathsf{statePSPACESIZE}_\delta$]\label{def:pspacesize}
\[
\mathsf{statePSPACESIZE}_\delta \coloneqq \mathsf{statePSPACE}_\delta \bigcap \mathsf{stateBQP/poly}_\delta.
\]
\end{definition}

By taking the intersection with $\mathsf{stateBQP/poly}_\delta$, we can ensure that it is a unitary circuit with a polynomial size description.
In \cref{lem:state-implies-decision}, this will allow us to efficiently find the whole description of circuits in $\mathsf{statePSPACESIZE}$ as well as efficiently apply the circuit given its description (see \cref{fact:universal-circuit}).

\Cref{thm:main-informal2,thm:main-informal3} require that $\mathfrak{C}$ be closed under restriction of qubits.
That is, if $\ket{\psi_x}$ can be synthesized for inputs of length $\abs{x} = n$, then it can also be synthesized for inputs of length $\abs{x} > n$.
This is a very natural restriction that prevents the circuit class from getting weaker when the input size increases.

\Cref{thm:main-informal} additionally requires for circuit class $\mathfrak{C}[s]$ that there exists some fixed constant $k$ such that $\mathsf{pureState}\mathfrak{C[s]}_0 \in \puresCircuit{s^k}_{\delta}$ for all $s = \poly(n)$ and $\delta \in (0. 0.49)$.
This assumption is only needed to relate it to the results in \cref{ssec:diagonalization} and applies to a wide variety of circuit models. 
To instantiate this claim, we prove that it holds for two popular depth-bounded circuit classes: $\mathsf{QNC}$ and $\mathsf{QAC}_f$ as introduced by \cite{hoyer2005fan,moore1999quantum}.

\begin{definition}[{$\mathsf{QNC}^k[s]$}]
    For $s : \N \rightarrow \mathbb{R}^+$ and $k \in \N$, let $\mathsf{QNC}^k[s]$ be the set of all $s$-size unitary circuits consisting entirely of arbitrary $1$ and $2$-qubit gates with depth at most $O(\log^k n)$.
    Then let $\mathsf{pureStateQNC}^k[s]_\delta$ be the set of state sequences that can be synthesized by a sequence of $\mathsf{QNC}^k[s]$ circuits to trace distance at most $\delta$.
\end{definition}

\begin{definition}[{$\mathsf{QAC}^k_f[s]$}]
    For $s : \N \rightarrow \mathbb{R}^+$ and $k \in \N$, let $\mathsf{QAC}^k_f[s]$ be the set of all $s$-size unitary circuits consisting entirely of arbitrary $1$ and $2$-qubit gates, arbitrary $\mathrm{C^nNOT}$ gates\footnote{$\mathrm{C^1NOT}$ is simply the $\mathrm{CNOT}$ gate and $\mathrm{C^2NOT} = \mathrm{CCNOT}$ is the Toffoli gate.}
    and arbitrary fanout gates
    \[
        \ket{b,x} \mapsto \ket{b, x \oplus b^m} \text{ for } b \in \{0, 1\}, x \in \{0, 1\}^m
    \]
    with depth at most $O(\log^k n)$.
    Then let $\mathsf{pureStateQAC}^k_f[s]_\delta$ be the set of state sequences that can be synthesized by a sequence of $\mathsf{QAC}^k_f[s]$ circuits to trace distance at most $\delta$.
\end{definition}

Note that the fanout gate only copies \emph{classical} data, as full quantum fanout is not allowed by the no-cloning theorem.
Without it, it would not be clear that $\QAC^0$ (i.e., without fanout) contains $\mathsf{AC}^0$.
This is unlike $\mathsf{QNC}^0$ trivially containing $\mathsf{NC}^0$, due to the bounded fanin gates limiting the effect of unbounded fanout.
Adding this fanout gate to $\mathsf{QAC}^0$ has the knock-on effect that $\mathsf{AC}^0 \subsetneq \mathsf{TC}^0 \subset \mathsf{QAC}^0_f$ \cite{linial1993constant,hoyer2005fan,takahashi2016collapse}.\footnote{$\mathsf{TC}^k$ is the set of $\poly(n)$-size classical circuits of unbounded fan-in $\mathrm{AND}$, $\mathrm{OR}$, $\mathrm{NOT}$, and threshold gates of $O(\log^k n)$ depth.}

Rosenthal \cite{rosenthal2023query} showed how to simulate $\QAC_f$ circuits with $\mathsf{QNC}$ circuits.

\begin{lemma}[{\cite[Lemma A.1]{rosenthal2023query}}]\label{lem:qac-to-qnc}
For all space-$s$, depth-$d$ $\QAC_f$ circuits $U$, there exists a space-$O(s)$, depth-$O(d \log s)$, size-$O(ds)$ $\mathsf{QNC}$ circuit $C$ such that \[
    C\left(I \otimes \ket{0 \dots 0}\right) = U \otimes \ket{0 \dots 0}.
\]
\end{lemma}

\begin{corollary}\label{cor:qac-to-qnc-to-mpoly}
    For arbitrary $\alpha \geq 1$, $k \in \N$, and polynomial $q : \N \rightarrow \R^+$,
    \[
        \mathsf{pureStateQAC}_f^k[n^\alpha]_0 \subset \mathsf{pureStateQNC}^{k+1}\left[n^\alpha \cdot \log^k n\right]_0
        \]
    and 
    \[
    \mathsf{pureStateQNC}^{k}\left[n^\alpha\right]_0 \subset \puresCircuit{n^{\alpha} q^4 \log^4 n}_{\exp(-q)}.
    \]
\end{corollary}
\begin{proof}
    The first statement comes directly from \cref{lem:qac-to-qnc}.
    
    Two show the second statement, we note that the \nameref{lem:solovay-kitaev} allows us to approximate arbitrary $1$ and $2$ qubit gates to $\exp(-q)/n^\alpha = \exp\left(-O(q \log n)\right)$ accuracy in operator norm using at most $q^4 \log^4 n$ gates.
    Let $C$ be an arbitrary circuit in $\mathsf{QNC}^{k}[n^\alpha]$.
    By applying Solovay-Kitaev to each gate in $C$ and then taking the triangle inequality over all $n^\alpha$ gates we can construct a unitary approximation $\hat{C}$ such that $\lVert C - \hat{C} \rVert_\infty \leq \exp(-q)$ using at most $O\left(n^{\alpha} q^4 \log^4 n\right)$ gates.
    Due to the \cref{fact:operator-to-td} and nature of the operator norm, \[\tracedistance{C \ket{\psi}, \hat{C}\ket{\psi}} \leq \lVert C \ket \psi - \hat{C} \ket \psi \rVert_2 \leq \exp(-q)\] for all $\ket \psi$.
    Therefore $\mathsf{pureStateQNC}^{k+1} \subset \puresCircuit{n^{\alpha} q^4 \log^4 n}_{\exp(-q)}$.
\end{proof}

By \cref{cor:qac-to-qnc-to-mpoly}, setting $q = \log 100$ gives us \[\mathsf{pureStateQNC}^k[n^\alpha]_{0} \subset \puresCircuit{n^{\alpha + \eps}}_{0.01}\] and  \[\mathsf{pureStateQAC}_f^k[n^\alpha]_{0} \subset \puresCircuit{n^{\alpha + \eps}}_{0.01}\]
for arbitrary $\eps > 0$.

\subsection{Decision Problem Complexity Classes}
We can also define the more traditional complexity classes using this language.
Given a language $L \in \{0, 1\}^*$, let the one-qubit state $\ket{x \in L}$ be defined as \[\ket{x \in L} \coloneqq \begin{cases} \ket{1} & x \in L \\ \ket{0} & x \not\in L \end{cases}.\]
For a \emph{uniform} state synthesis class $\mathsf{stateA}_\delta$, we take decision class $\mathsf{A}$ to be the set of languages $L \subseteq \{0, 1\}^*$ where there exists a state sequence $(\rho_x) \in \mathsf{stateA}_0$ such that for all $x \in \{0, 1\}^*$, \[\tr\left[\ketbra{x \in L}{x \in L} \cdot \rho_x \right] \geq \frac{2}{3}.\]
For a circuit class $\mathfrak{C}[s]$, we take the \emph{non-uniform} decision class $\mathfrak{C}$ to be the set of languages $L \subseteq \{0, 1\}^*$ where there exists a $\mathfrak{C}[\poly(n)]$-circuit sequence $(C_n)_{n \in \N}$ such that for all $x \in \{0, 1\}^*$, \[\tr\left[\ketbra{x \in L}{x \in L} \cdot C_n(\ketbra{x}{x}) \right] \geq \frac{2}{3}.\]

We now define explicitly define some decision classes that will be used in our proofs, namely in \cref{sec:final-proof}.

\begin{definition}[$\PSPACE$]
    We define $\PSPACE$ to be the set of languages that can be decided by a deterministic Turing machine that uses at most $\poly(n)$ space.\footnote{Note that $\PSPACE = \mathsf{BQPSPACE}$ \cite{WATROUS1999space,watrous2003complexity}.}
\end{definition}

\begin{definition}[$\BQTIME$]
We define $\BQTIME\left[f\right]$ to be the set of languages, that can be decided by an $f$-time-uniform general quantum circuit with one qubit of output.\footnote{We note that our definition uses general quantum circuits, while many others (critically \cite{arunachalam2022quantum}) use \emph{unitary} quantum circuits, rather than general, where the measurement is only implicitly done at the very end.  Since we are now only concerned with measurement statistics, the Principle of Deferred Measurement shows that these are equivalent definitions.}
I.e., for a language $L \in \BQTIME\left[f(n)\right]$ there exists a $(\rho)_x \in  \BQTIME\left[f(n)\right]_0$ such that 
\[
    \tr\left[\ketbra{x \in L}{x \in L} \cdot \rho_x \right] \geq \frac{2}{3}
\]
for all $x \in \{0, 1\}^*$.
\end{definition}

\begin{definition}[{$\mathsf{BQSIZE}[s]$}]
    We define $\mathsf{BQSIZE}[s]$ to be the set of languages that can be decided by non-uniform quantum circuits in the $\{H, \mathrm{CNOT}, T\}$ gate set with size at most $O(s)$.
\end{definition}

Finally, we define complexity classes for deterministic computation.
\begin{definition}[$\mathsf{DTIME}$]
We define $\mathsf{DTIME}\left[f\right]$ to be the set of languages, that can be decided by deterministic Turing machine in time $O(f)$.
\end{definition}

We will specifically take $\mathsf{E} \coloneqq \mathsf{DTIME}\!\left[2^{O(n)}\right]$, and $\mathsf{BQSUBEXP} \coloneqq \bigcap_{\gamma \in (0, 1)} \mathsf{BQTIME}\!\left[2^{n^\gamma}\right]$.

\section{Pseudorandomness}\label{sec:pseudo}

In the theory of pseudorandomness, one aims to efficiently construct states that cannot be distinguished from something that is true uniform random (under varying notions of what randomness means).
The strongest form of pseudorandomness would be \emph{statistical} pseudorandomness, such that the object is statistically close to true random.
In this way, even an adversary with unbounded computational time cannot distinguish the pseudorandom object.
Examples of this include $k$-wise independent distributions \cite{alon2007testing} and unitary $t$-designs \cite{dankert2009unitary,oszmaniec2022saturation}.

Unfortunately, constructing such objects can be prohibitively expensive.
Instead, we will settle on a weaker, but still very powerful, notion of \emph{computational} pseudorandomness whereby any computationally bounded adversary cannot distinguish the pseudorandom object from true random.
The goal will be to construct a set of states that looks like a Haar random state to any observer with at most $2^{n^{2\lambda}}$ time for $\lambda \in (0, 1/5)$ (see \cref{def:prs}).
The step was analogously done in \cite{arunachalam2022quantum} for distributions over bitstrings (i.e., a pseudorandom generator \cref{def:prg}), where they impressively gave a conditional PRG with near-optimal stretch with security against $2^{n^{2\lambda}}$-time \emph{quantum} adversaries (see \cref{lem:conditional-prg}).
In fact, we will bootstrap their construction in order to synthesize a set of pseudorandom quantum states.

Let us start by defining a pseudorandom generator (PRG) and a pseudorandom state (PRS), which is a more recent object due to \cite{ji-pseudorandom-states2018}.
Recall that for a function $f$ acting on the natural numbers, we implicitly take $f$ to be $f(n)$ for input size $n$.

\begin{definition}[PRG]\label{def:prg}
    Let $\ell, m : \N \rightarrow \N$, let $s: \N \rightarrow \R^+$, and let $\eps : \N \rightarrow [0, 1]$.
    We say that a family of functions $\left(G: \{0, 1\}^{\ell} \rightarrow \{0, 1\}^{m}\right)_{n \in \N}$ is a infinitely-often $(\ell, m, s, \eps)$-PRG against \emph{uniform} quantum algorithms if no quantum algorithm running in time $s$ can distinguish $G_n(x)$ from $y$ by at advantage at most $\eps$, where $x$ is drawn uniformly from $\{0, 1\}^{\ell}$ and $y$ is drawn uniformly from $\{0, 1\}^{m}$.
    Formally, for all single-qubit sequences $(\rho_x)_{x \in \{0, 1\}^*} \in \sBQTIME\left[s\right]_0$:
\begin{align*}
        \bigg\lvert \E_{x \sim \{0, 1\}^{\ell}} \trc\left[\ketbra{1}{1} \cdot  \rho_{G_n(x)}\right] -  \E_{y \sim \{0, 1\}^{m}} \trc\left[\ketbra{1}{1} \cdot \rho_{y}\right]\bigg\rvert \leq \eps \, 
\end{align*}
holds on infinitely many $n \in \mathbb{N}$.

\end{definition}

\begin{definition}[PRS]\label{def:prs}
    Let $\kappa, \ell, m : \N \rightarrow \N$, let $s: \N \rightarrow \R^+$, and let $\eps : \N \rightarrow [0, 1]$.
    We say that a sequence of keyed pure states $\left(\{\ket{\psi_k}\}_{k \in \{0, 1\}^\kappa} \right)_{n \in \N}$ is an infinitely-often $(\kappa, \ell, m, s, \eps)$-PRS if for a uniformly random $k \in \{0, 1\}^{\kappa}$, no quantum algorithm running in time $s$ can distinguish $m$ samples of $\ket{\psi_k}$ from $m$ samples of a Haar random state on $\ell$ qubits by at most $\eps$.
    Formally, for all $s$-time-uniform quantum circuits $(C_n)_{n \in \N}$ with one qubit of output:
    \begin{align*}
        \bigg\lvert \E_{k \sim \{0, 1\}^{\kappa}} \trc\left[\ketbra{1}{1} \cdot C_{n}\left(\ketbra{\psi_k}{\psi_k}^{\otimes m}\right) \right] -  \E_{\ket{\psi} \sim \mu_{\mathrm{Haar}}} \trc\left[\ketbra{1}{1}  \cdot C_{n}\left(\ketbra{\psi}{\psi}^{\otimes m}\right)\right] \bigg\rvert \leq \eps \,
\end{align*}
holds on infinitely many $n \in \mathbb{N}$.
\end{definition}

We will generally refer to the difference in expectation between the adversary on a pseudorandom object and the adversary on a true random object as the \emph{advantage}.

Note that we can consider the \emph{partial} pure state sequence $\left(\ket{\varphi_x}\right)_{n \in \N, x \in \{0, 1\}^{\kappa(n)}}$, which only holds for the image of $\kappa$ such that $\ket{\varphi_x} = \ket{\psi_k}$ from $\left(\{\ket{\psi_k}\}_{k \in \{0, 1\}^\kappa} \right)_{n \in \N}$.
By trivially letting $\ket{\varphi_x}$ be the zero state for input lengths outside the image of $\kappa$ we get a full state sequence $(\ket{\varphi^\prime_x})_{x \in \{0, 1\}^*}$ such that
\[
    \ket{\varphi^\prime_x} = \begin{cases}
        \ket{\psi_x} & \exists y \in \N, \abs{x} = \kappa(y)\\
        \ket{0} & \text{otherwise}
    \end{cases}.
\]
Therefore, if the PRS $\left(\{\ket{\psi_k}\}_{k \in \{0, 1\}^\kappa} \right)_{n \in \N}$ can be $\delta$-approximately synthesized in time $t: \N \rightarrow \R^+$ relative to the security parameter $n$, then $(\ket{\varphi^\prime_x}) \in \puresBQTIME\left[t \circ \kappa^{-1}\right]_\delta$.
Likewise, if the PRS $\left(\{\ket{\psi_k}\}_{k \in \{0, 1\}^\kappa} \right)_{n \in \N}$ can be $\delta$-approximately synthesized by non-uniform circuit class $\mathfrak{C}[s(n)]$, then $(\ket{\varphi^\prime_x}) \in \mathsf{pureState}\mathfrak{C}\left[s \circ \kappa^{-1}\right]_\delta$.
In an abuse of notation, we will often just refer to $(\ket{\varphi^\prime_x})$ as the PRS.

In order to construct our pseudorandom states we will need to go through an intermediary pseudorandom object called a pseudorandom function, which we will work to define now.

\begin{definition}[Quantum Oracle]
    Given a function $f: \{0, 1\}^{\ell} \rightarrow \{0, 1\}^m$, we define the quantum oracle for $f$ to be 
    \[
        \calO_f \coloneqq \sum_{\substack{x \in \{0, 1\}^\ell\\ y \in \{0, 1\}^m}} \ketbra{x, y + f(x)}{x, y}.
    \]
\end{definition}

We define an oracle circuit, $C^{(\cdot)}$, to be a general quantum circuit with $n$-qubit placeholder unitary $(\cdot)$ such that $C^{\calO}$ is the instantiation of the circuit but with each placeholder replaced by the $n$-qubit gate $\calO$.
We refer to $s$-time-uniform \emph{oracle} quantum circuits as $(C^{(\cdot)}_n)_{n \in \N}$.

Denote $\mathfrak{F}_{\ell, m}: = \{f: \{0, 1\}^{\ell} \rightarrow \{0, 1\}^{m}\}$ as the set of all functions from $\ell$-bits to $m$-bits. 

\begin{definition}[PRF]\label{def:prf}
    Let $\kappa, \ell, m : \N \rightarrow \N$, let $q, s: \N \rightarrow \R^+$, and let $\eps : \N \rightarrow [0, 1]$.
    We say that a sequence of keyed-functions $\left(\{F_k \in \mathfrak{F}_{\ell, m}\}_{k \in \{0, 1\}^\kappa}\right)_{n \in \N}$ is an infinitely-often $(\kappa, \ell, m, q, s, \eps)$-PRF if for a uniformly random $k \in \{0, 1\}^{\kappa}$, no quantum algorithm running in time $s$ can distinguish black-box access to $\calO_{F_k}$ from black-box access to $\calO_f$ for random function $f \in \mathfrak{F}_{\ell, m}$ using at most $q$ queries by at most $\eps$.
    Formally, for all $s$-time-uniform oracle quantum circuits $(C^{(\cdot)}_n)_{n \in \N}$ such that each $C^{\calO}_{n}$ takes no inputs, queries $\calO$ at most $q$ times, and outputs a single qubit state $\rho_{n}^\calO$:
    \begin{align*}
        \bigg\lvert \E_{k \sim \{0, 1\}^{\kappa}} \trc\left[\ketbra{1}{1} \cdot \rho_{n}^{\calO_{F_k}} \right] -  \E_{f \sim \mathfrak{F}_{\ell, m}}\trc\left[\ketbra{1}{1} \cdot \rho_{n}^{\calO_f}\right]\bigg\rvert \leq \eps \, 
\end{align*}
holds on infinitely many $n \in \mathbb{N}$.
\end{definition}

The reason we need these pseudorandom functions, is that there does not seem to be a direct quantum-secure PRG-to-PRS construction in the literature.
There are, however, known constructions of a quantum-secure PRF given a quantum-secure PRG as well as a PRS given a quantum-secure PRF.
We formalize these statements as follows, stated carefully for even sub-exponential time quantum adversaries.

We start with a PRG-to-PRF construction that works especially well when the stretch of the PRG is very large.
A similar idea was used by \cite[Theorem 3.4]{arunachalam2022quantum}, where the output of a PRG is viewed as the truth table of the function.
Thus, if the stretch in the PRG $G$ was the ideal $\{0, 1\}^n \rightarrow \{0, 1\}^{2^n}$ then one could view the string $G(k)$ as the truth table for a function $F_k : \{0, 1\}^n \rightarrow \{0, 1\}$.

\begin{lemma}\label{lem:prg-to-prf}
    Let $G$ be an infinitely-often $(\kappa, m, s, \eps)$-PRG against uniform quantum computations that is computable in time $t$ by a deterministic Turing machine. Then for $\ell \leq \lfloor \log_2 m \rfloor$ there exists an infinitely-often $\left(\kappa, \ell, 1, q, s-O(q\cdot 2^\ell), \eps \right)$-PRF against uniform quantum computations that can be computed in time $O(t)$.
\end{lemma}
\begin{proof}
    Observe that when given a string $s \in \{0, 1\}^{2^k}$ for some $k \in \N$, we can view it as the truth table of a function $\mathrm{fnc}^s : \{0, 1\}^k \rightarrow \{0, 1\}$ such that $\mathrm{fnc}^s(x)$ is the $x$-th bit of $s$.
    Furthermore, if $s$ is a uniformly random string then $f_s$ is a random function in $\mathfrak{F}_{k, 1}$.
    Therefore, let $G^\prime : \{0, 1\}^\kappa \rightarrow \{0, 1\}^{2^\ell}$ compute the first $2^\ell$ bits (note that $2^\ell \leq m$) of the output of $G$.
    Then we define our PRF to be $F_k(x) = \mathrm{fnc}^{G^\prime(k)}(x)$.
    Note that $F_k(x)$ can be computed in time $O(t + 2^\ell) = O(t)$, because $t = \Omega(m) = \Omega\!\left(2^\ell\right)$ as it always takes at least $\Theta(m)$ time to just write down the output string.

    Let $n$ be some hard instance for $G$.
    If we consider an adversary $\calA$ for $G$ that, given either a uniform output of $G(x)$ or a truly random string, by truncating to the first $2^\ell$ bits it will have the truth table to either $F_k$ or $f \in \mathfrak{F}_{\ell, 1}$.
    If there existed a distinguisher $\calB$ for $\{F_k\}$ in time $s - O\left(q \cdot 2^\ell\right)$ and advantage $\eps$, then $\calA$ could simulate the $q$ queries to $O_{F_k}$ or $O_f$ respectively in time $O(q \cdot 2^\ell)$ and therefore distinguish $G$ from a random string in time $s$ with advantage $\eps$ as well.
    By contradiction, this means that $\{F_k\}$ must be an infinitely-often $(\kappa, \ell, s - O(q \cdot 2^\ell), \eps)$-PRF against uniform quantum computations.    
\end{proof}

Note that using something like the \cite{goldreich1984randolli} construction, which is known to be quantum-secure \cite[Theorem 5.5]{Zhandry21construct}, is insufficient for our purposes.
This is because the \cite{goldreich1984randolli} construction of PRFs requires running the PRG multiple times.
Since the PRG in our specific setting is more expensive to compute than the adversaries it is secure against (see \cref{lem:conditional-prg}), the usual way to distinguish the PRG will take more time than is allowed to break the security of the PRG.
Meanwhile, observe that in \cref{lem:prg-to-prf} the PRG has already been run for us once and that we don't need to compute $G$ anymore for further distinguishing purposes.

To get a PRS from a Boolean output PRF, we utilize the following result of \cite{brakerski10.1007/978-3-030-36030-6_10} that gives an information theoretic hardness of distinguishing random binary phase states from Haar random states.

\begin{definition}[Phase State]\label{def:phase-state}
    For $f : \{0, 1\}^n \rightarrow \{0, 1\}$, define $n$-qubit state $\ket{f}$ as:
    \[
        \ket{f} \coloneqq \frac{1}{\sqrt{2^n}}\sum_{x \in \{0, 1\}^n}(-1)^{f(x)}\ket{x}.
    \]
\end{definition}

\begin{lemma}[{\cite[Theorem 1]{brakerski10.1007/978-3-030-36030-6_10}}]\label{lem:random-phase-state-vs-haar}
    For all $t \in \R^+$, $m$-copies of $\ket{f}$, for $f$ chosen uniformly from $\mathfrak{F}_{n, 1}$, cannot be distinguished from $m$-copies of a Haar random state by any (potentially computationally unbounded) adversary by advantage at most $\frac{4m^2}{2^n}$.
\end{lemma}

\begin{lemma}[Generalization of {\cite[Section 3.1]{brakerski10.1007/978-3-030-36030-6_10}}]\label{lem:prf-to-prs}
    \sloppy
    Let $(\{F_k\})$ be an infinitely-often $(\kappa, \ell, 1, m, s, \eps)$-PRF against uniform quantum computations that can be computed in time $t$. Then there exists an infinitely-often $\left(\kappa, \ell, m, s - O\left(\ell\right), \eps + \frac{4m^2}{2^\ell} \right)$-PRS against uniform quantum computations that can be synthesized \emph{exactly} in time $O(t + \ell)$ in the $\{H, \mathrm{CNOT}, T\}$ gate set.
\end{lemma}
\begin{proof}
    We first note that for each $F_k : \{0, 1\}^\ell \rightarrow \{0, 1\}$, $\ket{F_k}$ can be synthesized \emph{exactly} in the $\{H, \mathrm{CNOT}, T\}$ gate set.
    This starts by observing that we can construct $O_{F_k}$ using $O(t)$ Toffoli gates, which are known to be universal for classical computation.
    Furthermore, the Toffoli gate has an exact construction in the $\{H, \mathrm{CNOT}, T\}$ gate set \cite{welch2016efficient}.
    Thus, initializing the state $\ket{0^\ell}\ket{1}$, we can apply $H^{\otimes (n+1)}$ then $O_{F_k}$ to get
    \[
         \frac{1}{\sqrt{2^\ell}}\sum_{x \in \{0, 1\}^\ell}(-1)^{F_k(x)}\ket{x}\ket{1}.
    \]
    Tracing out the last qubit gives us $\ket{F_k}$.
    This takes time at most $O(\ell + t)$.
    Furthermore, for arbitrary function $f \in \mathfrak{F}_{\ell, 1}$, if $O_f$ is given as a black-box that takes $O(1)$ time, then this becomes time $O(\ell)$ and uses only a single query.
    
    If $(\{F_k\}_{k \in \{0, 1\}^\kappa})_{n \in \N}$ represents the $(\kappa, \ell, 1, m, s, \eps)$-PRF, we now claim that $\left( \{\ket{F_k}\} \right)$ forms our desired PRS.
    First, let $n \in \N$ be an arbitrary ``hard'' instance of the infinitely-often PRF.
    We define $3$ states
    \[
        \rho^m_{\text{PRS}} \coloneqq \Ex_{k \sim \{0, 1\}^\kappa}\left[\ketbra{F_k}{F_k}^{\otimes m}\right], \,\,\,\, \rho_{\text{phase}}^m \coloneqq \Ex_{f \sim \mathfrak{F}_{\ell, 1}}\left[\ketbra{f}{f}^{\otimes m}\right], \,\,\,\, \rho^m_{\mu_{\mathrm{Haar}}} \coloneqq \Ex_{\ket{\psi} \sim \mu_{\mathrm{Haar}}}\left[\ketbra{\psi}{\psi}^{\otimes m}\right]
    \]
    and show that they cannot be easily distinguished by an $\left[s-O(\ell)\right]$-time quantum algorithm.
    By \cref{def:prf} and the fact that $\ket{f}$ can be constructing in $O(\ell)$ time and a single query from $O_f$, $\rho^m_{\text{PRS}}$ and $\rho_{\text{phase}}^m$ respectively cannot be distinguished by more than $\eps$ in time $s - O(\ell)$.
    By applying \cref{lem:random-phase-state-vs-haar}, $\rho_{\text{phase}}^m$ and $\rho^m_{\mu_{\mathrm{Haar}}}$ cannot be distinguished by \emph{any} adversary with advantage more than $\frac{4m^2}{2^\ell}$.
    Using the reverse triangle inequality, an $\left[s-O(\ell)\right]$-time algorithm that distinguishes between $\rho^m_{\text{PRS}}$ and $\rho^m_{\mu_{\mathrm{Haar}}}$ with advantage $\eps + \frac{4m^2}{2^\ell}$ would imply an $\left[s-O(\ell)\right]$-time distinguisher for $\rho^m_{\text{PRS}}$ and $\rho_{\text{phase}}^m$ with advantage $\eps$, a contradiction.
    Therefore $(\{\ket{F_k}\}_{k \in \{0, 1\}^\kappa})_{n \in \N}$ forms an infinitely-often $\left(\kappa, \ell, m, s - O\left(\ell\right), \eps + \frac{4m^2}{2^\ell}\right)$-PRS against uniform quantum computations
\end{proof}

\begin{remark}
In a different gate set, a quantum adversary may not necessarily be able to exactly prepare the binary phase state (see \cref{def:phase-state}).
Because the adversary in \cref{lem:prf-to-prs}, when given $O_f$, only needs to apply $n+1$ Hadamard gates and a single $X$ gate to prepare $\ket{f}$, the \nameref{lem:solovay-kitaev} ensures that $\ket{f}$ can be prepared in any universal gate set to trace distance $\exp(-k)$ using $O\left(k^{2.71} + \log^{2.71} n\right)$ extra time.
By the reverse triangle inequality, contractivity of trace distance under trace-preserving completely positive maps, and \cref{fact:td-to-tv}, this would turn the PRS in \cref{lem:prf-to-prs} into an infinitely-often $\left(\kappa, \ell, m, s - O\left(\ell + k^{2.71} + \log^{2.71} n\right), \eps + \frac{4m^2}{2^\ell} + \exp(-k) \right)$-PRS against uniform quantum computation for all $k \in \N$.
We will choose to ignore this effect because (1) when used for our purposes, $k = \Theta(\ell)$ and $\ell = \omega(\poly\log n)$ such that the effect will be negligible and (2) it is not such a strong assumption that the universal gate set used instead can still exactly create the Hadamard gate and $X$ gate in constant size.
\end{remark}

\subsection{Pseudorandom Objects From Decision Problem Separations}

We finally state the critical result of \cite{arunachalam2022quantum}, which showed how to produce a nearly-optimal quantum-secure PRG given a complexity theoretic assumption.
By combining \cref{lem:conditional-prg,lem:prg-to-prf,lem:prf-to-prs} we get our desired final result of a conditional PRS that is secure even against sub-exponential time adversaries that can be constructed in $2^{O(n)}$ time.

\begin{lemma}[{\cite[Theorem 3.2, Theorem 5.1]{arunachalam2022quantum}}]\label{lem:conditional-prg}
    Suppose there exists a $\gamma > 0$ such that $\PSPACE \not\subset \BQTIME\left[2^{n^\gamma}\right]$.
    Then, for some choice of constants $\alpha \geq 1$ and $\lambda \in (0, 1/5)$, there exists an infinitely-often $(\ell, m, s, 1/m)$-PRG against uniform quantum computations where $\ell(n) \leq n^\alpha$, $m(n) = \lfloor 2^{n^\lambda} \rfloor$, and $s(n) = 2^{n^{2\lambda}}$.

    In addition, the PRG is computable by a deterministic Turing machine in time $O\left(2^{\ell}\right)$.
\end{lemma}

In the following statements and proofs of \cref{cor:conditional-prf,cor:conditional-prs-1}, note that $r(n)$ takes the place of $m(n)$, and $\kappa(n)$ takes the place of $\ell(n)$ in the \cref{lem:conditional-prg}.

\begin{corollary}\label{cor:conditional-prf}
     Suppose there exists a $\gamma > 0$ such that $\PSPACE \not\subset \BQTIME\left[2^{n^\gamma}\right]$.
    Then, for some choice of constants $\alpha \geq 1$ and $\lambda \in (0, 1/5)$, there exists an infinitely-often
    \[\left(\kappa, \ell, 1, q, s - O\left(m \cdot 2^\ell\right), 1/r \right)\text{-PRF}\]
    against uniform quantum computations where $\kappa(n) \leq n^\alpha$, $r(n) = \lfloor 2^{n^\lambda} \rfloor$, $\ell \leq \lfloor \log_2 r \rfloor$, and $s(n) = 2^{n^{2\lambda}}$.

    In addition, the PRF is computable by a deterministic Turing machine in time $O\left(2^{\kappa}\right)$.
\end{corollary}
\begin{proof}
    By \cref{lem:conditional-prg}, there exists an infinitely-often $(\kappa, r, s, 1/r)$-PRG against uniform quantum computations that can be computed in time $O\left(2^{\kappa}\right)$.
    Applying \cref{lem:prg-to-prf}, there must exist an infinitely-often \[\left(\kappa, \ell, 1, q, s - O\left(m \cdot 2^\ell\right), 1/r \right)\text{-PRF}\] against uniform quantum computations that can be computed in time $O(2^{\kappa})$.
\end{proof}

\begin{corollary}\label{cor:conditional-prs-1}
     Suppose there exists a $\gamma > 0$ such that $\PSPACE \not\subset \BQTIME\left[2^{n^\gamma}\right]$.
    Then, for some choice of constants $\alpha \geq 1$ and $\lambda \in (0, 1/5)$, there exists an infinitely-often
    \[
        \left(\kappa, \ell, m, s-O(m \cdot 2^\ell), \frac{1}{r} + \frac{4m^2}{2^\ell}\right)\text{-PRS}
    \]
    against uniform quantum computations where $\kappa(n) \leq n^\alpha$, $r(n) = \lfloor 2^{n^\lambda} \rfloor$, $\ell \leq \lfloor \log_2 r \rfloor$, and $s(n) = 2^{n^{2\lambda}}$.

    In addition, the PRS can be \emph{exactly} synthesized in time $O\left(2^{\kappa}\right)$ in the $\{H, \mathrm{CNOT}, T\}$ gate set.
\end{corollary}
\begin{proof}
    We can utilize \cref{cor:conditional-prf,lem:prf-to-prs} to show that there must exist an infinitely-often
    \[
        \left(\kappa, \ell, m, s-O(m \cdot 2^\ell + \ell), \frac{4m^2}{2^\ell} + \frac{1}{r}\right)\text{-PRS}
    \]
    against uniform quantum computations that can be exactly constructed in time $O\left(2^{\kappa}\right)$ in the $\{H, \mathrm{CNOT}, T\}$ gate set.
\end{proof}

\subsection{Pseudorandom States From State Synthesis Separations}

We note that \cref{cor:conditional-prs-1} relies on a separation between decision problem complexity class separations.
However, it will be important for us to condition on state synthesis class separations instead for the win-win argument in our main result \cref{thm:main}.
Observe that, while decision problem separations immediately imply state synthesis separations, the converse is not always clear.
However, if the size of the circuits are not too large then we can say the following.

\begin{lemma}\label{lem:state-implies-decision}
    \sloppy
    Let $k : \N \rightarrow \R^+$.
    For any $\mathsf{stateA}_\delta \subset \mathsf{stateBQP/poly}_\delta$, if $\mathsf{stateA}_\delta \not\subset \sBQTIME\left[k \cdot f\right]_{\delta + \exp(-k)}$ then $\mathsf{A} \not\subset \BQTIME\left[\frac{f}{n^{\nu}}
    \right]$ for some $\nu \geq 1$.
\end{lemma}
\begin{proof}
    By assumption, there exists some fixed state sequence $\{\rho_x\}_{x \in \{0, 1\}^n}$ synthesized by a sequence of $s$-size unitary circuits $(C_x)_{x \in \{0, 1\}^*}$  to accuracy $\delta$ for $s = \poly(n)$ and sufficiently large $n$.
    On the other hand, no circuit sequence that can be described by a deterministic Turing machine using $O(k \cdot f)$ time can approximately synthesize $\{\rho_x\}_{x \in \{0, 1\}^n}$ to even $\delta + \exp(-k)$ accuracy in trace distance.
    By \cref{lem:num-circuits}, the description length of each $C_x$ is at most $d \coloneqq O(s \log s)$ and define $\nu \geq 1$ to be some constant such that $d \log d \leq n^\nu$.
    We now define the language $L$ on $n^\prime \coloneqq n + \lceil \log_2 d \rceil= O(n)$ bits to be the problem of: Given inputs $x \in \{0, 1\}^n$ and $i \in \{0, 1, \dots, d-1\}$ (encoded in binary), output the $i$-th bit of the circuit description of $C_{x}$.
    $L$ is trivially in $\mathsf{A}$, and we will now show that $L \not\in \BQTIME\left[f\right]$.

    Observe that if $L$ could be decided in \emph{deterministic} $f$ time then the whole circuit description of $C_x$ could learned by iterating through all $d$ possible values of $i$, giving an $O(d \cdot f)$-time algorithm to describe the synthesis circuit for each $\ket{\psi_x}$.
    This na\"ively implies that $\mathsf{A} \not\subset \mathsf{DTIME}\left[\frac{f}{d}\right]$, where $\mathsf{DTIME}$ is the classical deterministic version of $\BQTIME$.
    We now need to argue that not even a \emph{quantum} algorithm could succeed.
    
    For the sake of contradiction, suppose on inputs $(x, i)$ of length $n^\prime$ that there did exist an $\frac{f}{n^\nu}$-time-uniform general quantum circuit $U$ that decided $L$ (i.e., $L \in \BQTIME\left[\frac{f}{n^\nu}\right]$).
    Using standard error reduction for decision problems we can construct a quantum circuit $U^\prime$ such that it has at most an $\exp(-2k)$ failure probability, with a multiplicative overhead of $O(k)$.
    Then, via the union bound, there exists a quantum circuit $V$ that uses $U^\prime$ as a sub-routine and, for a fixed $x \in \{0, 1\}^n$, iterates through all of the various possible bits $i$ to learn the entire description of $C_x$ with at most an $\exp(-2k)$ failure probability and an additional $O(d \log d)$ multiplicative overhead.\footnote{We note that Toffoli gates, which are classically universal, can be built exactly in the $\{H, \mathrm{CNOT}, T\}$ gate set \cite{welch2016efficient}.}
    This is equivalent to saying that the fidelity with $\ket{\textsc{desc}(U)}$ is at least $1-\exp(-2k)$, so by the upper bound of \cref{{fact:td-vs-fidelity}}, there is at most an $\exp(-k)$ trace distance to the computational basis state with the correct description of $\textsc{desc}(C_x)$).
    In $O(d \log d)$ time, the algorithm can then utilize \cref{fact:universal-circuit} to apply $C_x$ to a set of ancillary qubits, then trace out the qubits holding the description.
    By the \cref{fact:td-to-tv}, our output state will have distance at most $\exp(-k)$ from the output state of $C_x$.
    Finally, by the triangle inequality and the accuracy of $(C_x)$ in synthesizing $(\rho_x)$ to distance $\delta$, the resulting output quantum state $\hat{\rho}_x$ will then have
    \[
        \tracedistance{\rho_x, \hat{\rho}_x} \leq \delta + \exp(-k)
    \]
    for all sufficiently large $n$.
    
    Overall, for all sufficiently large $n$ and arbitrary $x \in \{0, 1\}^n$, a circuit that approximates $\rho_x$ to $\delta + \exp(-k)$ trace distance can be output by a deterministic Turing machine running in $O\!\left( k f \frac{d \log d}{n^\nu}\right) = O(k \cdot f)$ time such that $(\rho_x) \in \sBQTIME\left[k \cdot f\right]_{\exp}$.
    This is a contradiction of our initial assumption and it follows that $L \not\in \BQTIME\left[f\right]$ even though $L \in \mathsf{A}$.
\end{proof}

We note that the proof of \cref{lem:state-implies-decision} can also be used to show that $\mathsf{stateA}_\delta \not\subset \sBQTIME[f]_\delta$ implies $\mathsf{A} \not \subset \mathsf{EQTIME}[f]$ where $\mathsf{EQTIME}$ is the set of languages that can be \emph{exactly} decided by $f$-time-uniform quantum circuits.
This is because $U$, the circuit that previously decided $L$ with bounded error will have no error when used in contradiction.
The resulting circuit $V$ that finds the description of $C_x$ will then also have no error.

The following is a consequence of \cref{lem:state-implies-decision}, written in a way to most easily use in proving our main result \cref{thm:main}.

\begin{corollary}
\label{cor:conditional-prs-2}
     Suppose there exists a $\gamma > 0$ such that $\puresPSPACE_0 \not\subset \puresBQTIME\left[2^{n^\gamma}\right]_{\exp}$.
    Then for every $c \geq 1$, for some choice of constants $\alpha \geq 1$ and $\lambda \in (0, 1/5)$ and sufficiently large $n \in \N$, there exists an infinitely-often \[\left(\kappa, \lfloor \log_2 r \rfloor^{1/c}, m, s, \frac{4m^2}{2^\ell}+\frac{1}{r}\right)\text{-PRS}\] against uniform quantum computations where $\kappa \leq n^\alpha$, $r = \lfloor 2^{n^\lambda} \rfloor$, $s = 2^{n^{2\lambda}}$, and $m = o\left(\frac{s}{r}\right)$.
    In addition, the PRS lies in $\puresBQE_{\exp}$.
\end{corollary}
\begin{proof}
    \fussy
    Observe by \cref{def:pspacesize} that $\puresPSPACE_\delta \subset \mathsf{stateBQP/poly}_\delta$.
    By \cref{lem:state-implies-decision} and the contrapositives of \Cref{fact:2,fact:3} respectively, $\PSPACE \not\subset \BQTIME\left[\frac{2^{n^{\gamma}}}{\poly(n)}\right]$.
    This further implies $\PSPACE \not\subset \BQTIME\left[2^{n^{0.99 \cdot \gamma}}\right]$, such that we now invoke \cref{cor:conditional-prs-1} with $m = o\left(\frac{s}{r}\right)$.
\end{proof}

\section{Quantum State Learning}\label{sec:learning}

In this section, we prove that any sub-exponential time quantum state tomography for a class of \emph{pure} states implies the ability to distinguish said class of states from Haar random in sub-exponential time.
This ultimately implies that any sequence of pure states that efficiently learned cannot be used to form a sub-exponential-time-secure PRS ensemble, which will be necessary for proving a contradiction in \cref{sec:final-proof}.
As will be explained shortly, since the task can always be done in exponential samples (and exponential time), the informal takeaway is that slightly non-trivial learning algorithms imply some sort of lower bound against PRS constructions (with the more non-trivial learning algorithms leading to better lower bounds).

Note that \cite[Theorem 14]{zhao2023learning} and \cite[Theorem 2]{ji-pseudorandom-states2018} make similar statements. We supplement them with a simple proof that is both tighter and more fine-grained in parameters of the adversary. See \cref{remark:learning-to-distinguishing} for a discussion of the differences and why proving \cref{lem:learning-to-distinguishing} was necessary for our purposes.

We begin by defining what it means to learn a quantum pure state in the tomographical sense.

\begin{definition}
    Let $\calC_n$ be a class of $n$-qubit \emph{pure} quantum states such that $\calC = \bigcup_{n \geq 1} \calC_n$.
    We say that $\calC$ is $(m, t, \eps, \delta)$-learnable if there exists a $t$-time-general quantum circuit family $(C_n)_{n \in \N}$ that, with probability at least $1-\delta$, running $C_n$ on $m$ samples of $\ket \psi \in \calC_n$ outputs a \emph{description} of a unitary circuit that prepares the state $\hat{\rho}$ such that $\tracedistance{\ketbra{\psi}{\psi}, \hat{\rho}} \leq \eps$.
\end{definition}

We now introduce a weaker form of learning, which simply involves distinguishing a state from $\calC$ from a Haar random state.
Informally, distinguishing is the task of breaking pseudorandom states.
As such, this will be the true notion of learning that will drive our results.

\begin{definition}\label{def:distinguish}
    Let $\calC_n$ be a class of $n$-qubit quantum states such that $\calC = \bigcup_{n \geq 1} \calC_n$. We say that $\calC$ is $(m, t, \eps)$-distinguishable if there exists a $t$-time-uniform quantum circuit family with one bit of output $(C_n)_{n \in \N}$ that satisfies the following: For all sufficiently large $n \in \mathbb{N}$, for every $\rho \in \calC_n$,
    \begin{align*}
        \bigg\lvert \trc\left[\ketbra{1}{1} \cdot C_{n}\left(\rho^{\otimes m}\right) \right] -  \E_{\ket{\psi} \sim \mu_{\mathrm{Haar}}} \trc\left[\ketbra{1}{1}  \cdot C_{n}\left(\ketbra{\psi}{\psi}^{\otimes m}\right)\right] \bigg\rvert \geq \eps.
\end{align*}
\end{definition}

Note the difference between this and a PRS, where the distinguishing needs to hold for worst-case $\rho$, rather than in average-case.
We could have defined a model of learning/distinguishing that applied in average-case over all sufficiently large subsets of $\calC$.
However, this notion does not seem to be as common in the quantum state learning literature.

It is also important to observe that an $t$-time and $m$-sample distinguishing algorithm, in the model of learning, runs in time $t(n)$ and samples $m(n)$ where $n$ is the number of qubits.
Meanwhile, for a $(\cdot, \ell, m, s, \cdot)$-PRS, the $t$-time adversary runs in time $t(n)$ where $n$ is \emph{not} the number of qubits, but rather $\ell$ is.
Therefore, a $t$-time and $m$-sample distinguishing algorithm runs in time $t(\ell)$ and uses number of samples $m(\ell)$ relative to $n$.

Because of \cref{fact:td-to-tv} and \cref{cor:td-multi-copy}, we can show that distinguishing is robust against small perturbations.

\begin{lemma}\label{lem:robust-distinguish}
     Let $\calC$ be $( m, t, \eps)$-distinguishable and let $\calC_\delta$ be the class of states $\delta$-close to $\calC$ in trace distance.
     Then $\calC_\delta$ is $(m, t, \eps-\sqrt{m}\delta)$-distinguishable.
\end{lemma}
\begin{proof}
    Let $\ket{\phi} \in \calC$ be the closest state in $\calC$ to $\ket{\psi_x} \in \calC_\delta$ such that $\tracedistance{\ket{\psi_x}, \ket{\phi_x}} \leq \delta$.
    Using \cref{cor:td-multi-copy}, we find that $\tracedistance{\ket{\psi_x}^{\otimes m}, \ket{\phi_x}^{\otimes m}} \leq \sqrt{m} \cdot \delta$

    Finally, let $(C_n)_{n \in \N}$ form the distinguisher for $\calC$.
    Then by \cref{fact:td-to-tv}
    \[
        \bigg\lvert \trc\left[\ketbra{1}{1} \cdot C_n\left(\ketbra{\psi}{\psi}^{\otimes m}\right)\right] - \trc\left[\ketbra{1}{1}  \cdot C_n\left(\ketbra{\phi}{\phi}^{\otimes m}\right)\right]\bigg\rvert \leq \sqrt{m}\delta.
    \]
    By the triangle inequality, the distinguishing power of $(C_n)$ is at least $\eps - \sqrt{m}\delta$.
    
\end{proof}

We now work to show the implication that learning implies distinguishing from Haar random.
The intuition is two-fold.
The first is that a Haar random quantum state cannot be learned to even $\eps = 1-\frac{1}{2^{o(n)}}$ accuracy with sub-exponential samples.
The second is that whether or not a learning algorithm for \emph{pure} states succeeds can be verified using the SWAP test \cite{barenco1998stabilization,buhrman2001quantum,gottesman2001quantum}.
Together, running the SWAP test on the output of the learning algorithm should distinguish whether or not the learning algorithm was fed a ``correct'' input state or a Haar random state.

We utilize the following information theoretic result, which gives the optimal fidelity for learning a Haar random state given a fixed number of copies.
It was recently used by Yuen to prove optimal sample complexity for fidelity-based quantum state tomography \cite{Yuen2023improvedsample} and holds for even adversaries with unbounded computational power.

\begin{lemma}[{\cite[Eqn 16]{BRU1999249}}]\label{lem:haar-hard-to-learn}
    Given $m$ copies of a Haar random state $\ket \psi$, any quantum algorithm outputting a state $\hat{\rho}$ must have \[
    \Ex\left[\calF\left(\ketbra{\psi}{\psi}, \hat{\rho}\right)\right] \leq \frac{m + 1}{m + 2^n}
    \]
    where the expectation is over the randomness in the measurement results.
\end{lemma}

\begin{lemma}\label{lem:learning-to-distinguishing}
    If $\calC$ is $(m, t, 1-\eta, 1-\lambda)$-learnable for $\eta \geq 2^{-o(n)}$, $\lambda \geq 2^{-o(n)}$, and $m = 2^{o(n)}$.
    Then $\calC$ is $\left(m+1, t \log t, \frac{1-o(1)}{2} \eta \lambda\right)$-distinguishable.
\end{lemma}
\begin{proof}
    The distinguisher is simple to state and uses $m+1$ copies of $\ket \psi$.
    Let $(C_n)_{n \in \N}$ be the learning algorithm for $\calC$.
    First, run $C_n$ on $m$ copies of $\ket \psi$ to output a general quantum circuit that prepares $\hat{\rho}$.
    Then perform a SWAP test between $\hat{\rho}$ and an extra copy of $\ket{\psi}$.
    Recall that the SWAP test between states $\hat{\rho}$ and $\ketbra{\psi}{\psi}$ accepts with probability $\frac{1}{2}\left(1 + \braket{\psi | \hat{\rho} | \psi}\right)$.
    Therefore, we simply need to show that the expected bias of the SWAP test (i.e., half the fidelity $\frac{1}{2}\braket{\psi | \hat{\rho} | \psi}$) for each of our two cases is separated by $\eps$ for sufficiently large $n \in \N$.
    
    In the case where $\ket \psi \in \calC$, by \cref{fact:td-vs-fidelity}, $C_n$ will output the description of a circuit that synthesizes $\hat{\rho}$ such that $\braket{\psi | \hat{\rho} | \psi} \geq 1 - \tracedistance{\hat{\rho}, \ketbra{\psi}{\psi}} \geq \eta$ with probability at least $\lambda$.
    The test will accept with expected bias at least $\beta_1 \coloneqq \eta \cdot \lambda$ as a result.
    Conversely, when $\ket \psi$ is Haar random, we know from \cref{lem:haar-hard-to-learn} that the average fidelity is at most $\frac{m + 1}{2^n + 1} = \frac{1}{2^{\Theta(n)}}$.
    By linearity of expectations, the expected bias of the SWAP test is therefore at most $\frac{m + 1}{2^n + 1}$ as well.
    Finally, because $\eta \cdot \lambda \geq 2^{-o(n)}$, the gap in the biases of the two cases, $\abs{\beta_1 - \beta_2} = \frac{1}{2}\abs{\eta \lambda - 2^{-\Theta(n)}}$, is at least $\frac{1-o(1)}{2} \eta \lambda$  for some sufficiently large value of $n$.
    
    We note that the size of the circuit producing $\hat{\rho}$ is at most $O(t)$.
    Therefore, producing $\hat{\rho}$ from its description must take time at most $O(t \log t)$ by applying \cref{fact:universal-circuit}.
    Therefore, running $C_n$, producing $\hat{\rho}$, then running a SWAP test takes at most $O(t \log t)$ time as well.
    This show that $\calC$ is $\left(m+1, t \log t, \frac{1-o(1)}{2} \eta \lambda \right)$-distinguishable.
\end{proof}

We remark that, realistically, one would expect to be able to able to produce the output of the learned circuit, $\hat{\rho}$ in time $O(t)$ rather than $O(t \log t)$.
With this assumption, we get a $(m+1, t, (1-o(1)) \cdot \eta \lambda)$-distinguisher instead.
Furthermore, this logarithmic overhead is a result of treating the circuit as a black-box algorithm; usually an efficient learning algorithm can be made into a distinguisher in a white-box way in time at least as fast, and sometimes \emph{much} faster. See \cref{sec:discussion} for discussion.

\section{\texorpdfstring{$\puresPSPACE_0 \not\subset \puresCircuit{n^k}_{0.49}$}{pureStatePSPACE not contained in stateBQSIZE[poly(n)]}}

Our final major technical contribution will be a proof that for any fixed $k \in \N$, \[\puresPSPACE_0 \not\subset \puresCircuit{n^k}_{0.49}.\]
This is analogous to \cite[Lemma 3.3]{arunachalam2022quantum}, where they utilize the fact that for any fixed $k > 0$, $\PSPACE$ can diagonalize against Boolean circuits of size $n^k$.
However, while state synthesis problems are a generalization of decision problems, we cannot simply use the fact that for all $k \geq 1$, $\mathsf{PSPACE} \not\subset \mathsf{BQSIZE}[n^k]$ \cite{chia_et_al:LIPIcs.ITCS.2022.47}, as the notions of non-uniformity are actually different (see \cref{remark:non-uniform-difference}).
In more detail, the proofs of \cite[Lemma 3.3]{arunachalam2022quantum} and \cite{chia_et_al:LIPIcs.ITCS.2022.47} rely on the fact that the ways in which a \emph{single} circuit of bounded size can process all $x \in \{0, 1\}^n$ is limited.
Since a non-uniform circuit is now allowed to depend on $x$ in the case of state synthesis, it can decide all languages such that the decision separation does not lift to state synthesis in the way it does for uniform models of computation.

Instead, for each $n \in \N$ we need to find a single state that is sufficiently complicated enough, rather than a state sequence (or language) whose relationship with $x \in \{0, 1\}^n$ is complicated.
Even the high-level proof techniques for diagonalization no longer apply.
This is immediately obvious in the sense that quantum states lie in a continuous space, yet bit strings exist as a discrete set.
As such, while slightly perturbing a bitstring always creates a ``far'' string, a non-trivial action on a quantum state could leave you with a state very close in trace distance still.
To combat this, we will attempt to discretize the set of $n$-qubit quantum states via its packing number with respect to the trace distance (see \cref{def:packing-number}).
From there, we will argue over the course of \cref{ssec:circuit-heirarchy} that there exists a circuit size hierarchy, in that non-uniform \emph{unitary} circuits of larger size can always create states far away from non-uniform \emph{general} circuits of a smaller size.
Then in \cref{ssec:diagonalization}, we will use the ability to estimate the trace distance of states produced by polynomial size circuits using only a polynomial amount of space.
By doing a brute-force search over all possible states produced by $\puresCircuit{n^k}_{0.49}$, we can find a state in $\mathsf{pureStateCircuit}\left[n^{k^\prime}\right]_0$ for $k^\prime > k$ that is not in $\puresCircuit{n^k}_{0.49}$.
By using said state as part of the state sequence in $\puresPSPACE_0$ we get our desired separation of $\puresPSPACE_0 \not\subset \puresCircuit{n^k}_{0.49} $.

\begin{remark}
While a zero-error state synthesis class may seem odd, seeing as it is gate set dependent, we note that \cref{thm:pspace-diagonalization} holds for arbitrary universal gate sets as long as $\puresPSPACE$ and $\puresCircuit{n^k}$ share the \emph{same} universal gate set.
If this happens to not be the case, then the result simply becomes $\puresPSPACE_{\exp} \not\subset \puresCircuit{n^k}_{0.49}$ via the \nameref{lem:solovay-kitaev}.
\end{remark}

\subsection{Non-Uniform Quantum Circuit Size Hierarchy Theorem}\label{ssec:circuit-heirarchy}

In order to discretize the set of quantum states, the packing number $\npack(\calY, \eps)$ counts how many states in a set $\calY$ are $\eps$-far apart from each other in trace distance.
The high-level idea for its use will be that if $A \subset B \subseteq \calS$ and $\npack\left(A, \eps\right) < \npack\left(B, \eps\right)$ then some state in $B$ is $\eps$-far from \emph{every} state in $A$.

\begin{definition}[Packing Number]\label{def:packing-number}
    Given a set of pure states on $n$-qubits $\calY \subseteq \calS$ we define the \emph{packing number} to be,
    \[\npack\left(\calY, \eps\right) \coloneqq \max\{\abs{S} : \forall \ket{\psi}, \ket{\phi} \in S, \tracedistance{\ket{\psi}, \ket{\phi}} \geq \eps, S \subseteq \calY\}\]
\end{definition}

To get a hierarchy theorem, we will need both a lower and upper bound on $\npack$ for states created by polynomial size circuits.
We start with a very crude upper bound by counting how many possible circuits there are, even if they might produce the same quantum state (or states that are $\eps$-close to each other).

\begin{proposition}\label{cor:npack-upper}
    For $s \geq n$, the number of general quantum circuits of size $s$ is at most
    \[
        2^{s \cdot \left(3\log_2 s + 4\right)}.
    \]
\end{proposition}
\begin{proof}
    By \cref{lem:num-circuits}, we can write all such circuits using at most $s \cdot \left(3\log_2 s + 4\right)$ bits. Enumerating over all bits of that length provides an upper bound on the number of circuits.
\end{proof}

For the lower bound, we use the following result by \cite{oszmaniec2022saturation} that was originally used to show that random circuits for unitary $t$-designs and chaotic quantum systems.

\begin{lemma}{\cite[Lemma 11]{oszmaniec2022saturation}}\label{lem:npack-lower}
    Let $\purestate^r$ be the set of pure states generated by depth $r$ unitary quantum circuits from gate set $\calG$ on $n$-qubits. Then
    \[\npack(\purestate^r, \eps) \geq \left(\frac{2^n(1-4\eps^2)}{\alpha(r)}\right)^{\alpha(r)},\]
    where $\alpha(r) \coloneqq \lfloor \left(\frac{r}{n^2 \cdot c(\calG)}\right)^{1/11} \rfloor$ and $c(\calG)$ is a constant depending on the gate set $\calG$.
\end{lemma}

With both upper and lower bound in hand, we now state our hierarchy theorem.

\begin{lemma}\label{lem:circuit-size-hierarchy}
    Let $\mixedstate^s$ be the set of states generated by size $s$ \emph{unitary} quantum circuits from $\poly(n)$-qubits to $n$-qubits and let $\purestate^{s^\prime}$ be the set of quantum pure states generated by size $s^\prime$ \emph{unitary} quantum circuits on $n$-qubits (i.e., no ancilla). For sufficiently large $n$ and $s = 2^{o(n)}$, $\npack(\mixedstate^s, 0.495) < \npack(\purestate^{s^\prime}, 0.495)$ for some $s^\prime = O(s^{12} n^2 c(\calG))$ where $c(\calG)$ is a constant depending on the gate set $\calG$. 
\end{lemma}
\begin{proof}
    By \cref{cor:npack-upper}:
    \[
        \npack\left(\mixedstate^s, 0.495\right) \leq 2^{s \cdot \left(3\log_2 s + 4\right)} = o\left(2^{s^{12/11}}\right)
    \]

    In contrast, we use \cref{lem:npack-lower} to show that $\npack(\purestate^{s^\prime}, 0.495)$ is strictly larger, meaning that there must exist a state in $\purestate^{s^\prime}$ that is $\eps$ far from all states in $\mixedstate^s$.
    Since $s^\prime = O\left(s^{12} n^2 c(\calG)\right)$ we find that $\alpha(s^\prime) = \kappa \cdot s^{12/11}$  for some constant $\kappa = O(1)$.
    We then note that the depth of a circuit of size $s$ is at most $s$ as well.
    Applying \cref{lem:npack-lower}:
    \begin{align*}
        \npack\left(\mixedstate^{s^\prime}, \eps\right) &\geq \left(\frac{2^n(1-4\eps^2)}{\alpha(s^\prime)}\right)^{\alpha(s^\prime)}\\
        &= \left(\frac{2^n(1-4\eps^2)}{\kappa \cdot r^{12/11}}\right)^{\kappa \cdot s^{12/11}}\\
        &\geq \left(2^{n - 5.66 - \log_2 s - \log_2 \kappa}\right)^{\kappa \cdot s^{12/11}}\\
        &= 2^{\Theta\left(n s^{12/11}\right)}
    \end{align*}
    where the last step holds because $\log_2\left(1-4\cdot 0.49^2\right) \geq -5.66$ and $s$ is sub-exponential in $n$ such that $\log_2 s = o(n)$.
    Thus, for sufficiently large number of qubits $n$ we find that $\npack(\mixedstate^s, 0.495) < \npack(\purestate^{s^\prime}, 0.495)$.
\end{proof}

We emphasize that for $\npack(\purestate^{s^\prime}, \eps)$ we don't allow any ancilla.
Nevertheless, we still achieve the desired hierarchy result because our bound also lower bounds states produced by size $s^\prime$ \emph{with} ancilla.
However, an improvement to \cref{lem:npack-lower} that accounts for extra qubits could greatly reduce the constant powers relating $s$ to $s^\prime$.

\subsection{Quantum State ``Diagonalization''}\label{ssec:diagonalization}

From \cref{lem:circuit-size-hierarchy} we know that for circuits with sub-exponential size, a polynomially larger size can synthesize strictly more state sequences.
We now argue that a polynomial space algorithm can find such a state that cannot be created by a smaller circuit size, allowing for $\puresPSPACE_0$ to ``diagonalize'' against any circuit of a fixed polynomial size (i.e., $\puresCircuit{n^k}_{0.49}$).

The following folklore result will be a critical subroutine of this algorithm.

\begin{restatable}{lemma}{apxpspace}
\label{lem:pspace-compute-td}
    Let $\rho$ and $\sigma$ be two $n$-qubit quantum states that are produced by \emph{unitary} quantum circuits of size at most $2^{\poly(n)}$ and space $\poly(n)$. Then $\tracedistance{\rho, \sigma}$ can be approximated to arbitrary $\exp(-\poly(n))$ error by a deterministic Turing machine in $\poly(n)$ space.
\end{restatable}
\begin{proof}
    See \cref{apx:td}.
\end{proof}

\begin{theorem}\label{thm:pspace-diagonalization}
    For every fixed $k \in \N$,
    \[
        \puresPSPACE_0 \not\subset \puresCircuit{n^k}_{0.49}.
    \]
\end{theorem}
\begin{proof}
    For a particular value of $n$, let $U_1, U_2, \dots$ be some arbitrary ordering on all \emph{unitary} circuits of size at most $s \coloneqq O(n^k)$ with $n$-qubits of outputs such that $\rho_i$ is the $n$-qubit state produced by $C_i$.
    Similarly, let $V_1, V_2, \cdots$ be some arbitrary ordering on all \emph{unitary} circuits of size $s^\prime$ on $n$-qubits \emph{with no ancilla}, for $s^\prime = \poly(n)$ that we will define later, and let $\ket{\phi_j}$ be produced by $V_j$.

    We now define the \emph{pure} state sequence $\{\ket{\psi_x}\}$ such that $\{\ket{\psi_x}\} \in \puresPSPACE_0$ but not $\puresCircuit{n^k}$.
    To trivially show that $\{\ket{\psi_x}\} \in \puresPSPACE_0$, we define its synthesizing circuits $(C_x)$ directly in terms of a $\PSPACE$ algorithm.
    The algorithm works as follows: For all $V_j$, iterate through all $U_i$ and estimate if $\tracedistance{\rho_i, \ketbra{\phi_j}{\phi_j}}$ to strictly less than $1/40 = 0.025$ accuracy using \cref{lem:pspace-compute-td}.
    The first time an estimate of $\tracedistance{\rho_i, \ketbra{\phi_j}{\phi_j}}$ is greater than $0.4925$, for all $x \in \{0, 1\}^n$ set $\ket{\psi_x} \coloneqq V_j \ket{x}$ such that $\ket{\psi_{0^n}}$ is produced by $V_j$ on the all-zeros state.
    Since both sequences of circuits are of polynomial size, they can be written down using $\poly(n)$ bits (see \cref{lem:num-circuits}) and therefore iterated through in $\poly(n)$ space.
    Combined with the fact that \cref{lem:pspace-compute-td} uses only $\poly(n, \log 41)$ space, the whole algorithm can be performed by a deterministic Turing machine in $\poly(n)$ space.
    We conclude that the state sequence $(\ket{\psi_x})_{x \in \{0, 1\}^*}$ generated by this procedure, should it terminate, always has an \emph{exact} synthesizing circuit that has $\poly(\abs{x})$ size, meaning that $(\ket{\psi_x}) \in \puresPSPACE_0$.
    
    We now show that this algorithm is not only guaranteed to terminate, but also guaranteed to produce a state that is more than $0.49$-far from any state in $\puresCircuit{n^k}$.
    \Cref{lem:circuit-size-hierarchy} shows that $\npack(\mixedstate^s, 0.495) < \npack(\purestate^{s^\prime}, 0.495)$ for some $r^\prime \coloneqq O\left(s^{12} n^3 c(\calG)\right)$ and sufficiently large $n$.
    There must then exist at least one $\ket{\phi_{j^*}} \in \purestate^{s^\prime}$ that is at least $0.495$-far from all $\rho_i \in \mixedstate^s$.
    As a result of the $\PSPACE$ algorithm estimating trace distance to accuracy $< 1/40$, if the algorithm reaches this $\ket{\phi_{j^*}}$ then by the triangle inequality it is guaranteed to terminate.
    
    We note that it is also possible for the algorithm to terminate earlier, but by the triangle inequality any state that could cause this must be greater than $0.49$-far from any state in $\puresCircuit{n^k}_0$, thus also succeeding.
\end{proof}

\section{Circuit Lower Bounds from Learning}\label{sec:final-proof}

\subsection{Learning vs Pseudorandomness}

We now formally prove that learning algorithms and pseudorandomness can be combined to give lower bounds for state synthesis.
We state \cref{lem:prs-and-learning-to-lowerbound} with as much generality as possible relative to the PRS, as we expect improvements to PRS constructions from \cref{cor:conditional-prs-2} to be possible.
This will allow the main theorems to be easily improved in the future.

\begin{remark}
    It is worth reminding the reader that the concept of problem size (i.e., the value of $n$) differs depending on the context.
    For a learning problem, for instance, $n$ is the number of qubits in the quantum state.
    Likewise, $n$ for a PRS is the security parameter and $n$ for a state sequence defined using that PRS is actually the key length $\kappa$.
    Throughout this work we have used $n$ to be consistent with the specific type of problem, but in \cref{thm:main} we will have to move between various ideas of what `$n$' means.
    Importantly, given a $(\kappa, \ell, q, s, \eps)$-\text{PRS}, the learning algorithm will run as a function of $\ell$.
    Therefore, if the PRS can be computed as a state sequence with some resource (such as time or space or size) growing as a function of the key length $f(\kappa)$, it is also computed with the same resource growing as $f \circ \kappa \circ \ell^{-1}$ relative to $\ell$, which is the viewpoint of the learning algorithm.
    Similarly, if the PRS has some value $f(n)$ computed relative to the security parameter $n$, then it will be $f \circ \ell^{-1}$ relative to $\ell$.
\end{remark}

\begin{lemma}\label{lem:prs-and-learning-to-lowerbound}
    For arbitrary fixed $f : \N \rightarrow \R^+$ and $\delta : \N \rightarrow  [0, 1]$, let $\mathfrak{C}$ be a circuit class that is closed under restrictions and define $\calC_\ell$ to be the set of pure states on $\ell$ qubits that can be constructed by $\mathfrak{C}\left[f(\ell)\right]$.
    Assume the existence of an infinitely-often $(\kappa, \ell, m, s, \eps)$-PRS against uniform quantum computations that can be computed in time $t$.
    If the concept class $\calC \coloneqq \bigcup_{\ell \geq 1} \calC_\ell$ is $\left(m \circ \ell^{-1},\,s \circ \ell^{-1},\,\eps \circ \ell^{-1} + \sqrt{m \circ \ell^{-1}} \cdot \delta\right)$-distinguishable then \[\mathsf{pureStateBQTIME}\left[t \circ \kappa^{-1} \right]_{\exp} \not\subset \mathsf{pureState}\mathfrak{C}\left[f \circ \ell \circ \kappa^{-1} \right]_\delta.\]
\end{lemma}
\begin{proof}
    It is easy to see that the PRS (as a state sequence) is in $\mathsf{pureStateBQTIME}\left[t \circ \kappa^{-1} \right]_{\exp}$, so we now need to show that it is not in $\mathsf{pureState}\mathfrak{C}\left[f \circ \ell \circ \kappa^{-1} \right]_\delta$.
    Relative to the security parameter $n$, the number of samples used by the distinguishing algorithm is $m \circ \ell^{-1} \circ \ell = m(n)$, the running time is $O\left(s \circ \ell^{-1} \circ \ell\right) = O(s(n))$, and the advantage is $\eps \circ \ell^{-1} \circ \ell + \sqrt{m \circ \ell^{-1} \circ \ell} \cdot \delta = \eps + \sqrt{m}\cdot \delta$.
    Finally, the size of the circuits generating $\calC_\ell$ are $O\left(f \circ \ell\right) = O\left(f \circ \ell \circ \kappa^{-1} \circ \kappa\right)$ such that the size relative to the key parameter $\kappa$ is $O\left(f \circ \ell \circ \kappa^{-1}\right)$.
    It follows by the parameters of the PRS and \cref{lem:robust-distinguish} that if $\calC$ could be learned then the PRS does not lie in $\mathsf{pureState}\mathfrak{C}\left[f \circ \ell \circ \kappa^{-1} \right]_\delta$.
\end{proof}

\subsection{Win-win Argument}

We now prove our main theorem of non-trivial quantum state learning (or even just distinguishing) implying state synthesis lower bounds, followed by a similar proof of decision problem circuit lower bounds.
As with many of the previous literature on learning-to-hardness for boolean concepts \cite{FORTNOW2009efficient,klivans2013constructing,oliveira2017conspiracies,arunachalam2022quantum}, we will use a win-win argument.
This will allow us to not have any complexity-theoretic assumptions, despite them being necessary in \cref{lem:conditional-prg,cor:conditional-prs-1}.

The first scenario is where $\puresPSPACE_{\exp} \subseteq \mathsf{pureStateBQSUBEXP}_{\exp}$.
Here, we do not even have to actually use the assumption of a non-trivial learner and just apply \cref{thm:pspace-diagonalization}.
In the other scenario, we combine \cref{cor:conditional-prs-2,lem:prs-and-learning-to-lowerbound}.

\begin{theorem}\label{thm:main}
    For arbitrary $\delta : \N \rightarrow  [0, 1]$, let $\mathfrak{C}$ be a circuit class that is closed under restrictions.
    There exists universal constants $\alpha \geq 1$ and $\lambda \in (0, 1/5)$ such that the following is true:
    
    Define $\calC_\ell$ to be the set of pure states on $\ell$ qubits that can be exactly constructed by $\mathfrak{C}\left[\poly(\ell)\right]$.
    For a fixed constant $c \geq 2$, if the concept class $\calC \coloneqq \bigcup_{\ell \geq 1} \calC_\ell$ is $\left(m, t, \eps\right)$-distinguishable for $m  \leq  2^{\ell^{0.99}}$, $t \leq O\left(2^{\ell^c}\right)$, and \[\eps \geq \frac{63\cdot 4^{\ell^{0.99}}}{2^\ell} + \frac{1}{2^{n^{\lambda}}} + \sqrt{m} \cdot \delta,\] then at least one of the following must be true: 
    \begin{itemize}
        \item for all $k \geq 1$, $\puresBQSUBEXP_{\exp} \not\subset \mathsf{pureStateBQSIZE}\left[n^{k}\right]_{0.49}$,
        \item $\puresBQE_{\exp} \not\subset \mathsf{pureState}\mathfrak{C}_\delta$.
    \end{itemize}
\end{theorem}
\begin{proof}
    One of two possibilities are true of the relationship between $\puresPSPACE_{\exp}$ and $\mathsf{pureStateBQSUBEXP}_{\exp}$, such that we will prove the separation for both possibilities.

    \sloppy
    In the case that $\puresPSPACE_{\exp} \subseteq \mathsf{pureStateBQSUBEXP}_{\exp}$ then \cref{thm:pspace-diagonalization} tells us that, for each $k \geq 1$, there exists some state sequence $(\ket{\psi_x})_{x \in \{0, 1\}^*}$ that is in $\puresPSPACE_0$ but not $\mathsf{pureStateBQSIZE}\!\left[n^{k}\right]_\delta$.
    Since $\puresPSPACE_0 \subseteq \mathsf{pureStateBQSUBEXP}_{\exp}$ by our assumption, $(\ket{\psi_x})_{x \in \{0, 1\}^*} \in \mathsf{pureStateBQSUBEXP}_{\exp}$ as well.
    This completes one side of the win-win argument.

    \fussy
    On the other hand, if $\puresPSPACE_{\exp} \not\subset \mathsf{pureStateBQSUBEXP}_{\exp}$ then \cref{cor:conditional-prs-2} tells us that there exists an infinitely-often
    \[\left(\kappa, \lfloor \log_2 r \rfloor^{2/c}, q, s, \frac{4q^2}{2^{\lfloor \log_2 r \rfloor^{2/c}}}+\frac{1}{r}\right)\text{-PRS}\] against uniform quantum computation
    that lies in $\mathsf{pureStateBQTIME}\left[2^{\kappa \circ \kappa^{-1}(n)}\right]_{\exp} = \puresBQE_{\exp}$ for some $\lambda \in (0, 1/5)$, $\alpha \geq 1$, $\kappa(n) \leq n^\alpha$, $r(n) = \lfloor 2^{n^\lambda} \rfloor$, $q = 2^{n^{1.98 \cdot \lambda / c}}$ and $s(n) = 2^{n^{2\lambda}}$.
    Observe that for $n \geq 1$: \[n^\lambda \geq \log_2 r \geq \lfloor \log_2 r \rfloor > \log_2 r - 1 = \log_2 \lfloor 2^{n^{\lambda}} \rfloor - 1 > \log_2\left(2^{n^\lambda} - 1\right) - 1 \geq n^\lambda - 2. \]
    Therefore, for sufficiently large $\ell \coloneqq \lfloor \log_2 r \rfloor^{2/c} = O\!\left(n^{2\lambda/c}\right)$:
    \begin{align*}
        m(\ell) = m\left(\lfloor \log_2 r \rfloor^{2/c}\right) &\leq  2^{\lfloor \log_2 r \rfloor^{1.98/c}} < 2^{n^{1.98 \cdot \lambda /c}} =  q  = q \circ \ell^{-1}(\ell)\\
        t(\ell) = t\left(\lfloor \log_2 r \rfloor^{2/c}\right)  &\leq  2^{\lfloor \log_2 r \rfloor^2} < 2^{n^{2\lambda}} = s = s \circ \ell^{-1}(\ell) \\
        \eps(\ell) = \eps\left(\lfloor \log_2 r \rfloor^{2/c}\right) &\geq \frac{63\cdot 4^{\lfloor \log_2 r \rfloor^{1.98/c}}}{2^{\lfloor \log_2 r \rfloor^{2/c}}} + \frac{1}{2^{n^\lambda}} + \sqrt{m} \cdot \delta > \frac{63\cdot 4^{\lfloor \log_2 r \rfloor^{1.98/c}}}{2^{\lfloor \log_2 r \rfloor^{2/c}}} + \frac{1}{2^{n^\lambda}} + \sqrt{m} \cdot \delta\\
        &> \frac{4\cdot 4^{n^{1.98 \cdot \lambda / c}}}{2^{\lfloor \log_2 r \rfloor^{2/c}}} + \frac{1}{2^{n^{\lambda}}} + \sqrt{m} \cdot \delta \geq \frac{4q^2}{2^{\lfloor \log_2 r \rfloor^{2/c}}} + \frac{1}{r} + \sqrt{m} \cdot \delta
    \end{align*}
    Consequently, because $\calC$ is $\left(m, t, \eps \right)$-distinguishable, then by \cref{lem:prs-and-learning-to-lowerbound} we find that \[\mathsf{pureStateBQTIME}\left[t \circ \kappa^{-1} \right]_{\exp} \not\subset \mathsf{pureState}\mathfrak{C}\left[\poly(\ell \circ \kappa^{-1}) \right]_\delta = \mathsf{pureState}\mathfrak{C}[\poly(n)]_\delta. \qedhere\]
\end{proof}

While the choice of parameters may seem somewhat opaque, the most important thing is to consider the relationship between $m$ and $\delta$ and how it affects $\eps$.
For example, when $m = \poly(n)$ then there exists some inverse-poly distinguishing that gives a separation with $\mathsf{pureState}\mathfrak{C}$.
Similarly, when $m$ is just less than $2^{\ell^{0.99}}$ then inverse-sub-exponential distinguishing gives a separation from $\mathsf{pureState}\mathfrak{C}_{\exp}$.
We formally prove the latter statement.

\begin{corollary}[Formal statement of \cref{thm:main-informal3}]\label{cor:main2}
    Let $\mathfrak{C}$ be a circuit class that is closed under restrictions.
    Define $\calC_\ell$ to be the set of pure states on $\ell$ qubits that can be exactly constructed by $\mathfrak{C}[\poly(\ell)]$.
    If there exists some constant $c \geq 2$ such that $\calC \coloneqq \bigcup_{\ell \geq 1} \calC_\ell$ is $\left(m, t, \eps\right)$-distinguishable for $m  \leq 2^{\ell^{0.99}}$, $t \leq O\!\left(2^{\ell^c}\right)$, and $\eps \geq \frac{1}{2^{\ell^{0.99}}}$, then at least one of the following must be true: 
    \begin{itemize}
        \item for all $k \geq 1$, $\puresBQSUBEXP_{\exp} \not\subset \mathsf{pureStateBQSIZE}\left[n^{k}\right]_{0.49}$,
        \item $\puresBQE_{\exp} \not\subset \mathsf{pureState}\mathfrak{C}_{\exp}$.
    \end{itemize}
\end{corollary}
\begin{proof}
    \sloppy
    For arbitrary polynomial $p$, \[\eps \geq \frac{1}{2^{\ell^{0.99}}} \geq \frac{63\cdot 4^{\ell^{0.99}}}{2^\ell} + \frac{1}{2^{\ell^{c/2}}} + \sqrt{ 2^{\ell^{0.99}}} \cdot \exp(-p) \geq \frac{63\cdot 4^{\ell^{0.99}} }{2^\ell} + \frac{1}{2^{n^\lambda}} + \sqrt{m} \cdot \exp(-p)\]
    with sufficiently large $\ell$.
    By \cref{thm:main}, either (1) for all $k \geq 1$, $\puresBQSUBEXP_{\exp} \not\subset \mathsf{pureStateBQSIZE}\left[n^{k}\right]_{0.49}$ or (2) $\puresBQE \not\subset \mathsf{pureState}\mathfrak{C}_{\exp(-p)}$.
    In case (2), since $p$ is an arbitrary polynomial, the state is not in $\mathsf{pureState}\mathfrak{C}_{\exp}$ as well.
\end{proof}

Using our connection between learning and distinguishing (see \cref{lem:learning-to-distinguishing}) we can now state our main result.
Note that instead of listing two possible outcomes like in \cref{cor:main2}, we take the intersection of the outcomes such that it is true regardless of which outcome.
However, while the implied result would still be interesting for many circuit classes, it is decidedly weaker than either of the two original possibilities.

\begin{corollary}[Formal statement of \cref{thm:main-informal}]\label{cor:main}
    Let $\mathfrak{C}$ be a circuit class that $\mathsf{pureState}\mathfrak{C}_{\exp} \subset \mathsf{pureStateBQP/poly}_{0.49}$.
    Define $\calC_\ell$ to be the set of pure states on $\ell$ qubits that can be exactly constructed by $\mathfrak{C}[\poly(\ell)]$.
    If there exists some constant $c \geq 2$ such that $\calC \coloneqq \bigcup_{\ell \geq 1} \calC_\ell$ is $\left(m, t, 1-\eta, 1-\gamma\right)$-learnable for $m - 1  \leq 2^{\ell^{0.99}}$, $t \leq O\!\left(\frac{2^{\ell^c}}{\ell^c}\right)$ and $\eta \cdot \gamma = \frac{4}{2^{\ell^{0.99}}}$, then for every $k \geq 1$, $\puresBQE_{\exp} \not\subset \mathsf{pureState}\mathfrak{C}[n^k]_{\exp}$.
\end{corollary}
\begin{proof}
    \Cref{lem:learning-to-distinguishing} tells us that each $\calC^k$ is $\left(m, t^\prime, \eps\right)$-distinguishable for $t^\prime = t \log t \leq O\!\left(2^{\ell^2}\right)$ and $\eps = \frac{1-o(1)}{2} \eta \cdot \gamma \geq \frac{1}{2^{\ell^{0.99}}}$ for sufficiently large $\ell$.
    We then appeal to \cref{cor:main2}.
    By invoking the condition that $\mathsf{pureState}\mathfrak{C}_{\exp} \subset \mathsf{pureStateBQP/poly}_{0.49}$, in both cases the following statement is true: for every $k \geq 1$, $\puresBQE_{\exp} \not\subset \mathsf{pureState}\mathfrak{C}[n^k]_{\exp}$.
\end{proof}

\begin{remark}
    It is actually possible to have a more fine-grained approach to the learning/distinguishing algorithm than in \cref{thm:main,cor:main2}.
    For instance, if the learning algorithm only holds up to $\mathfrak{C}[n^k]$ for some fixed $k$, then by a more careful application of \cref{lem:prs-and-learning-to-lowerbound}, the separation in the second case would be
    \[\mathsf{pureStateBQE}_{\exp} \not\subset \mathsf{pureState}\mathfrak{C}\left[n^{\alpha \cdot k / \lambda} \right]_\delta\]
    instead.
    As a result, \cref{cor:main} remains true if there is a learning algorithm for each $\calC^k \coloneqq \bigcup_{\ell \geq 1} \calC^k_\ell$ where $\calC^k_\ell$ refers to $\ell$-qubit states produced by $\mathfrak{C}[n^k]$.
    That is, the learning algorithm can work up to any polynomial size, but needs to know an upper bound on the polynomial in advance.
\end{remark}

\section*{Acknowledgements}
We would like to thank William Kretschmer, Sabee Grewal, Vishnu Iyer, Shih-Han Hung, Srinivasan Arunachalam, Henry Yuen, Scott Aaronson, Kai-Min Chung, and Ruizhe Zhang for a variety of extremely useful discussions and feedback. DL would also like to thank Igor Oliviera for helping fix bugs in the statement of the PRG in \cref{lem:conditional-prg}, Nick Hunter-Jones for helping with all of \cref{ssec:circuit-heirarchy}, and Gregory Rosenthal for helping with the name of $\sPSPACE$ and aiding in the understanding of the results in \cite{rosenthal2023efficient}.
Part of this research was performed while DL was visiting the Institute for Pure and Applied Mathematics (IPAM), which is supported by the National Science Foundation (Grant No. DMS-1925919).
DL is supported by the US NSF award FET-2243659.
FS is supported in part by the US NSF grants CCF-2054758 (CAREER) and CCF-2224131.  NHC is supported by NSF Awards FET-2243659 and FET-2339116 (CAREER), Google Scholar Award, and DOE Quantum Testbed Finder Award DE-SC0024301.

\bibliographystyle{alphaurl}
\bibliography{refs}

\appendix

\section{Decision Problem Circuit Lower Bounds With an Extra Circuit Constraint}

We now show the interesting result that non-trivial quantum state tomography can imply decision complexity class separations between entirely \emph{classical} computational models.
However, we will need to impose two assumptions on $\mathfrak{C}$.
The weaker assumption is the ability to implement the unitary $H^{\otimes (n+1)} \cdot I^{\otimes n} \otimes X$, which was used in \cref{lem:prf-to-prs} to construct binary phase states.
The second assumption is the ability to perform error reduction to arbitrary $\exp(-\poly(n))$ accuracy.
The simplest way to get error reduction is via majority and fanout gates, but since it's not clear that $\mathsf{QNC}^0$ or $\mathsf{QAC}^0$ can compute these gates, the weakest circuit class that we've introduced that \emph{does} meet these requirements is $\mathsf{QAC}^0_f$.
As such, we will state our results in terms of $\mathsf{QAC}^0_f$ for succinctness, though it should be understood that any circuit class that meets these two requirements suffices as well.

\begin{lemma}\label{lem:prf-and-learning-to-decision-lowerbound}
    For arbitrary fixed $f : \N \rightarrow \R^+$ and $\delta : \N \rightarrow  [0, 1]$, let $\mathfrak{C} \supseteq \mathsf{QAC}^0_f$ be a circuit class that is closed under restrictions and define $\calC_\ell$ to be the set of pure states on $\ell$ qubits that can be constructed by $\mathfrak{C}\left[f(\ell) + \ell\right]$ with depth at most $d+3$.
    Assume the existence of an infinitely-often $(\kappa, \ell, m, s, \eps)$-PRF against uniform quantum computations that can be computed in time $t$ by a deterministic Turing machine.
    If the concept class $\calC \coloneqq \bigcup_{\ell \geq 1} \calC_\ell$ is $\left(m \circ \ell^{-1},\,s \circ \ell^{-1},\,\eps \circ \ell^{-1} + \sqrt{m \circ \ell^{-1}} \cdot \delta\right)$-distinguishable then \[\mathsf{DTIME}\left[t \circ \kappa^{-1} \right] \not\subset \text{(depth $d$)-}\mathfrak{C}\left[f \circ \ell \circ \kappa^{-1} \right].\]
\end{lemma}
\begin{proof}
    Let $\left(\{F_k\}_{k \in \{0, 1\}^{\kappa(n)}}\right)_{n \in \N}$ be the the infinitely-often PRF against uniform quantum computations.
    Define the language $L$ such that inputs $(x, k) \in L$ if and only if $F_k(x) = 1$ when $x \in \{0, 1\}^\ell$ and $k \in \{0, 1\}^{\kappa}$.
    Because the PRF is computable in time $t$ by a deterministic Turing machine, $L$ lies in $\mathsf{DTIME}\left[t \circ \kappa^{-1} \right]$.
    
    We now need to show that $L$ is not in $\text{(depth $d$)-}\mathfrak{C}\left[f \circ \ell \circ \kappa^{-1}\right]$.
    For the sake of contradiction, assume that it was.
    Then using \cref{lem:prf-to-prs} and error-reduction via majority and fanout, $\text{(depth $d+3$)-}\mathfrak{C}\left[n + f \circ \ell \circ \kappa^{-1}(n)\right]$ can create pseudorandom states.
    The fact that the depth only increases by $3$ follows from the fact that $\mathfrak{C} \supset \mathsf{QAC}^0_f$ allows us to perform arbitrary classical fanout, approximately compute the PRF, then take the majority, to perform error reduction using only a depth increase of $2$.
    Performing the circuit in \cref{lem:prf-to-prs} adds the final extra layer of depth.
    Since this only needs $O(n)$ Hadamard and $X$ gates, the size increases by at most $O(n)$ as well.
    Relative to the security parameter $n$: the number of samples used by the distinguishing algorithm is $m \circ \ell^{-1} \circ \ell = m(n)$, the running time is $O\left(s \circ \ell^{-1} \circ \ell\right) = O(s(n))$, and the advantage is $\eps \circ \ell^{-1} \circ \ell + \sqrt{m \circ \ell^{-1} \circ \ell} \cdot \delta = \eps + \sqrt{m}\cdot \delta$.
    Finally, the size of the circuits generating $\calC_\ell$ are $O\left(f \circ \ell\right) = O\left(f \circ \ell \circ \kappa^{-1} \circ \kappa\right)$ such that the size relative to the key parameter $\kappa$ is $O\left(f \circ \ell \circ \kappa^{-1}\right)$.
    It follows by the parameters of the PRS and \cref{lem:robust-distinguish} that if $\calC$ could be learned then the PRS does not lie in $\mathsf{pureState}\mathfrak{C}\left[f \circ \ell \circ \kappa^{-1} \right]_\delta$.
    This is a contradiction, meaning that $L \not\in \text{(depth $d$)-}\mathfrak{C}\left[f \circ \ell \circ \kappa^{-1}\right]$.
\end{proof}

\begin{remark}
    Observe that \cref{lem:prf-and-learning-to-decision-lowerbound} requires a PRF while \cref{lem:prs-and-learning-to-lowerbound} only requires a PRS.
    Using \cite[Theorem 7.1]{rosenthal2023efficient}, if one only wants to show that $\mathsf{EXP} \not\subset \text{(depth $d + O(1)$)-}\mathfrak{C}$, it is actually possible to weaken \cref{lem:prf-and-learning-to-decision-lowerbound} to only use a PRS with the special property that each amplitude (including phase information) can be computed in time $\exp(\poly(n))$.
\end{remark}

We can now combine \cref{lem:prf-and-learning-to-decision-lowerbound} and \cref{cor:conditional-prf} to show a conditional circuit lower bound.
To handle the other side of the win-win-argument, we need the following result.

\begin{lemma}[\cite{chia_et_al:LIPIcs.ITCS.2022.47}]\label{lem:decision-diagonalization}
    For every $k \in \N$, there exist a language $L_k \in \PSPACE$ such that $L_k \not\in \mathsf{BQSIZE}[n^k]$.
\end{lemma}

\begin{theorem}[Formal Statement of \cref{thm:main-informal2}]\label{thm:decision-lower}
    Let $\mathfrak{C} \supseteq \mathsf{QAC}^0_f$ be a circuit class that is closed under restrictions.
    Define $\calC_\ell$ to be the set of pure states on $\ell$ qubits that can be exactly constructed by $\mathfrak{C}[\poly(\ell)]$ with depth at most $d+3$.
    If there exists a fixed constant $c \geq 2$ such that $\calC \coloneqq \bigcup_{\ell \geq 1} \calC_\ell$ is $\left(m, t, \eps\right)$-distinguishable for $m  \leq 2^{\ell^{0.99}}$, $t \leq O\left(2^{\ell^c}\right)$, and $\eps \geq \frac{1}{2^{\ell^{0.99}}}$, then at least one of the following must be true: 
    \begin{itemize}
        \item for all $k \geq 1$, $\mathsf{BQSUBEXP} \not\subset \mathsf{BQSIZE}\left[n^{k}\right]$,
        \item $\mathsf{E} \not\subset \text{(depth $d$)-}\mathfrak{C}[\poly(n)]$.
    \end{itemize}
\end{theorem}
\begin{proof}[Proof Sketch]
    We will again use a win-win argument, but now with $\mathsf{PSPACE}$ and $\mathsf{BQSUBEXP}$.
    If $\mathsf{PSPACE} \subseteq \mathsf{BQSUBEXP}$ then using \cref{lem:decision-diagonalization}, we get $\mathsf{BQE} \not\subset \mathsf{BQSIZE}[n^k]$.

    On the other hand, if $\mathsf{PSPACE} \not\subseteq \mathsf{BQSUBEXP}$ then we can invoke \cref{cor:conditional-prf} to show that there exists some infinitely-often $\left(\kappa, \ell, q, s, \eps\right)$-PRF against uniform quantum computations where $\kappa(n) \leq n^\alpha$, $r(n) = \lfloor 2^{n^\lambda} \rfloor$, $\ell \leq \lfloor \log_2 r \rfloor$, and $s(n) = 2^{n^{2\lambda}}$.
    By \cref{lem:prf-and-learning-to-decision-lowerbound} and a similar analysis to \cref{thm:main,cor:main2}, we find that $\mathsf{E} \not\subset \text{(depth $d$)-}\mathfrak{C}[\poly(n)]$.
\end{proof}

Because $\mathsf{TC}^0 \subset \mathsf{QAC}^0_f$ \cite{hoyer2005fan,takahashi2016collapse}, the second scenario in \cref{thm:decision-lower} also implies that $\mathsf{E} \not\subset \mathsf{TC}^0$.

\section{Conditional (Non-Adaptive) Pseudorandom Unitaries and Circuit Lower Bounds for Unitary Synthesis}

Due to the recent results on pseudorandom unitaries, we sketch how the ability to non-adaptively distinguish unitaries from Haar random (without access to the inverse) implies unitary synthesis separations.
Unitary synthesis capture an even more general set of problems than state synthesis, such as certain kinds of problems with quantum inputs and quantum outputs.
For details on unitary synthesis complexity class definitions, see \cite{metger2023pspace,bostanci2023unitary}.

We first need to introduce a distance measure between trace-preserving completely positive maps. Like \nameref{def:td}, it bounds the maximum distinguishability between two quantum operations.

\begin{definition}[Diamond Distance]\label{def:diamond-dist}
    Let $\Phi, \Gamma : \mixedstate_m \rightarrow \mixedstate_n$ be two trace-preserving completely positive maps.
    We define the diamond distance to be
    \[
        \diamonddistance{\Phi, \Gamma} \coloneqq \max_{\rho \in \mixedstate_{2m}} \lVert \left(\Phi \otimes 1_m\right) \rho - \left(\Gamma \otimes 1_m \right)  \rho \rVert_1
    \]
    where $1_m$ is the identity channel on $m$ qubits.
\end{definition}

Note that unlike \nameref{def:td}, diamond distance lies in $[0, 2]$.

We also need a version of \cref{cor:td-multi-copy} for trace-preserving completely positive maps.
\begin{lemma}\label{lem:dia-multi-copy}
    For trace-preserving completely positive maps $\Psi$ and $\Phi$, and $m \in \N$,
    \[
        \diamonddistance{\Psi^{\otimes m}, \Phi^{\otimes m}} \leq m \cdot \tracedistance{\ket{\psi}, \ket{\phi}}.
    \]
\end{lemma}
\begin{proof}
    Follows from the subadditivity of \nameref{def:diamond-dist} with respect to the tensor product.
\end{proof}

\begin{definition}[{$\uBQTIME\!\left[f\right]_\delta$}, {$\uBQSPACE\!\left[f\right]_\delta$}]\label{def:uniform-unitary-synthesis}
    Let $\delta : \mathbb{N} \rightarrow [0, 1]$ and $f : \mathbb{N} \rightarrow \mathbb{R}^+$ be functions. Then $\uBQTIME\!\left[f\right]_\delta$ (resp. $\uBQSPACE\!\left[f\right]_\delta$) is the class of all sequences of unitary matrices $(U_x)_{x \in \{0, 1\}^*}$ such that each $U_x$ is a unitary on $\poly(\abs{x})$ qubits, and there exists an $f$-time-uniform (resp. $f$-space-and-size-uniform) family of general quantum circuits $\left(C_x\right)_{x \in \{0, 1\}^*}$ such that for all sufficiently large input size $\abs{x}$,
    \[
        \diamonddistance{C_x, U_x} \leq \delta.
    \]
\end{definition}

\begin{definition}[$\uBQP_\delta$, $\uPSPACE_\delta$]\label{def:unitaryBQP}
    
    \begin{align*}
        \uBQP_\delta \coloneqq \bigcup_p \uBQTIME\!\left[p\right]_\delta
    \,\,\,\text{  and  }\,\,\,
    \uPSPACE_\delta \coloneqq \bigcup_p \uBQSPACE\!\left[p\right]_{\delta}
    \end{align*}
    where the union is over all polynomials $p : \mathbb{N} \rightarrow \mathbb{R}$.
\end{definition}

We can likewise define $\mathsf{unitaryBQE}_\delta$ and $\mathsf{unitaryBQP/poly}$ to be the analogues of $\mathsf{stateBQE}$ and $\mathsf{stateBQP/poly}$ respectively.
Finally, when dealing with \emph{unitary} circuit families, we will refer to the complexity classes with the prefix $\mathsf{pureUnitary}$-, such as in $\mathsf{pureUnitaryBQE}$.

We now need to generalize both \cref{lem:state-implies-decision,thm:pspace-diagonalization} but for unitary synthesis.
The high-level idea is that both proofs implicitly work at the level of unitaries, rather than states, in that they find a unitary that has the correct property before arguing that the state that the unitary creates also has a similar property.
Put another way, \cref{lem:state-implies-decision,thm:pspace-diagonalization} follow as corollaries of \cref{lem:unitary-implies-decision,thm:pspace-diagonalization-unitary}.

\begin{lemma}\label{lem:unitary-implies-decision}
    Let $k : \N \rightarrow \R^+$.
    For any $\mathsf{unitaryA}_\delta \subset \mathsf{unitaryBQP/poly}_\delta$, if $\mathsf{unitaryA}_\delta \not\subset \uBQTIME\left[k \cdot f\right]_{\delta + \exp(-k)}$ then $\mathsf{A} \not\subset \BQTIME\left[\frac{f}{n^{\nu}}\right]$ for some $\nu \geq 1$.
\end{lemma}
\begin{proof}[Proof Sketch]
    The proof works the same ways as \cref{lem:state-implies-decision} and using the same language.
    This is because in the proof of \cref{lem:state-implies-decision}, we argue that if $\mathsf{A} \subset \BQTIME\left[\frac{f}{n^{\nu}}\right]$ then the description of any unitary in $\mathsf{A}$ could be learned, presenting a contradiction.
\end{proof}

\begin{lemma}\label{thm:pspace-diagonalization-unitary}
    For every fixed $k \in \N$,
    \[
        \mathsf{pureUnitaryPSPACE}_0 \not\subset \mathsf{pureUnitaryBQSIZE}[n^k]_{0.98}.
    \]
\end{lemma}
\begin{proof}[Proof Sketch]
    The proof of \cref{thm:pspace-diagonalization} involves finding a unitary $V$ such that all unitaries $U_i$ that can be synthesized by circuits in $\mathsf{BQSIZE}[n^k]$ cannot create the same state as $V$ when acting on the all zeros state.
    In fact, no $U_i$ can create a state that is even $0.49$-close in \nameref{def:td}.
    It follows by the definition of \nameref{def:diamond-dist} that $\diamonddistance{U_i, V} > 0.98$ for all $U_i$.\footnote{Recall that \nameref{def:diamond-dist} lies in $[0, 2]$.}
    Since this holds for every $k$, we find that $\mathsf{pureUnitaryPSPACE}_0 \not\subset \mathsf{pureUnitaryBQSIZE}[n^k]_{0.98}$.
\end{proof}

\subsection{Pseudorandom Unitaries}

We start by adapting our pseudorandom objects from quantum states to unitaries, which was a notion also introduced by \cite{ji-pseudorandom-states2018}.
\cite{metger2024pseudorandom} recently showed how to build unitaries that are pseudorandom when the adversary is not allowed to adapt the algorithm based on prior measurement results.

\begin{definition}[PRU]\label{def:pru}
    Let $\kappa, \ell, m : \N \rightarrow \N$, let $s: \N \rightarrow \R^+$, and let $\eps : \N \rightarrow [0, 1]$.
    We say that a sequence of keyed pure unitaries $\left(\{U_k\}_{k \in \{0, 1\}^\kappa} \right)_{n \in \N}$ is an infinitely-often $(\kappa, \ell, m, s, \eps)$-PRU if for a uniformly random $k \in \{0, 1\}^{\kappa}$, no quantum algorithm running in time $s$ can distinguish $m$ queries of $U_k$ from $m$ queries to a Haar random unitary on $\ell$ qubits by at most $\eps$.
    Formally, for all $s$-time-uniform quantum oracle circuits $(C_n^{(\cdot)})_{n \in \N}$ that output the one qubit state $\rho_n^\calO$ when querying oracle $\calO$:
    \begin{align*}
        \bigg\lvert \E_{k \sim \{0, 1\}^{\kappa}} \trc\left[\ketbra{1}{1} \cdot \rho_n^{U_k} \right] -  \E_{\calU \sim \mu_{\mathrm{Haar}}} \trc\left[\ketbra{1}{1}  \cdot \rho_n^{\calU}\right] \bigg\rvert \leq \eps \,
\end{align*}
holds on infinitely many $n \in \mathbb{N}$.
\end{definition}

We now detail the construction used in \cite{metger2024pseudorandom}.
Denote $\calP_{\ell}: = \{p: \{0, 1\}^{\ell} \rightarrow \{0, 1\}^\ell\}$ as the set of all permutations on $\ell$-bits, and let $p^{-1}$ refer to the inverse permutation such that $p\circ p^{-1}$ is the identity function.
Furthermore, for $p \in \calP_\ell$ denote the $\ell$-qubit in-place permutation unitary to be $U_p \coloneqq \ket{x} \mapsto \ket{p(x)}$.

\begin{lemma}[Proof of {\cite[Theorem 3.1]{metger2024pseudorandom}}]\label{lem:non-adaptive-pru-td}
    Let $p \in \calP$ be a random permutation on $\ell$-qubits, $F$ be a $2^\ell \times 2^\ell$ diagonal matrix whose entries are random $\{\pm 1\}$ values (i.e., the diagonal is the truth table for $f \sim \mathfrak{F}_{\ell, 1}$), and $C$ be a random Clifford circuit on $\ell$ qubits. If $\calU$ is an $\ell$-qubit Haar random unitary, then $\left(U_pFC\right)^{\otimes m}$ is $O\!\left(\frac{m}{\sqrt{2^n}}\right)$-close to $\calU^{\otimes m}$ in expected \nameref{def:diamond-dist}.
\end{lemma}

Because of these random permutations, we will not introduce the idea of a pseudorandom permutation.

\begin{definition}[PRP]\label{def:prp}
    Let $\kappa, \ell: \N \rightarrow \N$, let $q, s: \N \rightarrow \R^+$, and let $\eps : \N \rightarrow [0, 1]$.
    We say that a sequence of keyed-permutations $\left(\{P_k \in \mathfrak{F}_{\ell, \ell}\}_{k \in \{0, 1\}^\kappa}\right)_{n \in \N}$ is an infinitely-often $(\kappa, \ell, q, s, \eps)$-PRP if for a uniformly random $k \in \{0, 1\}^{\kappa}$, no quantum algorithm running in time $s$ can distinguish black-box access to $\calO_{P_k}$ from black-box access to $\calO_p$ for random \emph{permutation} $p \in \calP_\ell$ using at most $q$ queries by at most $\eps$.
    Formally, for all $s$-time-uniform oracle quantum circuits $(C^{(\cdot)}_n)_{n \in \N}$ such that each $C^{\calO}_{n}$ takes no inputs, queries $\calO$ at most $q$ times, and outputs a single qubit state $\rho_{n}^\calO$:
    \begin{align*}
        \bigg\lvert \E_{k \sim \{0, 1\}^{\kappa}} \trc\left[\ketbra{1}{1} \cdot \rho_{n}^{\calO_{P_k}} \right] -  \E_{p \sim \mathcal{P}_{\ell}}\trc\left[\ketbra{1}{1} \cdot \rho_{n}^{\calO_f}\right]\bigg\rvert \leq \eps \, 
\end{align*}
holds on infinitely many $n \in \mathbb{N}$.
\end{definition}

\begin{remark}\label{rem:cca-vs-pca}
The above definition, where the quantum adversary can only make XOR queries for the forward direction of the permutation, is what is known as quantum chosen-plaintext attack (qCPA) secure.
However, for $p \in \calP_\ell$ note that queries to XOR oracle $\calO_p$ are not the same as queries to the in-place permutation $U_p$.
To achieve security against adversaries with queries to $U_p$ like we will need to invoke \cref{lem:non-adaptive-pru-td}, it suffices to look at an even stronger notion of security called quantum chosen-ciphertext attack (qCCA) secure, where the adversary has access to $\calO_{p^{-1}}$ as well.
This is because one query to $\calO_p$ and $\calO_{p^{-1}}$ each can be used to simulate $U_p$ (see proof of \cref{cor:conditional-pru}).
\end{remark}

\begin{lemma}[Generalization of proofs of \cite{ristenpart2013mix,morris2014sometimes,zhandry2016note}]\label{lem:cca-security}
    For $r = \Theta(n^2)$, let $a_{1}, \dots, a_r$ be random $n$-bit strings and $f_1, \dots, f_r$ be random functions from $\mathfrak{F}_{n, 1}$.
    Using these resources, there is an $\Theta(n^2)$-time construction of permutations $p$ and $p^{-1}$ that are distinguishable from a true random permutation with advantage at most $\frac{1}{2^n}$, even when the adversary is computationally unbounded and has learned the entire truth table.
\end{lemma}

Note that \cref{lem:cca-security} is extremely powerful, as the adversary is allowed to know everything about the permutation, as well as be computationally unbounded.
Because of this combination, even quantum queries to the inverse can be simulated, showing that the construction is qCCA-secure.\footnote{This observation is the central point of \cite{zhandry2016note}, as the original analyses in \cite{ristenpart2013mix,morris2014sometimes} were only for classical security.}

\begin{lemma}\label{lem:prg-to-prp}
    Let $G$ be an infinitely-often $(\kappa, m, s, \eps)$-PRG against uniform quantum computations that is computable in time $t$ by a deterministic Turing machine. Then for $\ell = o\!\left(\log \frac{m}{\log^2 m}\right)$, there exists an infinitely-often qCCA-secure \[\left(\kappa,\, \ell,\, q,\, s-O(q\cdot m \cdot \ell^2),\, \eps + \frac{1}{2^\ell} \right)\text{-PRP}\] against uniform quantum computations that can be computed in time $O(t + m \cdot \ell^2)$.
\end{lemma}
\begin{proof}
    We will simply replace the random strings and random functions in \cref{lem:cca-security} with pseudorandom strings and pseudorandom truth tables from the output of the PRG.

    First we show that we have enough pseudorandom bits to do so.
    Let $r = \Theta(\ell^2)$ be the parameter from \cref{lem:cca-security}.
    We start by setting aside $r \cdot \ell$ bits for the $a_i$.
    We then need $r \cdot 2^\ell$ bits for the truth table of the $f_i$.
    This means that we need at most $r \cdot 2^{\ell + 1} = \Theta\!\left(\ell^2 \cdot 2^\ell\right)$ bits, such that if $\ell = o\!\left(\log \frac{m}{\log^2 m} \right)$ then $\Theta(\ell^2 \cdot 2^\ell) \leq m$ for sufficiently large $n$.

    As for time complexity, getting all of the pseudorandom bits takes $O(m)$ time, for $r = \Theta(\ell^2)$ rounds.
    This makes the total time complexity simply $\Theta(m \cdot \ell^2)$.

    Finally, for security we use two hybrids: first by replacing the pseudorandom strings and functions with truly random strings and functions, and then finally comparing this to a truly random permutation.
    By the reverse triangle inequality, and \cref{lem:cca-security}, an advantage of $\eps + \frac{1}{2^\ell}$ would suffice to distinguish the underlying PRG with advantage at least $\eps$.
\end{proof}

\begin{lemma}\label{lem:prg-to-pru}
    Let $G$ be an infinitely-often $(\kappa, m, s, \eps)$-PRG against uniform quantum computations that is computable in time $t$ by a deterministic Turing machine. Then for $\ell = o\!\left(\log \frac{m}{\log^2 m}\right)$, there exists an infinitely-often non-adaptive qPCA-secure \[\left(\kappa,\, \ell,\, q,\, s-O(q\cdot m \cdot \ell^2),\, \eps + O\!\left(\frac{q}{\sqrt{2^\ell}}\right) \right)\text{-PRU}\] against uniform quantum computations that can be computed in time $O(t + m \cdot \ell^2)$.
\end{lemma}
\begin{proof}
    We split up the output of the PRG to get two different infinitely-often $(\kappa, m/2, s, \eps)$-PRG.
    Note that this ensures that the two PRGs are hard instances at the same time, as opposed to two independent infinitely-often PRGs that may only be hard on different input lengths.
    Use \cref{lem:prg-to-prp} to create a pseudorandom permutation $p$ from one of the PRGs and \cref{lem:prg-to-prf} to create pseudorandom function $f$ using the other PRG.
    Because the output of the original PRG was halved, this only additively decreases the maximum value of $\ell$ (i.e., the number of qubits $p$ and $f$ act on) by at most $1$.
    
    We now need to create a non-adaptive PRU as per \cref{lem:non-adaptive-pru-td}.
    To turn $\calO_p$ into $U_p$, we can use the fact that $\calO_{p^{-1}}$ is also efficiently computable to perform
    \[
        \calO_{p^{-1}} \cdot \mathrm{SWAP}_{AB} \cdot \calO_p \ket{x}_A\ket{0}_B = \calO_{p^{-1}} \cdot \mathrm{SWAP}_{AB} \ket{x}_A{\ket{p(x)}}_B = \calO_{p^{-1}}\ket{p(x)}_A \ket{x}_B = \ket{p(x)}_A\ket{0}_B.
    \]
    Likewise, $f$ can be efficiently made into the diagonal of a matrix by the use of Hadamards.\footnote{This is the standard XOR-to-phase oracle conversion \cite[Lecture 17]{aaronson_qis}}

    By the reverse Triangle inequality, \cref{def:diamond-dist,fact:td-to-tv,lem:prg-to-prf,lem:prg-to-prp,lem:non-adaptive-pru-td}, the allowed advantage is $\eps + \frac{1}{2^\ell} +  O\left(\frac{q}{\sqrt{2^\ell}}\right) = \eps +  O\left(\frac{q}{\sqrt{2^\ell}}\right)$.
    Since the infinitely-often PRP $p$ and infinitely-often PRF $f$ are both hard on the smae instances an infinite number of times, the resulting construction is an infinitely-often PRU.
    
    Sampling a Clifford circuit on $\ell$-qubits takes $O(\ell^2)$ time \cite{berg2021simple}, and creating a phase unitary from an XOR oracle takes $O(m + \ell)$ time given the output of the PRG. However, creating the permutation from \cref{lem:prg-to-prp} takes $O(m \cdot \ell^2)$ time, so querying the PRU takes time $O\left(m \cdot \ell^2\right)$ time given access to the output of the PRG.
    Therefore, the PRU is secure against $s - O\left(q \cdot m \cdot \ell^2\right)$-time adversaries.
    The key length becomes $\kappa^\prime \coloneqq \kappa + O(\ell^2) \leq n^{\alpha^\prime}$ for some $\alpha^\prime \geq \alpha \geq 1$, to account for the sampling of the random Clifford circuit.
\end{proof}

\begin{corollary}\label{cor:conditional-pru}
     Suppose there exists a $\gamma > 0$ such that $\mathsf{PSPACE} \not\subset \mathsf{BQTIME}\left[2^{n^\gamma}\right]$.
    Then, for some choice of constants $\alpha \geq 1$, and $\lambda \in (0, 1/5)$ and sufficiently large $n \in \N$, there exists an infinitely-often non-adaptive qPCA-secure \[\left(\kappa, \ell, q, s, \frac{1}{r} + O\left(\frac{q}{\sqrt{2^\ell}}\right)\right)\text{-PRU}\] against uniform quantum computations where $\kappa \leq n^\alpha$, $r = \lfloor 2^{n^\lambda} \rfloor$, $\ell = o\!\left(\log \frac{r}{\log r}\right)$, $s = 2^{n^{2\lambda}}$, and $m = o\left(\frac{s}{r \cdot \ell^2}\right)$.
    In addition, the PRU lies in $\mathsf{pureUnitaryBQE}_{\exp}$.
\end{corollary}
\begin{proof}
    Combine \cref{lem:conditional-prg,lem:prg-to-pru}.
\end{proof}

\subsection{Unitary Distinguishing Implies Non-Uniform Unitary Synthesis Lower Bounds}

Informally, define $(m, t, \eps)$-distinguishing of unitaries from Haar random unitaries to be analogous to distinguishing states from Haar random states (\cref{def:distinguish}).
For our purposes, the distinguisher will not have access to the inverse unitary as we do not have proven security for this learning model using \cref{lem:non-adaptive-pru-td}.

\begin{lemma}\label{lem:pru-and-learning-to-lowerbound}
    For arbitrary fixed $f : \N \rightarrow \R^+$ and $\delta : \N \rightarrow  [0, 1]$, let $\mathfrak{C}$ be a circuit class that is closed under restrictions and define $\calC_\ell$ to be the set of pure states on $\ell$ qubits that can be constructed by $\mathfrak{C}\left[f(\ell)\right]$.
    Assume the existence of an infinitely-often (resp. non-adaptive) $(\kappa, \ell, q, s, \eps)$-PRU against uniform quantum computations that can be computed in time $t$.
    If the concept class $\calC \coloneqq \bigcup_{\ell \geq 1} \calC_\ell$ is $\left(q \circ \ell^{-1},\,s \circ \ell^{-1},\,\eps \circ \ell^{-1} + m \circ \ell^{-1} \cdot \delta\right)$-distinguishable (resp. using non-adaptive queries) then \[\mathsf{pureUnitaryBQTIME}\left[t \circ \kappa^{-1} \right]_{\exp} \not\subset \mathsf{pureUnitary}\mathfrak{C}\left[f \circ \ell \circ \kappa^{-1} \right]_\delta.\]
\end{lemma}
\begin{proof}[Proof Sketch]
    The proof follows the same way as in \cref{lem:prs-and-learning-to-lowerbound}, except using \cref{lem:dia-multi-copy} in place of \cref{cor:td-multi-copy}.
    This change from $\sqrt{m}$ to simply $m$ does not affect the proof as it is still killed by the $\delta = \exp(-p)$ term for arbitrary polynomial $p$.
\end{proof}

\begin{theorem}[Formal Statement of \cref{thm:main-unitary-informal}]\label{thm:main-unitary}
    Let $\mathfrak{C}$ be a circuit class that is closed under restrictions.
    Define $\calC_\ell$ to be the set of unitaries on $\ell$ qubits that can be exactly constructed by $\mathfrak{C}[\poly(\ell)]$.
    If there exists a fixed constant $c \geq 2$ such that $\calC \coloneqq \bigcup_{\ell \geq 1} \calC_\ell$ is $\left(m, t, \eps\right)$-distinguishable for $m  \leq 2^{\ell^{0.99}}$, $t \leq O\left(2^{\ell^c}\right)$, and $\eps \geq \omega\left(\frac{1}{2^{\ell^{0.99}}}\right)$ by a non-adaptive algorithm without access to the inverse, then at least one of the following is true:
    \begin{itemize}
        \item for every $k \geq 1$: $\mathsf{pureUnitaryBQSUBEXP}_{\exp} \not\subset \mathsf{pureUnitaryBQSIZE}[n^k]_{0.49}$,
        \item $\mathsf{pureUnitaryBQE}_{\exp} \not\subset \mathsf{pureUnitary}\mathfrak{C}_{\exp}$.
    \end{itemize}
\end{theorem}
\begin{proof}[Proof Sketch]
    We instantiate the win-win argument in the same way as \cref{thm:main,cor:main2}, but replacing the PRS with the PRU from \cref{cor:conditional-pru} and with the cases based on the relationship between $\mathsf{pureUnitaryPSPACESIZE}_{\exp}$ and $\mathsf{pureUnitaryBQSUBEXP}_{\exp}$.
    When $\mathsf{pureUnitaryPSPACESIZE}_{\exp} \subset \mathsf{pureUnitaryBQSUBEXP}_{\exp}$ then we apply \cref{thm:pspace-diagonalization-unitary}.
    In contrast, when $\mathsf{pureUnitaryPSPACESIZE}_{\exp} \not\subset \mathsf{pureUnitaryBQSUBEXP}_{\exp}$ we use \cref{lem:unitary-implies-decision} to show that $\mathsf{PSPACE} \not\subset \mathsf{BQSUBEXP}$.
    This allows us to use \cref{cor:conditional-pru,lem:pru-and-learning-to-lowerbound} to complete the proof using a similar analysis to \cref{thm:main,cor:main2}
\end{proof}

\begin{remark}
Distinguishing copies of a fixed quantum state from Haar random actually follows as a special case of algorithms that distinguish unitaries from Haar random.
This makes distinguishing with unitary query access strictly easier, even when restricted to non-adaptivity and no access to the inverse.
It intuitively follows that the resulting separation is weaker, as state synthesis separations imply unitary synthesis separations.  
\end{remark}

\begin{remark}
If \cref{lem:non-adaptive-pru-td} were replaced with a sufficiently efficient and secure PRU construction that allowed adaptivity then the distinguishing algorithm in \cref{thm:main-unitary} would also be allowed adaptivity.
Likewise, if the PRU was made secure against access to the inverse unitary then the learning algorithms in \cref{thm:main-unitary} would also be allowed access to the inverse unitary.
\end{remark}

\begin{remark}
We would ideally like to generalize \cref{thm:main-unitary} to a version like \cref{cor:main} that deals with \emph{learning}, as in quantum process tomography, rather than distinguishing.
A na\"ive approach, assuming access to the inverse, would be to replace the SWAP test in \cref{lem:learning-to-distinguishing} with the identity property test \cite{montanaro2016survey}.
This simply involves measuring the Choi state of $\hat{C} \cdot C^{\dagger}$ in the Bell basis and seeing if it returns the canonical Bell state.
As noted by \cite[Section 5.1.1]{montanaro2016survey}, identity testing to even a constant distance in operator norm requires $\Omega\!\left(\sqrt{2^n}\right)$ queries to the unknown quantum circuit.
Instead, these tests work well in an average-case sense, such as when \[d_{\text{avg}}(U, V) \coloneqq \Ex_{\ket {\psi} \sim \mu_{\text{Haar}}}\left[\tracedistance{U \ket{\psi}, V \ket{\psi}}\right]\] is small.
Luckily, this soundness lowerbound does not directly apply, as we only need to have soundness against Haar random unitaries, which is easier than dealing with adversarially chosen unitaries.
We also need our test to be \emph{tolerant}, such that when $\diamonddistance{U, V}$ (or $d_{\text{avg}}$) is small, rather than just $0$, then the test accepts with high probability, giving us the ability to distinguish the two cases.
We leave this analysis (and/or analysis of other candidate tests and learning-to-decision reductions) as an open question for future work.

We note that learning with respect to the uniform distribution over computational basis state is the exact setting of \cite[Theorem 3.1]{arunachalam2022quantum} and uses natural properties to distinguish the unitary from a unitary implementing the XOR query of a random Boolean function.
\cite{zhao2023learning} likewise uses learning with respect to $d_{\text{avg}}$ to distinguish against random Boolean functions.
Both results are (to the authors' knowledge) not comparable to distinguishing against Haar random unitaries.
\end{remark}

\begin{remark}\label{remark:unitary-decision}
    Because of the additional complexities of constructing the PRU/PRP in \cref{cor:conditional-pru}, it is not clear how to get a version of \cref{thm:decision-lower} for unitary distinguishing by building a PRU.
    Specifically, due to the recursive nature of the constructions in \cite{ristenpart2013mix,morris2014sometimes}, na\"ively implementing the PRP would increase the circuit depth by an additive factor of $\Theta(n^2)$.
    This would seem to disqualify all of $\mathsf{QAC}_f$ and any other shallow-depth circuit classes for instance.
    Ideally, we would have have been able to use a simpler construction of pseudorandom permutations, such as the Ruby-Lackoff constructions based on balanced Feistel networks with constant rounds, but there is currently no known proof of security against quantum adversaries with access to the in-place permutation oracle for these constructions.
    While it was shown that 4-round Ruby-Lackoff is qCPA-secure with the XOR oracle\,\cite{Hosoyamada20194RoundLC}, 4-rounds was also shown to not be qCCA-secure \cite{ito2019chosen} (see \cref{rem:cca-vs-pca} for why this matters).\footnote{Based on \cite[Footnote 3]{unruh2023towards} there is reason to believe that even the qCPA-security of 4-round Ruby-Lackoff is not clear.}
    Because Clifford circuits fall under $\mathsf{QAC}^0_f$ \cite{aaronson2004simulation,rosenthal2023efficient}, if $O(1)$-round Ruby-Lackoff \emph{was} shown to be provably qCCA-secure then the depth increase from implementing the PRU would only be a constant and allow us to make a similar statement to \cref{thm:decision-lower} for process tomography.
    Additionally, as we learn more about efficient construction of PRUs \cite{lu2023quantum,metger2024pseudorandom,chen2024efficient}, the need for a pseudorandom permutation might be removed altogether.

    An alternative approach to showing circuit lower bounds for decision problems via learning algorithms would be to generalize the learning-to-distinguisher in \cite[Theorem 19]{zhao2023learning} to adversaries with sub-exponential number of queries allowed, as the stated result only applies to adversaries with a polynomial number of allowed queries.
    By substituting this reduction in place of \cite[Theorem 3.1]{arunachalam2022quantum}, one would recover a result showing that learning with respect to $d_{\text{avg}}$ (i.e., squared loss under the Haar measure, as opposed to uniform over computational basis states like in \cite{arunachalam2022quantum}) implies circuit lower bounds for decision problems.
\end{remark}

\section{Approximating Trace Distance in Polynomial Space}\label{apx:td}

We combine the ideas behind the proof sketch of \cite[Corollary 10 and Proposition 11]{watrous2002quantum} with the proof that $\mathsf{BQP} \subseteq \PSPACE$ \cite{bernstein1997quantum} to show that the trace distance of states produced by $\poly(n)$ size general quantum circuits can be computed using only $\poly(n)$ space by a deterministic Turing machine (i.e., \cref{lem:pspace-compute-td}).

To start, we show the folklore result(s) that, when represented as a matrix, any state $\rho$ that is produced by a general quantum circuit of size $\poly(n)$ can have its entries be approximated to arbitrary inverse exponential precision in $\poly(n)$ space \cite[Section 4.5.5]{nielsen2002quantum}.\footnote{In the language of \cite[Definition 5.1]{rosenthal2023efficient}, these would be considered $\mathsf{polyL}\text{-explicit}$.}
To do so, we first start with applying the unitary (\cref{lem:pspace-compute-amplitude}), then tracing out the ancilla qubits (\cref{cor:pspace-compute-reduced}).

For ease of notation, we will say that a value can be computed in $\poly(n)$ space if the value can be approximated to arbitrary $\exp(-\poly(n))$ accuracy in $\poly(n)$ space.
As long as the values are bounded in $[-1, 1]$ and there are not a $\omega(\exp(\poly(n)))$ many arithmetic operations, then the triangle inequality and Cauchy-Schwarz ensures us that we can act as if there are no errors, by simply decreasing the error of each value accordingly.

\begin{lemma}[Folklore]\label{lem:pspace-compute-amplitude}
    Given \emph{unitary} quantum circuit $C$ of size $s \leq 2^{\poly(n)}$ and space $m = \poly(n)$ and a quantum state $\rho$ on at most $m$ qubits whose entries (i.e., $\rho_{ij} \coloneqq \braket{i | \rho | j})$ can computed in $\poly(n)$ space, the entries of $C\rho C^\dagger$ can also be computed in $\poly(n)$ space.
\end{lemma}
\begin{proof}
    The proof idea is similar to the one showing that $\mathsf{BQP} \subseteq \PSPACE$.
    Let $C$ be broken up into its elementary gates as $C\coloneqq C_1  C_2 \dots C_s$.
    Since $I = \sum_{y \in \{0, 1\}^m} \ketbra{y}{y}$,
    we can rewrite the output expression as:
    \begin{align*}
        \braket{i | C \rho C^\dagger | j} &= \braket{i | C_1 C_2 \dots C_s \rho C_s^\dagger \dots C_2^\dagger C_1^\dagger| j}\\
        &= \sum_{y_1, y_2, \dots, y_{2s} \in \{0, 1\}^m}\braket{i | C_1 | y_1}\braket{y_1 | C_2 | y_2} \dots \braket{y_{s-1} | C_s | y_s}\braket{y_s | \rho | y_{s+1}} \dots \braket{y_{2s} | C_1^\dagger | j}.
    \end{align*}
    Since each $C_i \in \{H, T, \mathrm{CNOT}\}$, each $\braket{y_{i} | C_{i+1} | y_{i+1}}$ must lie in the set $\{0, \pm \frac{1}{\sqrt{2}}, 1, i\}$.
    Therefore, each $\braket{y_{i} | C_{i+1} | y_{i+1}}$ can be computed exactly using $\poly(m)$ space.\footnote{If we weren't using the $\{H, \mathrm{CNOT}, T\}$ gate set, then we can just use the \nameref{lem:solovay-kitaev} to approximate it with the $\{H, \mathrm{CNOT}, T\}$ gate set to sufficient error. This will not affect the amount of space used by more than a polynomial.}
    Furthermore, we observe that each term in the above summation is the product of $2s$ many $\braket{y_{i} | C_{i+1} | y_{i+1}}$ multiplied by $\braket{y_s | \rho | y_{s+1}}$, such that the product either has the form $\braket{y_s | \rho | y_{s+1}} \frac{i^\ell}{\sqrt{2}^k}$ for $\ell \in \{0, 1, 2, 3\}$ and $k \in \{0, \dots, 2s\}$ or is just zero.
    Either way, it can be computed using $\poly(n, m, \log s)$ space.
    It follows that the entire summation can be approximated to accuracy $\exp(-k)$ in $\poly(n, m, \log s, k) = \poly(n, k)$ space by approximating up to $\poly(k)$-bits of precision.
\end{proof}

\begin{corollary}[Folklore]\label{cor:pspace-compute-reduced}
    Given a \emph{unitary} quantum circuit $C$ of size $s \leq 2^{\poly(n)}$ and space $m = \poly(n)$ and let $\rho$ be the output of $C$. The $(i, j)$-th entry of $\rho$ (i.e., $\rho_{ij} \coloneqq \braket{i | \rho | j}$) can be computed in $\poly(n)$ space.
\end{corollary}
\begin{proof}
    Any unitary circuit can be decomposed into the unitary stage and then the tracing out stage at the end.
    Since it is clear that the entries of our starting state, $\ketbra{0\dots0}{0 \dots 0}$ can be computed in $\poly(n)$ space, it follows from \cref{lem:pspace-compute-amplitude} that after applying the unitary to the all-zeros state, the new quantum state's entries can be computed in $\poly(n)$ space.
    We then need to figure out the effect of tracing out at most $\poly(n)$-qubits.

    We can assume WLOG that we trace out the final qubit(s) because we use SWAP gates at the end of the unitary to move the qubits that will be traced out to the end.
    If we are tracing out $k = \poly(n)$ qubits, we can then express the new entries to compute as:
    \begin{align*}
        \braket{i| C \rho C^\dagger | j}
        &= \braket{i |\left(\sum_{x \in \{0, 1\}} \left(I^{\otimes m-k} \otimes \bra{x}\right) \rho  \left(I^{\otimes m-k} \otimes \ket{x}\right) \right)| j}\\
        &= \sum_{x \in \{0, 1\}^k} \braket{i, x | \rho |j, x},\\
    \end{align*}
    which is just the sum of $\exp(\poly(n))$-many things that we can compute in $\poly(n)$ space.
\end{proof}

We now move onto the problem of approximating trace distance.
To compute the trace distance between these two states, we need to find the sum of the absolute value of its eigenvalues.
The following two results will allow us to find the roots of the characteristic polynomial (i.e., the eigenvalues) of $\rho - \sigma$.

\begin{lemma}[\cite{Ben-Or1998fast,neff1994specified}]\label{lem:compute-real-roots}
    Given a polynomial $p(z)$ of degree $d$ with $m$-bit coefficients and an integer $\mu$, the problem of determining all its roots with error less than $2^{-\mu}$ is considered.
    It is shown that this problem can be solved by a 
    $\polylog\left( d + m + \mu\right)$-space-uniform \emph{Boolean} circuit of depth at most $\polylog\left(d + m + \mu\right)$ on $\poly\log(d + m + \mu)$ many bits if $p(z)$ has all real roots.
\end{lemma}

\begin{lemma}[{\cite[Lemma 1]{borodin1977relating}}]\label{lem:size-to-space-simulate}
    A \emph{space-uniform} Boolean circuit of size $s$ and depth $d$ can be simulated by a deterministic Turing machine in space at most $d + \log(s)$.
\end{lemma}

By combining \cref{lem:compute-real-roots,lem:size-to-space-simulate}, we get that the roots of certain polynomials can be approximate to high precision in $\poly(n)$ space. 

We now have the ingredients necessary to show that the trace distance of states produced by $O\!\left(2^{\poly(n)}\right)$-size and $\poly(n)$-space unitary quantum circuits can also be approximated to high precision in $\poly(n)$ space. 

\apxpspace*
\begin{proof}
    By \cref{cor:pspace-compute-reduced}, each entry of both $\rho$ and $\sigma$ can be approximated to high accuracy using $\poly(n)$ space, so each entry of $\rho - \sigma$ can also be well approximated.
    From here, let $f(x)$ be the characteristic polynomial of $\rho - \sigma$.
    It follows that any coefficient of $f(x)$ can be highly approximated using $\poly(n)$ space as well, via the Faddeev-Leverrier algorithm \footnote{In principle this approach is not numerically stable \cite{rehman2011la,wilkinson1988algebraic}, such that a more numerically stable approach would be more space efficient. This is because far less bits of precision would have to be used in previous steps to compensate for the numerical instability.} \cite[Corollary 2]{csanky1975fast}.

    We note that the characteristic polynomial of $\rho - \sigma$ has a degree $2^n$ with all real roots.
    It follows from \cref{lem:compute-real-roots,lem:size-to-space-simulate} that $\tracedistance{\rho, \sigma} \coloneqq \frac{1}{2}\sum_i \abs{\lambda_i}$ can be computed by a deterministic Turing machine in $\poly(n)$ space.
\end{proof}

\section{Trivial Learners}

Observe that in \cref{cor:main2,thm:decision-lower}, there are three important parameters of the algorithm that all must simultaneously meet some condition.
In this section, we highlight how satisfying any two of these conditions is trivially easy for states produced by polynomial-size unitary quantum circuits.
Clearly, if samples are not a concern then full pure state tomography can be done in $2^{O(n)}$ time to constant advantage \cite{franca_et_al:LIPIcs.TQC.2021.7}.
If one instead wants to satisfy the advantage and sample requirements, for instance, then one can run an exhaustive search using classical shadows \cite{HKP20-classical-shadows,buadescu2021improved,zhao2023learning} over all unitary circuits of size at most $n^{\omega(1)}$.
This will only use $O(n^{\omega(1)})$ samples, but the time complexity will be $2^{\omega(\poly(n))}$.

Finally, if only $O\!\left(\frac{1}{2^{n}}\right)$ advantage is desired then no measurements even have to be taken.
Specifically, by outputting a random stabilizer state, every concept class can be
\[
\left(0, n^2, 1-\frac{\theta}{2^{n+1}}, (1-\theta)^2\left(\frac{1}{2} + \frac{1}{2^{n+1}}\right)\right)\text{-learned}
\]
for arbitrary $\theta \in [0, 1]$.
In fact, any $2$-design suffices for this argument, the only difference is how efficiently can such a state be sampled directly affects the runtime of the ``learner''.
To this end, we note that uniformly random stabilizer states can be efficiently sampled in $O(n^2)$ time.

\begin{lemma}[\cite{berg2021simple}]\label{lem:clifford-sample}
    There is a classical algorithm that samples a uniformly random element of the $n$-qubit Clifford group and outputs a Clifford circuit implementation in time $O(n^2)$.
\end{lemma}

To show that we are guaranteed a state with a certain fidelity, we will need to use the following anti-concentration inequality.

\begin{lemma}[Paley-Zygmund inequality]\label{lem:paley-zygmund}
    If $X \geq 0$ is a random variable with finite variance, and if $\theta \in [0, 1]$ then
    \[
        \Pr\left[X > \theta \Ex\left[x\right]\right] \geq (1-\theta)^2\frac{\Ex\left[X\right]^2}{\Ex\left[X^2\right]}.
    \]
\end{lemma}

Let $X$ be the random variable associated with the fidelity of an arbitrary quantum state with a random stabilizer state.
By bounding the moments of $X$, we can use the \nameref{lem:paley-zygmund} to show that our learner succeeds with a certain probability.

\begin{fact}\label{fact:stabilizer-moments}
    For arbitrary $\theta \in [0, 1]$, with probability at least $(1-\theta)^2\left(\frac{1}{2} + \frac{1}{2^{n+1}}\right)$, a random stabilizer state $\ket{\phi}$ will have fidelity at least $\frac{\theta}{2^{n+1}}$ with an arbitrary quantum state $\ket{\psi}$.
\end{fact}
\begin{proof}[Proof Sketch]
    Because stabilizer states are a $3$-design \cite{kueng2015qubit}, the values of $\Ex_{\ket{\phi}\sim \mathrm{Stab}}\left[\abs{\braket{\psi | \phi}}^2\right]$ and $\Ex_{\ket{\phi}\sim \mathrm{Stab}}\left[\abs{\braket{\psi | \phi}}^4\right]$ are the same as when $\ket{\phi}$ is replaced by a Haar random state.
    By symmetry arguments, we get 
    \[
        \Ex_{\ket{\phi} \sim \mathrm{Stab}}\left[\abs{\braket{\psi | \phi}}^2\right] = \Ex_{\ket{\phi} \sim \mu_{\mathrm{Haar}}}\left[\abs{\braket{\psi | \phi}}^2\right] = \frac{1}{2^n}
    \]
    and 
    \[
        \Ex_{\ket{\phi} \sim \mathrm{Stab}}\left[\abs{\braket{\psi | \phi}}^4\right] = \Ex_{\ket{\phi} \sim \mu_{\mathrm{Haar}}}\left[\abs{\braket{\psi | \phi}}^4\right] = \frac{2}{2^n (2^n + 1)}.
    \]
    Since the absolute value of inner products of unit vectors are in the interval $[0, 1]$, by the \nameref{lem:paley-zygmund}:
    \[\Pr_{\ket{\phi} \sim \mathrm{Stab}}\left[\abs{\braket{\psi | \phi}}^2 \geq \frac{\theta}{2^n}\right] \geq (1-\theta)^2 \frac{1}{4^n} \frac{2^n (2^n+1)}{2} = (1-\theta)^2\left(\frac{1}{2} + \frac{1}{2^{n+1}}\right). \qedhere\]
\end{proof}

\begin{lemma}\label{lem:trivial-learner}
    For arbitrary $\theta \in [0, 1]$, every concept class $\calC$ can be
\[
\left(0, n^2, 1-\frac{\theta}{2^{n+1}}, (1-\theta)^2\left(\frac{1}{2} + \frac{1}{2^{n+1}}\right)\right)\text{-learned}.
\]
\end{lemma}
\begin{proof}
    In $O(n^2)$ time, we can sample a random $n$-qubit Clifford circuit using \cref{lem:clifford-sample} and then apply it to the all-zeros state to get a random stabilizer state.
    By outputting this state, \cref{fact:stabilizer-moments} ensures us that the fidelity will be at least $\frac{\theta}{2^n}$ with probability at least $(1-\theta)^2\left(\frac{1}{2} + \frac{1}{2^{n+1}}\right)$.
    Using the upper bound in \cref{fact:td-vs-fidelity}, the trace distance between the unknown quantum and our random stabilizer state is at most $\sqrt{1-\frac{\theta}{2^n}} \leq 1 - \frac{\theta}{2^{n+1}}$
\end{proof}

\end{document}